\pdfoutput=1
\documentclass[11pt,final]{article}
\usepackage[T1]{fontenc}
\usepackage[utf8]{inputenc}
\usepackage{etoolbox}
\usepackage[margin=1in]{geometry}
\usepackage{graphicx}
\usepackage{xcolor}
\usepackage{inconsolata}
\usepackage{lmodern}
\usepackage{mathtools}
\usepackage{amssymb}
\usepackage{amsthm}
\usepackage{amsfonts}
\usepackage{bm}

\usepackage[linesnumbered,ruled,vlined]{algorithm2e}
\SetKwInput{KwInput}{Input}
\SetKwInput{KwOutput}{Output}
\SetKwFor{RepTimes}{repeat}{times}{end}
\SetKwComment{Comment}{$\triangleright$\ }{}
\let\oldnl\nl
\newcommand{\nonl}{\renewcommand{\nl}{\let\nl\oldnl}}

\usepackage{hyperref}
\usepackage{url}
\usepackage[capitalize]{cleveref}
\usepackage[bf]{caption}
\usepackage{booktabs}
\usepackage{titling}
\usepackage{multicol}
\usepackage{array}
\usepackage[style=alphabetic,backref=true,maxnames=99,maxalphanames=4,natbib=true]{biblatex}

\renewbibmacro*{doi+eprint+url}{%
  \iftoggle{bbx:doi}
    {\iffieldundef{url}{\printfield{doi}}{}}
    {}%
  \newunit\newblock
  \iftoggle{bbx:eprint}
    {\usebibmacro{eprint}}
    {}%
  \newunit\newblock
  \iftoggle{bbx:url}
    {\usebibmacro{url+urldate}}
    {}%
}
\DeclareSourcemap{
  \maps[datatype=bibtex]{
    \map{
      \step[fieldset=urldate, null]
    }
  }
}
\DefineBibliographyStrings{english}{
    backrefpage = {p.},
    backrefpages = {pp.}
}
\definecolor{citegreen}{HTML}{208054}
\definecolor{citeblue}{HTML}{0055cc}

\theoremstyle{plain}
\newtheorem{theorem}{Theorem}[section]
\newtheorem{lemma}[theorem]{Lemma}
\newtheorem{corollary}[theorem]{Corollary}
\newtheorem{proposition}[theorem]{Proposition}

\newtheorem{claim}[theorem]{Claim}
\newtheorem{remark}[theorem]{Remark}

\theoremstyle{definition}
\newtheorem{definition}[theorem]{Definition}

\renewcommand{\O}{\mathcal{O}}
\newcommand{\Ot}{\widetilde{\mathcal{O}}}

\newcommand{\poly}{\mathrm{poly}}

\newcommand{\R}{\mathbb{R}}
\newcommand{\C}{\mathbb{C}}
\newcommand{\N}{\mathbb{N}}
\newcommand{\Z}{\mathbb{Z}}

\newcommand{\I}{\mathbb{I}}
\renewcommand{\Re}{\operatorname{Re}}
\renewcommand{\Im}{\operatorname{Im}}

\newcommand{\Orth}{\mathrm{O}}
\newcommand{\SO}{\mathrm{SO}}
\newcommand{\Ominus}{\Orth^{-}}

\newcommand{\U}{\mathrm{U}}
\newcommand{\Sp}{\mathrm{Sp}}

\newcommand{\Sym}{\mathcal{S}}

\DeclareMathOperator{\diag}{diag}

\DeclareMathOperator{\rank}{rank}
\DeclareMathOperator*{\E}{\mathbb{E}}

\DeclareMathOperator{\Haar}{Haar}

\DeclareMathOperator{\tr}{tr}
\DeclareMathOperator{\atantwo}{atan2}
\DeclareMathOperator{\pf}{pf}
\newcommand{\op}[2]{\ket{#1}\!\bra{#2}}
\newcommand{\ip}[2]{\langle #1 | #2 \rangle}
\newcommand{\melem}[3]{\langle #1 | #2 | #3 \rangle}
\newcommand{\ev}[2]{\melem{#2}{#1}{#2}}
\newcommand{\ket}[1]{| #1 \rangle}
\newcommand{\bra}[1]{\langle #1 |}

\newcommand{\T}{\mathrm{T}}

\renewcommand{\l}[1]{\mathopen{}\left#1}
\renewcommand{\r}[1]{\right#1\mathclose{}}

\DeclareMathOperator*{\argmin}{arg\,min}

\DeclareMathOperator{\median}{median}
\DeclareMathOperator{\dividealg}{divide}

\DeclareMathOperator{\dist}{\mathsf{dist}}
\DeclareMathOperator{\diamdist}{\mathsf{dist}_{\diamond}}
\DeclareMathOperator{\trdist}{\mathsf{dist}_{tr}}
\DeclareMathOperator{\frodist}{\mathsf{dist}_{\mathit{F}}}
\DeclareMathOperator{\phdist}{\mathsf{dist}_{ph}}
\newcommand{\act}{\Phi}
\newcommand{\pas}{\Phi_{\mathrm{pas}}}
\newcommand{\acttext}{{\mathrm{act}}}
\newcommand{\pastext}{{\mathrm{pas}}}
\newcommand{\ph}{{\mathrm{ph}}}
\newcommand{\anc}{\mathsf{a}}
\newcommand{\covmat}{\Gamma}
\newcommand{\rdm}{D}
\newcommand{\Fock}{\mathcal{F}}
\newcommand{\parity}{\mathsf{Par}}
\newcommand{\Num}{\mathsf{Num}}
\newcommand{\Circ}{\mathcal{C}}
\newcommand{\vac}[1]{0^{#1}}
\newcommand{\USp}[1]{\Orth(#1) \cap \Sp(#1, \R)}
\newcommand{\fermEPR}{\mathrm{fEPR}}

\DeclareMathOperator{\alg}{\mathcal{A}}
\newcommand{\slatertomo}{\mathsf{SlaterTomo}}
\newcommand{\gausstomo}{\mathsf{GaussianTomo}}
\newcommand{\coltomo}{\mathsf{PhaselessTomo}}
\newcommand{\colphasealg}{\mathsf{ColumnPhases}}
\newcommand{\phaseest}{\mathsf{PhaseEst}}
\newcommand{\pastomo}{\mathsf{PassiveTomo}}
\newcommand{\acttomo}{\mathsf{ActiveTomo}}
\newcommand{\bootstrap}{\mathsf{Bootstrap}}

\newcommand{\email}[1]{\href{mailto:#1}{\color{black}{\texttt{#1}}}}

\title{Learning fermionic linear optics with Heisenberg scaling and physical operations}

\author{%
Aria Christensen\thanks{Ohio State University and
Sandia National Laboratories, \email{acchris@sandia.gov}}
\and
Andrew Zhao\thanks{Sandia National Laboratories, \email{azhao@sandia.gov}}
}

\date{February 9, 2026}

\hypersetup{
colorlinks=true,
citecolor=citegreen,
linkcolor=citeblue,
linktocpage=true
}

\numberwithin{equation}{section}

\begin{document}

\maketitle

\begin{abstract}
    We revisit the problem of learning fermionic linear optics (FLO), also known as fermionic Gaussian unitaries. Given black-box query access to an unknown FLO, previous proposals required $\widetilde{\mathcal{O}}(n^5 / \varepsilon^2)$ queries, where $n$ is the system size and $\varepsilon$ is the error in diamond distance. These algorithms also use unphysical operations (i.e., violating fermionic superselection rules) and/or $n$ auxiliary modes to prepare Choi states of the FLO. In this work, we establish efficient and experimentally friendly protocols that obey superselection, use minimal ancilla (at most $1$ extra mode), and exhibit improved dependence on both parameters $n$ and $\varepsilon$. For arbitrary (active) FLOs this algorithm makes at most $\widetilde{\mathcal{O}}(n^4 / \varepsilon)$ queries, while for number-conserving (passive) FLOs we show that $\mathcal{O}(n^3 / \varepsilon)$ queries suffice. The complexity of the active case can be further reduced to $\widetilde{\mathcal{O}}(n^3 / \varepsilon)$ at the cost of using $n$ ancilla. This marks the first FLO learning algorithm that attains Heisenberg scaling in precision. As a side result, we also demonstrate an improved copy complexity of $\widetilde{\mathcal{O}}(n \eta^2 / \varepsilon^2)$ for time-efficient state tomography of $\eta$-particle Slater determinants in $\varepsilon$ trace distance, which may be of independent interest.
\end{abstract}

\tableofcontents

\section{Introduction}

Fermions are fundamental particles of matter, having half-integer spins and obeying the Pauli exclusion principle. In this work, we consider many-body systems of noninteracting, or \emph{free}, fermions. Such systems are computationally efficient to solve and they serve as invaluable models, for example in descriptions of tight-binding physics or BCS superconductivity~\cite{bruus2004many}. Algorithmically, free-fermion techniques are at the backbone of many modern electronic and nuclear structure methods, both on classical and quantum computers~\cite{ring2004nuclear,helgaker2013molecular,motta2022emerging}. The exact solvability of free fermions also has surprising non-fermionic consequences, for example, in connections to the classical Ising model~\cite{schultz1964two}, interacting quantum spin systems~\cite{lieb1961two,fendley2019free,elman2021free}, and matchgate circuits~\cite{knill2001fermionic,jozsa2008matchgates}.

While the efficient computation of free fermions has been utilized for decades, the rigorous study of their \emph{learnability} has been initiated only recently~\cite{aaronson2023efficient,ogorman2022fermionic,bittel2025optimal}. In both cases, efficiency ultimately stems from the existence of a complete, polynomial-sized description of these systems. This theme persists among many other well-known families with efficient descriptions, such as Gaussian bosons~\cite{bittel2025optimalboson,fanizza2024efficient,bittel2025energy,fanizza2025efficient}, stabilizer and near-stabilizer states~\cite{montanaro2017learning,grewal2025efficient,leone2024learning,hangleiter2024bell,chia2024efficient}, and matrix product states~\cite{cramer2010efficient}. Broadly speaking, the investigation of such families is motivated by the fact that learning generic, unstructured quantum many-body systems necessarily requires an exponential amount of resources. On the other hand, the compact description of free-fermion circuits has already enabled, among other results, a scalable benchmarking protocol for noisy quantum devices~\cite{helsen2022matchgate} and a certifiable scheme to demonstrate quantum advantage~\cite{oszmaniec2022fermion}.

The primary contribution of this paper is to give efficient algorithms for learning free-fermion unitaries. Let us establish some terminology first:~such objects go by a number of different names in the literature, for instance, fermionic \{Gaussian unitaries, linear optics, basis rotations, Bogoliubov transformations\}. The nomenclature ``matchgates'' is also common, due to their equivalence with matchgate circuits in one dimension~\cite{knill2001fermionic,jozsa2008matchgates}. To minimize confusion, for the rest of this paper we shall adopt the terminology \emph{fermionic linear optics} (FLO). This will be convenient because our results delineate between unitaries that conserve particle number, called \emph{passive} FLO, and those that generically do not (\emph{active} FLO). We treat active FLO as a superset of passive FLO, so when we refer to FLO without qualifier we typically mean active FLO unless otherwise specified.

While it is already known that FLOs can be efficiently learned, even in the strongest accuracy metric of diamond distance, prior algorithms are suboptimal in both system size and error dependence~\cite{oszmaniec2022fermion,iyer2025mildly,austin2025efficiently}. Moreover, we argue that these approaches are in some sense unnatural. First, they involve learning the compact one-body representation of the FLO in an ``entry-by-entry'' manner. In contrast, optimal unitary tomography algorithms operate by learning ``column-by-column,'' i.e., by running state tomography on $U\ket{1}, U\ket{2}$, etc.~\cite{haah2023query}. This is not only conceptually cleaner, but in fact leads to smaller error bounds (roughly speaking, the errors from learning the full state can be made isotropic, so that they do not compound egregiously as an entrywise estimate would). Second, none of these prior works achieve the gold standard of Heisenberg scaling for estimating unitary processes~\cite{giovannetti2004quantum,zwierz2010general}.  Finally, they feature resource requirements that are not strictly necessary for the task. This is either in the form of state preparation and measurements that violate fermionic superselection rules~\cite{streater2000pct}, which we will refer to as unphysical operations;~or in the use of entanglement with a large auxiliary system. The algorithms we present in this paper will address all three of these deficiencies.

\subsection{Main results}\label{sec:results}

Throughout, $n$ denotes the number of fermion modes in our system of interest. We assume a black-box access model:~one may query the unknown unitary, potentially multiple times and interleaved with controllable operations, before measuring the system. We say that one such prepare--apply--measure process is a single experimental run.
Because we only deal with pure states and unitary channels in this paper, we make the following convenient but slightly unconventional definitions of trace and diamond distances:
\begin{align}
    \trdist(\ket{\psi}, \ket{\phi}) &\coloneqq \sqrt{1 - |\ip{\psi}{\phi}|^2},\\
    \diamdist(U, V) &\coloneqq \max_{\ket{\psi}} \trdist(U\ket{\psi}, V\ket{\psi}).
\end{align}
Indeed, these coincide with their usual definitions over general mixed states and quantum channels~\cite{watrous2018theory}. We always assume that kets $\ket{\psi}$ represent normalized vectors, and $\Ot(\cdot)$ denotes an asymptotic upper bound suppressing polylogarithmic factors.

\paragraph{Learning active FLOs.} Our main contribution is the following result.

\begin{theorem}[Active FLO learner, \cref{thm:active_alg_main}]\label{thm:active_alg_intro}
    Let $\act$ be an FLO. There exists an algorithm which makes $\Ot(n^4 / \varepsilon)$ queries to $\act$ and uses $\poly(n, 1/\varepsilon)$ classical computational effort to output an efficient classical description of an FLO $\widehat{\bm{\act}}$ such that, with high probability,
    \begin{equation}
        \diamdist(\widehat{\bm{\act}}, \act) \leq \varepsilon.
    \end{equation}
    Each experiment only requires Fock (standard basis) initial states, implements $\O(n^3 / \varepsilon)$ elementary FLO gates (equivalently, two-qubit matchgates), and uses at most $1$ ancillary mode.
\end{theorem}

This algorithm improves upon the prior art in a number of ways. First and foremost, the query complexity exhibits a $1/\varepsilon$ dependence, which is the optimal Heisenberg scaling of quantum metrology. Second, we enjoy improved dependence on system size;~the best prior algorithm uses $\Ot(n^5 / \varepsilon^2)$ queries~\cite{oszmaniec2022fermion}. Third, our algorithm only uses Gaussian inputs and operations, which are furthermore physically admissible. We defer a more thorough comparison to prior work in \cref{sec:related_work}.

We remark that the ancillary mode used in this algorithm serves a singular role:~to detect a $\pm 1$ relative phase that $\act$ imparts between the even- and odd-parity sectors of the $n$-mode Fock space, $\Fock = \Fock_0 \oplus \Fock_1$, while obeying superselection rules on the extended $(n+1)$-mode system. This is a rather subtle detail, and in many physically relevant contexts we can drop the ancilla entirely.

\begin{remark}[Ancilla-free prerequisites]\label{rem:ancilla_free_intro}
    In the algorithm of \cref{thm:active_alg_intro}, if either:
    \begin{enumerate}
        \item We can prepare states of the form $\frac{\ket{0} + \ket{1}}{\sqrt{2}}$, or \label{item:anc_1}
        
        \item We only demand the output obey (with high probability) \label{item:anc_2}
    \begin{equation}
        \max_{\ket{\psi} \in \Fock_0 \cup \Fock_1} \trdist(\widehat{\bm{\act}} \ket{\psi}, \act \ket{\psi}) \leq \varepsilon,
    \end{equation}
    \end{enumerate}
    then no ancilla are required.
\end{remark}

Broadly speaking, \cref{item:anc_1} is readily available in qubit-based experiments, while \cref{item:anc_2} is an operationally meaningful metric for fermionic channels (superselection forbids preparing superpositions between $\Fock_0$ and $\Fock_1$ in the first place).

Finally, we point out that it is possible to use our techniques to design an efficient algorithm with only $\Ot(n^3 / \varepsilon)$ query complexity. However this approach requires $n$ ancilla modes, as it reduces the problem to the tomography of (fermionic) Choi states. This is essentially an improved version of the FLO steps in the algorithms of \cite{iyer2025mildly,austin2025efficiently};~see \cref{sec:cubic_choi_alg} for formal statements.

\paragraph{Learning passive FLOs.} The algorithm of \cref{thm:active_alg_intro} is structured such that it learns active and passive components of the FLO in two separate stages (we describe this in \cref{para:FLO_descrip}). In the case that $\act$ is promised to be passive we can bypass the first stage entirely, yielding an algorithm in its own right for learning passive FLOs.

\begin{theorem}[Passive FLO learner, \cref{thm:passive_alg_main}]\label{thm:passive_alg_intro}
    Let $\pas$ be a passive FLO. There exists an algorithm which makes $\O(n^3 / \varepsilon)$ queries to $\pas$ and uses $\poly(n, 1/\varepsilon)$ classical computational effort to output an efficient classical description of a passive FLO $\widehat{\bm{\act}}_{\pastext}$ such that, with high probability,
    \begin{equation}
        \diamdist(\widehat{\bm{\act}}_{\pastext}, \pas) \leq \varepsilon.
    \end{equation}
    Each experiment only requires Fock initial states, implements $\O(n^3 / \varepsilon)$ elementary FLO gates, and uses at most $1$ ancillary mode.
\end{theorem}

Just as in \cref{rem:ancilla_free_intro}, if we only care about parity-conserving inputs then we can drop the ancilla. In fact, we can go further if we wish to comport with the number symmetry of passive FLOs. That is, if we restrict to subspaces of fixed particle number, then we only need to apply passive FLO gates throughout the protocol. For fermions this space is $\wedge^\eta \C^n$, the antisymmetric subspace of $\eta$ particles.

\begin{corollary}[Passive FLO learner within number sectors]\label{cor:passive_alg_intro_number_sym}
    Let $0 \leq \eta \leq n$ be an integer. The algorithm of \cref{thm:passive_alg_intro} can be simplified to produce an output obeying the weaker guarantee,
    \begin{equation}
        \max_{\ket{\psi} \in \wedge^\eta \C^n} \trdist(\widehat{\bm{\act}}_{\pastext} \ket{\psi}, \pas \ket{\psi}) \leq \varepsilon.
    \end{equation}
    This version of the algorithm makes only $\O(n^2\eta/\varepsilon)$ queries, implements $\O(n^2\eta/\varepsilon)$ passive FLO gates per experiment, and uses no ancilla.
\end{corollary}

As an aside, this result also applies with almost no modification to learning passive \emph{bosonic} linear optics within fixed number sectors. This is because they carry a $\U(n)$-representation that obeys a stability bound identical to that of passive FLOs~\cite{arkhipov2015bosonsampling}.

\paragraph{Improved tomography of Slater determinants.} Our FLO learning algorithm relies heavily on the efficient tomography of fermionic Gaussian states. For pure Gaussian states without number symmetry, the best-known copy complexity is $\Ot(n^3 / \varepsilon^2)$ to achieve $\varepsilon$ trace-distance error~\cite{bittel2025optimal}. This is essentially sufficient for our active FLO learner, although for technical reasons we describe a variant of the protocol in \cref{sec:gaussian_state_tomo}.

In the case of passive FLOs, we mostly restrict to number-conserving Gaussian states known as Slater determinants. In this case, we desire an improved copy complexity when the particle number $\eta \ll n$. We would also like the state tomography protocol to only use number-conserving operations. The best prior result with these desiderata gave a bound of $\O(n^3 \eta^2 / \varepsilon^4)$ copies~\cite{aaronson2023efficient}, which is insufficient to achieve the query complexity of \cref{thm:passive_alg_intro}. Thus, we need to improve the copy complexity as follows.

\begin{theorem}[Slater determinant tomography, \cref{thm:slater_tomo_geq1,thm:slater_tomo_eq1}]\label{thm:slater_tomo_intro}
    Let $\ket{\psi}$ be an $n$-mode, $\eta$-particle Slater determinant. There exists an algorithm which consumes $\Ot(n \eta^2 / \varepsilon^2)$ copies of $\ket{\psi}$ and uses $\poly(n, \eta, 1/\varepsilon)$ classical computational effort to output an efficient classical description of a Slater determinant $\ket{\widehat{\bm{\psi}}}$ such that
    \begin{equation}
        \trdist(\ket{\widehat{\bm{\psi}}}, \ket{\psi}) \leq \varepsilon
    \end{equation}
    with high probability. Each experiment is a single-copy measurement of $\ket{\psi}$ and implements $\O(n^2)$ elementary passive FLO gates.
    
    In the case of $\eta = 1$, the copy complexity can be sharpened to $\O(n / \varepsilon^2)$ (i.e., without any $\log(n)$ factors).
\end{theorem}

From the generic bound $\eta \leq n$, it is clear that our algorithm performs no worse than that of \cite{bittel2025optimal} when applied to Slater determinants. Even in the setting where $\eta = \Theta(n)$, our algorithm still enjoys the advantage of using simpler number-conserving operations. Note that our passive FLO algorithm will not actually need the full strength of \cref{thm:slater_tomo_intro};~we use the special $\eta = 1$ case, simplifying the analysis and avoiding a logarithmic factor. (The active FLO algorithm uses a ``perturbative'' version applied to Gaussian states;~see \cref{sec:Un_shadows_redux} for details.)
 
\subsection{Technical overview}\label{sec:techniques}

\paragraph{\label{para:fermion_primer}Primer on fermions.} Let us first provide a brief review of fermions in second quantization, FLOs, and Gaussian states. The algebra of fermionic operators on an $n$-mode system is generated by \emph{creation and annihilation operators} $a_j^\dagger, a_j$ for $j = 1, \ldots, n$. They obey the canonical anticommutation relations (CAR) $a_j a_k + a_k a_j = 0$ and $a_j a_k^\dagger + a_k^\dagger a_j = \delta_{jk} \I$. Define the number operator $\Num = \sum_{j=1}^n a_j^\dagger a_j$, which has spectrum $\{0, 1, \ldots, n\}$. Its $0$-eigenspace is $1$-dimensional, spanned by the so-called vacuum state $\ket{\vac{n}}$. We label the vacuum by the all-zeros string to indicate that each mode is unoccupied;~the creation operator $a_j^\dagger$ places a fermion into the $j$th mode, for example $a_j^\dagger \ket{\vac{n}} = \ket{0^{j-1} \, 1 \, 0^{n-j}}$. These single-particle Fock states will be rather important in this paper, so we shall use the shorthand $\ket{1_j}$. More generally, arbitrary $k$-products of unique creation operators produce all Fock states $\ket{b}$ with $|b| = k$ particles (where $b \in \{0, 1\}^n$). These constitute the standard basis of $\wedge^k \C^n$. The entire Hilbert space for the fermions is then a direct sum of $k$-particle sectors, called \emph{Fock space}:~$\Fock = \bigoplus_{k=0}^n \wedge^k \C^n \cong (\C^2)^{\otimes n}$.

If we view arbitrary fermionic operators as (non-commutative) polynomials in the creation and annihilation operators, then FLOs are the class of unitary transformations which preserve polynomial degree. In physics language these are known as \emph{Bogoliubov transformations}. A passive FLO is a unitary operator $\pas(U)$ parametrized by a smaller unitary matrix $U \in \U(n)$ such that
\begin{equation}
    \pas(U)^\dagger a_j \pas(U) = \sum_{k=1}^n U_{jk} a_k.
\end{equation}
The map $\pas : \U(n) \to \U(\Fock)$ is a (projective) representation, which in particular satisfies the group homomorphism property $\pas(U) \pas(V) = \pas(UV)$.\footnote{Technically, projective representations only obey homomorphism up to a potential $(U, V)$-dependent phase;~but this phase is global, hence unphysical, so we abuse notation and drop it from our equations. Note that this paper will not use any particularly sophisticated representation theory.} All passive FLOs commute with the number operator;~however, the group $\U(n)$ is not the broadest possible class of fermionic Bogoliubov transformations.

It is convenient to define the \emph{Majorana operators} $\gamma_1, \ldots, \gamma_{2n}$ as
\begin{equation}
    \gamma_j = a_j + a_j^\dagger, \quad \gamma_{j+n} = -i(a_j - a_j^\dagger).
\end{equation}
In this picture, the CAR reduces to a single relation, $\gamma_p \gamma_q + \gamma_q \gamma_p = \delta_{pq} \I$. An active FLO is then a unitary $\act(Q)$ parametrized by an orthogonal matrix $Q \in \Orth(2n)$ such that
\begin{equation}
    \act(Q)^\dagger \gamma_p \act(Q) = \sum_{q=1}^{2n} Q_{pq} \gamma_q.
\end{equation}
As with the passive case, $\act : \Orth(2n) \to \U(\Fock)$ is a (projective) homomorphism:~$\act(Q) \act(R) = \act(QR)$. Note that $\Orth(2n)$ has two connected components, $\SO(2n)$ and $\Ominus(2n) = \{Q \in \Orth(2n) : \det(Q) = -1\}$. We say that an operator respects superselection if and only if it commutes with the parity operator $\parity = (-1)^\Num$. It can be checked that $[\act(Q), \parity] = 0$ if and only if $Q \in \SO(2n)$;~otherwise, they anticommute. We can represent passive FLOs in the Majorana basis as
\begin{equation}
    Q = \begin{pmatrix}
        \Re U & -\Im U\\
        \Im U & \Re U
    \end{pmatrix} \in \SO(2n) \implies \act(Q) = \pas(U).
\end{equation}
Indeed, the set of all such matrices is precisely the group $\USp{2n}$, which by the two-out-of-three property is isomorphic to $\U(n)$.

A pure \emph{fermionic Gaussian state} is any state of the form $\ket{\psi} = \act(Q) \ket{\vac{n}}$ with $Q \in \Orth(2n)$. Gaussian states are eigenstates of the parity operator:~if $Q \in \SO(2n)$ (resp.~$\Ominus(2n)$), then $\ket{\psi}$ is a superposition solely over even- (resp.~odd-)number Fock states. Gaussian states are fully characterized by a polynomial-sized object called a \emph{covariance matrix} $\covmat \in \R^{2n \times 2n}$. This is a skew-symmetric matrix with the entries
\begin{equation}
    \covmat_{pq} = - \frac{i}{2} \ev{[\gamma_p, \gamma_q]}{\psi}.
\end{equation}
The covariance matrix is orthogonal ($\covmat \covmat^\T = -\covmat^2 = \I$) if and only if $\ket{\psi}$ is pure Gaussian.

If $\ket{\psi}$ is also an eigenstate of the number operator, say $\Num\ket{\psi} = \eta\ket{\psi}$, then it further lies within a subclass of Gaussian states called $\eta$-particle \emph{Slater determinants}. Any such state can be expressed as $\ket{\psi} = \pas(U) \ket{1^\eta \, 0^{n-\eta}}$ for some $U \in \U(n)$. Slater determinants are one-to-one with a lower-dimensional object,\footnote{Under an appropriate basis change, $\rdm$ is simply a block of $\covmat$;~for Slater determinants, the other blocks are either zero or a copy of $\rdm$, so $\rdm$ is sufficient information.} the $1$-particle \emph{reduced density matrix} ($1$-RDM) $\rdm \in \C^{n \times n}$:
\begin{equation}
    \rdm_{jk} = \ev{a_j^\dagger a_k}{\psi}.
\end{equation}
$\rdm$ is a rank-$\eta$ orthogonal projector if and only if $\ket{\psi}$ is an $\eta$-particle Slater determinant. We will refer to both $\rdm$ and $\covmat$ as one-body representations of their parent Slater/Gaussian states.

\paragraph{Improved error analysis with fermionic shadows.} The foundation of our FLO unitary learning algorithm is fast state tomography. We deploy fermionic classical shadows, which comes in two variations:~one which measures in random passive FLO bases~\cite{low2022classical}, and one using random parity-conserving active FLO measurements~\cite{zhao2021fermionic,heyraud2025unified}. We refer to these as $\U(n)$-shadows and $\SO(2n)$-shadows, respectively. To learn a Gaussian state, we show in \cref{sec:gaussian_state_tomo} that $\SO(2n)$-shadows can estimate its covariance matrix $\covmat$ up to $\delta$ error in operator norm, using $\Ot(n^2 / \delta^2)$ copies of $\ket{\psi}$. Although this complexity was already known since \cite{bittel2025optimal}, we give a particularly clean proof of the statement via shadows.

To learn a Slater determinant, we make the key observation that Low's estimator for the fermionic RDM~\cite{low2022classical} can be expressed as the random matrix
\begin{equation}\label{eq:rdm_est_intro}
    \widehat{\bm{\rdm}} = \bm{V}^\dagger E(\bm{b}) \bm{V}, \quad \text{where } E(b) = (n+1) \diag(b) - |b|\I,
\end{equation}
where $\bm{V} \sim \Haar(\U(n))$ is the random FLO applied before measuring outcome $\bm{b} \in \{0, 1\}^n$. We show this in \cref{prop:low_shadow_rdm}. Then if we estimate $\rdm$ by gathering $N$ i.i.d.~copies of $\widehat{\bm{\rdm}}$, their mean tightly concentrates in a $\delta$-ball around $\rdm$ (in operator norm) once $N \gtrsim \frac{\sigma^2 \log n}{\delta^2}$. The variance parameter $\sigma^2$ here is essentially $\|{\E[\widehat{\bm{\rdm}}^2]}\|$, which we show simplifies remarkably:
\begin{equation}
    \widehat{\bm{\rdm}}^2 = \bm{V}^\dagger E(\bm{b})^2 \bm{V} = (n + 1 - 2\eta) \widehat{\bm{\rdm}} + \eta(n + 1 - \eta) \I.
\end{equation}
Higher-order contributions from $\bm{V}$ do not appear, and so we only need the second moment to evaluate $\E[\widehat{\bm{\rdm}}]$. But by construction this is $\E[\widehat{\bm{\rdm}}] = \rdm$~\cite{low2022classical}, hence with $\|\rdm\| \leq 1$ we have $\sigma^2 = \Theta(n\eta)$. In contrast, \cite{low2022classical} only provided an average-case variance over the individual RDM entries because a useful closed form for the third moment is currently unknown. This is in contrast to the ``traditional wisdom'' of classical shadows, that having control of the third moment is crucial to assess the variance~\cite{huang2020predicting,wan2023matchgate}.

When $\ket{\psi}$ is a Slater determinant, we can convert the RDM error $\delta$ to trace distance $\varepsilon$ using the sharp bound of Bittel, Mele, Eisert, and Leone~\cite{bittel2025optimal}. As we show in \cref{thm:slater_tomo_geq1}, it suffices to choose $\delta = \frac{\varepsilon}{2\sqrt{\eta}}$. For the special $\eta = 1$ case, we further recognize that \cref{eq:rdm_est_intro} coincides with the uniform POVM used by Gu{\c{t}}{\u{a}}, Kahn, Kueng, and Tropp~\cite{guta2020fast} for ordinary state tomography. Their improved concentration bound immediately applies to our setting, implying that $N \gtrsim n/\delta^2 = 4n/\varepsilon^2$ copies suffice.

\paragraph{Reducing passive FLOs to unitary tomography.} We now extend this $\eta = 1$ reduction to FLOs. Haah, Kothari, O'Donnell, and Tang~\cite{haah2023query} showed that, given access to a unitary channel $U$, one can construct an estimate $\widehat{\bm{U}}$ such that $\phdist(\widehat{\bm{U}}, U) \leq \delta$ using $\Theta(n^2/\delta)$ queries. The metric
\begin{equation}\label{eq:phdist}
    \phdist(U, V) \coloneqq \min_{\theta \in \R} \|U - e^{i\theta} V\|
\end{equation}
is the projective operator-norm distance. Although our setting does not provide direct access to $U$ as a channel, we can emulate it by restricting to single-particle inputs. That is, the state $\pas(U) \ket{1_j} \in \Fock$ is merely a lifted representation of $U \ket{j} \in \C^n$. Its $1$-RDM is $\rdm_j = U \op{j}{j} U^\dagger$, which we can efficiently learn via $\U(n)$-shadows. Because the algorithm of \cite{haah2023query} is based on repeating pure state tomography~\cite{guta2020fast} over the columns of $U$, this connection is sufficient for us to utilize their algorithm with almost no modification. Note that a naive reduction to state tomography would cost $\O(n^2/\delta^2)$ queries;~they boost to Heisenberg scaling $1/\delta$ by using a bootstrapping process, which we will elaborate on later.

\paragraph{\label{para:phase_est}Phase estimation step.} To establish \cref{cor:passive_alg_intro_number_sym}, we only need a conversion from $\phdist$ error;~\cref{claim:passive_FLO_stability} tells us that $\delta = \varepsilon/\eta$ suffices when we restrict attention to $\wedge^\eta \C^n$. However to prove the stronger diamond distance result of \cref{thm:passive_alg_intro}, we need to further estimate the $\U(1)$ phase on $U$ which the \cite{haah2023query} algorithm cannot detect. For FLOs, this is \emph{not} an unphysical global phase, as $\pas : e^{i\theta} \I \mapsto e^{i\theta\,\Num}$ maps $\U(1)$ to a nontrivial operation on $\Fock$. We accomplish this by inverting the projective estimate:~$U \widehat{\bm{U}}^\dagger = e^{i \bm{\theta}} \bm{W}$ where $\bm{W} \approx \I$. We can then estimate $\bm{\theta}$ with standard interferometry:~if we could prepare $\frac{\ket{0} + \ket{1}}{\sqrt{2}}$ in some mode and apply $\pas(e^{i \bm{\theta}} \bm{W}) \approx e^{i\bm{\theta}\,\Num}$, then the quadratures $X = \gamma_1$ and $Y = \gamma_{n+1}$ have expectations $\langle X \rangle = \cos\bm{\theta}$ and $\langle Y \rangle = \sin\bm{\theta}$. Taking $\atantwo(\langle Y \rangle, \langle X \rangle)$ recovers an estimate of $\bm{\theta}$ mod $2\pi$, up to sampling error and the closeness of $\bm{W}$ to the identity (see \cref{sec:phase_learning} for the full error analysis). This implies \cref{item:anc_1} from \cref{rem:ancilla_free_intro}.

If we are constrained by fermionic superselection then states of the form $\frac{\ket{0} + \ket{1}}{\sqrt{2}}$ are forbidden. A partial fix is to instead prepare $\frac{\ket{00} + \ket{11}}{\sqrt{2}}$, say in the first two modes;~the corresponding quadratures become $X = a_1^\dagger a_2^\dagger + a_2 a_1$ and $Y = i(a_1^\dagger a_2^\dagger - a_2 a_1)$. This imprints a phase of $e^{i 2\bm{\theta}}$, allowing us to identify $\bm{\theta}$ mod $\pi$ (not mod $2\pi$). As a consequence, the final estimate carries an ambiguous $\pm 1$ phase relative to the even- and odd-parity sectors, as seen by the fact that $e^{i(\theta + k\pi)\,\Num} = \parity^{k} e^{i\theta\,\Num}$ for integer $k$. This is the basis for \cref{item:anc_2} from \cref{rem:ancilla_free_intro}.

In order to learn $\bm{\theta}$ without ambiguity while still using parity eigenstates, we propose appending an ancillary mode $\anc$ and preparing the state $\frac{\ket{0_1 0_\anc} + \ket{1_1 1_\anc}}{\sqrt{2}}$. The same interferometric principle holds, but because $\pas(U \widehat{\bm{U}}^\dagger) \approx e^{i\bm{\theta}\,\Num}$ only acts on the system register, $e^{i\bm{\theta}\,\Num} \ket{1_1 1_\anc} = e^{i\bm{\theta}} \ket{1_1 1_\anc}$ acquires the appropriate phase. The quadratures we measure in this case are $X = a_1^\dagger a_\anc^\dagger + a_\anc a_1$ and $Y = i(a_1^\dagger a_\anc^\dagger - a_\anc a_1)$. This is the only piece of our algorithm (for both passive and active FLOs) that requires an ancilla.

Once we have a method to estimate the phase, we can incorporate it cleanly into the \cite{haah2023query} algorithm. This outputs some $\bm{U}^\sharp \in \U(n)$ such that $\|\bm{U}^\sharp - U\| \leq \delta$, using $\O(n^2/\delta)$ queries to $\pas(U)$. The stability bound of Oszmaniec, Dangniam, Morales, and Zimbor{\'a}s~\cite{oszmaniec2022fermion} (\cref{prop:FLO_stability}) then guarantees $\varepsilon$ error in diamond distance if $\delta = \varepsilon/n$.

\paragraph{\label{para:FLO_descrip}Extending to active FLOs.} Our strategy for learning active FLOs proceeds in two stages. First, we learn the covariance matrix $\covmat$ of the Gaussian state $\act(Q) \ket{\vac{n}}$ using $\SO(2n)$-shadows. Denoting the estimate by $\widehat{\bm{\covmat}}$, we can extract an orthogonal matrix $\bm{Q}_{\acttext} \in \Orth(2n)$ by computing the normal form for skew-symmetric matrices:~$\widehat{\bm{\covmat}} = \bm{Q}_{\acttext} \bm{\Lambda} \bm{Q}_{\acttext}^\T$, where
\begin{equation}
    \bm{\Lambda} = \begin{pmatrix}
        0 & \diag(\bm{\lambda})\\
        -\diag(\bm{\lambda}) & 0
    \end{pmatrix}, \quad \bm{\lambda} \in \R^n_{\geq 0}.
\end{equation}
The suggestive notation $\bm{Q}_{\acttext}$ stems from the fact that if $\act(R)$ is passive, then $\act(R) \ket{\vac{n}} \propto \ket{\vac{n}}$. Thus any orthogonal matrix extracted in this manner is ambiguous to right multiplication by $\USp{2n} \cong \U(n)$, so we only learn a ``maximally active'' component of $Q$ this way. As we show in \cref{lem:FLO_error_from_covmat}, if it holds that $\|\widehat{\bm{\covmat}} - \covmat\| \leq \delta$, then there exists some $\bm{Q}_{\pastext} \in \USp{2n}$ such that $\|\bm{Q}_{\acttext} \bm{Q}_{\pastext} - Q\| \leq \delta$.

The second stage of our algorithm then aims to learn $\bm{Q}_{\pastext}$. We query $\act(\bm{Q}_{\acttext}^\T) \act(Q) \approx \act(\bm{Q}_{\pastext})$, which is approximately passive;~thus we can (approximately) learn it using the passive FLO algorithm. Supposing that these approximations are good enough, then combining the estimates from both stages yields a suitable estimate of $Q$.

We require a more sophisticated error analysis to handle the fact that $\act(\bm{Q}_{\acttext}^\T) \act(Q)$ is only approximately passive. Write $\bm{Q}_{\acttext}^\T Q = \bm{Z} \bm{Q}_{\pastext}$ where $\bm{Z} \approx \I$, and let $\bm{U} \in \U(n)$ be the unique unitary associated to $\bm{Q}_{\pastext}$. The states used for the passive learner are now $\act(\bm{Z}) \pas(\bm{U}) \ket{1_j}$, which are no longer Slater determinants---the active FLO $\act(\bm{Z})$ behaves like a perturbation that leaks the ideal state $\pas(\bm{U}) \ket{1_j}$ into other particle-number sectors. Nonetheless, we show that we can still learn their RDMs via $\U(n)$-shadows (\cref{claim:generalized_Un_shadow}), and that they are only off by a systematic error on the order of $\|\bm{Z} - \I\|$ (\cref{cor:E_bound}).

The main challenge we encounter is controlling the variance of $\U(n)$-shadows when the state is no longer a number eigenstate. We facilitate this analysis by expressing the Bogoliubov transformation of $\act(\bm{Z})$ in the quasi-particle basis (\cref{claim:FLO_bogoliubov}), rather than the Majorana basis. This allows us to separate out the part of $\bm{Z}$ which conserves particle number from the part that mixes occupations. The former affects the systematic error but not the variance. For the latter, we show in \cref{thm:rdm_var_arb} that the variance increases to roughly $\sigma^2 = \O(n(1 + \sqrt{n} \|\bm{Z} - \I\|))$. Thus if we take the precision of the first stage (learning $\covmat$) as $\delta = \O\l(\frac{1}{\sqrt{n}}\r)$, the downstream effect is that we can recover $\sigma^2 = \O(n)$, same as in the exactly-passive setting. However this costs us an extra factor of $n$, yielding a ``base tomography'' algorithm for learning $Q$ in operator norm with $\Ot(n^3/\delta^2)$ queries.

\paragraph{Bootstrapping to diamond distance.} Methods to attain Heisenberg scaling using temporal (rather than spatial) coherence have been explored at least as early as \cite{higgins2007entanglement,higgins2009demonstrating,kimmel2015robust}. More recently, these ideas have been adapted to the multi-parameter/many-body regime~\cite{kura2018finite,huang2023learning,haah2023query,dutkiewicz2024advantage}. Naturally, we will follow the methodology of \cite{haah2023query}, which already gave the complexity for our passive FLO algorithm with only minor adjustments.

The basic premise behind their bootstrap process is to repeatedly run a ``base algorithm'' which learns with constant error, wherein each iteration adaptively updates the unitary being queried. For example, suppose we have an algorithm $\alg : (U, \delta) \mapsto \bm{V}$ such that $\dist(\bm{V}, U) \leq \delta$, where the metric $\dist$ is either the projective or non-projective distance of unitary matrices in operator norm. If we could also query $U^p$ for some integer $p \geq 1$, then the estimate $\bm{V} = \alg(U^p, \delta)$ would still obey $\dist(\bm{V}, U^p) \leq \delta$. One can show that calculating the (principal) $p$th root of $\bm{V}$ yields a matrix such that $\dist(\bm{V}^{1/p}, U) \leq \frac{\pi\delta}{p}$, provided that $\bm{V}$ and $U^p$ are within a constant-radius ball (radius $= \frac{1}{3\pi}$ suffices) of the identity.

Although we cannot simply set $\delta = \Theta(1)$ and $p = \Theta(1/\varepsilon)$, as this would violate $U^p$ being close to $\I$, we can do the next best thing:~for each iteration $t = 0, 1, \ldots, T$, run $\alg((U \widehat{\bm{U}}_{t}^\dagger)^{p_{t}}, \frac{1}{50})$ where $\widehat{\bm{U}}_{t}$ is the current best estimate of $U$ (i.e., by recursively updating $\widehat{\bm{U}}_{t} = \bm{V}^{1/p_{t-1}} \widehat{\bm{U}}_{t-1}$). The constant error $\frac{1}{50}$ assures us that the iterates $U \widehat{\bm{U}}_t^\dagger$ are sufficiently close to $\I$, and picking a logarithmically spaced schedule $p_t = 2^t$ guarantees that the final estimate $\widehat{\bm{U}}_{T+1}$ is $\varepsilon$-close to $U$ by an inductive argument, provided that $T = \lceil \log_2(1/\varepsilon) \rceil$. If each base call uses $q$ queries, then the total number over all $T + 1$ rounds is $\O(q/\varepsilon)$. This bootstrapping argument is easily applicable to the orthogonal group by viewing it as a subgroup, $\Orth(2n) \subset \U(2n)$.\footnote{The fact that $\Orth(2n)$ has a component disconnected from the identity is not an issue because if $Q, R \in \Orth(2n)$ are sufficiently close to each other then $QR^\T \in \SO(2n)$.}

To apply this to FLO learning, we simply use the fact that $\act$ and $\pas$ are homomorphisms between their one-body representations and the Fock space. Hence, taking powers on the Fock space (which is what we have physical access to) is equivalent to taking powers on the smaller representations. This is especially nice because (1)~all of our classical computation remains time-efficient, and (2)~we only need to run the base tomography with constant error throughout the entire process. This second point is noteworthy because the conversion to $\varepsilon$ diamond distance requires us to learn the one-body representation with $\varepsilon/n$ error;~however, with the bootstrap process we only need the number of iterations $T$ to depend on $\varepsilon/n$, not each iteration. As a consequence, we only pay linearly in $n$, not quadratically.

\subsection{Related work}\label{sec:related_work}

We are only aware of two prior results specifically on learning FLOs~\cite{oszmaniec2022fermion,cudby2024learning}. However, two recent works considered near-FLO unitaries~\cite{iyer2025mildly,austin2025efficiently}, and a quick re-analysis of their results in the exact-FLO limit yields two more points of comparison. We also review the existing literature on learning Gaussian states and Slater determinants.

\paragraph{Comparison to \cite{oszmaniec2022fermion}.} The algorithm of Oszmaniec, Dangniam, Morales, and Zimbor{\'a}s uses $\Ot(n^5 / \varepsilon^2)$ queries to learn active FLOs. It proceeds by preparing initial states from a family $\{\ket{\psi_q}\}_{q \in [2n]}$, applying $\act(Q)$ to each, and then measuring each Majorana operator $\{\gamma_p\}_{p \in [2n]}$. The states are defined such that $Q_{pq} = \ev{\act(Q)^\dagger\gamma_p \act(Q)}{\psi_q}$, allowing them to build an estimate of $Q$ entry-by-entry. Because all the Majorana operators anticommute, each expectation value per state must be measured one at a time, requiring $(2n)^2$ different experiments. They then prove a stability bound, which converts an $\alpha$ error on $Q$ to an $\alpha n$ error on $\act(Q)$. This eventually implies that $\Ot(n^3 / \varepsilon^2)$ samples per experiment achieves the target error.

This approach also does not obey fermionic superselection, both in the initial states $\ket{\psi_q}$ and the observables $\gamma_p$. It is easy to see that $\gamma_p$ does not commute with $\parity$ because it has odd Majorana degree. Meanwhile the states $\ket{\psi_q}$ are of the form $\frac{\ket{\vac{n}} + e^{i\theta} \ket{1_j}}{\sqrt{2}}$.

\paragraph{Comparison to \cite{cudby2024learning}.} Another FLO learning algorithm is due to Cudby and Strelchuk. Rather than diamond distance, they consider the Frobenius distance
\begin{equation}\label{eq:fro_dist}
    \frodist(U, V) \coloneqq \sqrt{1 - \frac{1}{2^n} |{\tr(U^\dagger V)}|}
\end{equation}
as their metric. Their algorithm requires $\Ot(n^{13} / \varepsilon_F^4)$ queries\footnote{Their Theorem 1 claims that $\Ot(n / \eta^2 + n^2/\eta^2)$ queries suffice to achieve $n^3 \eta$ Frobenius error, provided that $\eta \leq C/n^6$. However, the first term should actually be $n/\eta^4$ since it corresponds to learning the \emph{squared} entries of $Q$ to $\eta$ error (e.g., see \cite{huang2021information}). This dominates the total complexity since we take $\eta = \varepsilon_F/n^3$.} to get an estimate $\widehat{\bm{Q}}$ such that $\frodist(\act(\widehat{\bm{Q}}), \act(Q)) \leq \varepsilon_F$, provided that $\varepsilon_F \leq C/n^3$ for some constant $C > 0$. Like \cite{oszmaniec2022fermion}, their algorithm also learns $Q$ entrywise;~however, their approach is to perform Bell-like measurements on the Choi state of $\act(Q)$ to estimate all the magnitudes $|Q_{jk}|$ first. A second family of experiments is then performed to deduce the signs.

It is nontrivial to make a faithful comparison between average-case (Frobenius distance) versus worst-case (diamond distance) performance. Indeed, for any $U, V \in \U(2^n)$:
\begin{equation}
    \frac{1}{2} \frodist(U, V) \leq \diamdist(U, V) \leq \sqrt{2^{n-1}} \frodist(U, V).
\end{equation}
In terms of the resources required, this algorithm uses an ancilla register of $n$ modes and prepares EPR states $\frac{1}{\sqrt{2^n}} \sum_{b \in \{0,1\}^n} \ket{b} \otimes \ket{b}$ to produce the Choi states. The operations it implements also do not respect superselection rules. Finally, it queries the inverse $\act(Q)^\dagger$, which is not guaranteed to be available in black-box scenarios.

\paragraph{Learning near-Gaussian states.} A recent line of study concerns fermionic states and unitaries which are ``doped'' with a few non-Gaussian operations. For states, Mele and Herasymenko~\cite{mele2025efficient} showed that $t$-doped Gaussian states (states prepared by circuits of arbitrary FLO gates but at most $t$ non-FLO gates) are efficiently learnable as long as $t = \O(\log n)$. Central to this result is a compressibility lemma which states that any $t$-doped Gaussian state can be expressed as $\ket{\psi} = \act(Q) (\ket{\phi} \otimes \ket{0^{n-\Theta(t)}})$, for some $Q \in \Orth(2n)$ and some non-Gaussian state $\ket{\phi}$ on $\Theta(t)$ modes. The first step of their algorithm is to learn $\act(Q)$ by measuring copies of $\ket{\psi}$. Recall however that this does not uniquely determine $Q$, even in the $t = 0$ setting. Indeed, this is essentially equivalent to the first stage of our active FLO algorithm, returning only an equivalence class from the quotient space $\Orth(2n)/\U(n)$. This is sufficient for their state tomography protocol, but not for unitaries.

\paragraph{Learning near-Gaussian unitaries.} Iyer~\cite{iyer2025mildly} and Austin, Morales, and Gorshkov~\cite{austin2025efficiently} both showed that any $t$-doped FLO can be efficiently learned with respect to diamond distance, provided that $t = \O(\log n)$. Their ideas build on the compressibility lemma of \cite{mele2025efficient}. Taking $t = 0$, we can recover algorithms for learning active FLOs. In this regime, both algorithms use $\Ot(n^5 / \varepsilon^2)$ queries,\footnote{The advertised bound of \cite{iyer2025mildly} naively implies $\Ot(n^6 / \varepsilon^6)$ queries, but this can be substantially improved by relaxing their error parameter $\alpha$ from $\varepsilon^3 / n^{3/2}$ to merely $\varepsilon/n$, which is sufficient when $t = 0$.} matching that of \cite{oszmaniec2022fermion}. However they rely on the Choi-state approach to learn $Q$, thereby requiring an auxiliary register of $n$ extra modes. It is worth noting that \cite{austin2025efficiently} is the only prior work we are aware of that explicitly considers superselection rules in this learning context.

\paragraph{Prior work on learning Slater determinants.} Aaronson and Grewal first considered the efficiently learnability of Slater determinants in \cite{aaronson2023efficient}. O'Gorman~\cite{ogorman2022fermionic} showed that $\Ot(n^7 \eta^2 / \varepsilon^4)$ copies of an $\eta$-particle Slater determinant suffice (where $\varepsilon$ is the trace distance here). Aaronson and Grewal later improved this in the final version of their paper, proving a copy complexity of $\O(n^3 \eta^2 / \varepsilon^4)$~\cite{aaronson2023efficient}.

To improve the $1/\varepsilon^4$ dependence, Bittel, Mele, Eisert, and Leone~\cite{bittel2025optimal} established a sharp bound on the trace distance between Gaussian states in terms of their covariance matrices. With this, they could demonstrate an algorithm for learning pure fermionic Gaussian states using only $\Ot(n^3 / \varepsilon^2)$ single-copy measurements. As Slater determinants are a special class of Gaussian states, this result applies to them as well.

\paragraph{Information-theoretic bounds.} Zhao \emph{et al.}~\cite{hzhao2024learning} showed that, for the family of states prepared by $G$ gates, there exists an algorithm that learns any such state using $\Ot(G/\varepsilon^2)$ copies. Any Gaussian state (resp.~Slater determinant) can be prepared by a circuit with $G = \O(n^2)$ (resp.~$G = \O(n\eta)$) gates~\cite{kivlichan2018quantum,jiang2018quantum}, implying a quadratic copy complexity. This algorithm is nearly sample-optimal and uses only single-copy measurements, but is computationally inefficient, taking time exponential in $G$.

An alternative scheme was recently introduced by Walter and Witteveen~\cite{walter2025random}. They first introduce a pure Gaussian state learner, analogous to Hayashi's POVM for ordinary state tomography~\cite{hayashi1998asymptotic}. This is an entangling POVM over $\O(n^2/\varepsilon^2)$ joint copies of the Gaussian state. Then modifying the random purification channel trick~\cite{pelecanos2025mixed} for FLOs, they reduce the mixed case to the pure case with the same copy complexity. \cite{walter2025random} also prove a lower bound for this task, showing that $\Omega(n^2/\delta)$ copies are necessary to learn within $\delta$ infidelity.

As for unitary learning, \cite{hzhao2024learning} also prove upper and lower bounds for learning quantum circuits of gate complexity $G$. When taking the diamond distance, they show that the query complexity is exponential in $G$. However, under the average-case Frobenius distance (defined in \cref{eq:fro_dist}), the complexity becomes nearly linear in $G$. Similar to their state-learning algorithm, this unitary learner is also time-inefficient.

\subsection{Discussion}

In this work, we provide a learning algorithm to estimate fermionic linear optics with fewer queries than prior art. We improve in the scaling with both system size and precision;~but equally importantly, our algorithm takes a natural approach that only uses physical (i.e., parity-conserving) operations and minimal ancilla. For passive FLOs, our algorithmic design borrows strongly from the framework of \cite{haah2023query}, made query- and time-efficient via the one-body representation of FLOs. We show how to furthermore estimate the overall $\U(1)$ phase of this representation, which is physically relevant in this context.

For active FLOs, we have to handle two additional details:~(1)~an error analysis on the quotient space $\Orth(2n)/\U(n)$, and (2)~control over the effect of non-particle-conserving perturbations. This requires us to develop an analysis of fermionic classical shadows~\cite{low2022classical,zhao2021fermionic,heyraud2025unified} for estimating RDMs and covariance matrices using random matrix theory tools. In particular, we show how to reformulate the estimator of \cite{low2022classical} such that we can analyze its worst-case variance without needing to compute the third moment of the associated Haar integral. As a corollary, this allows us to establish a copy complexity for learning Slater determinants which improves upon previously known time-efficient algorithms~\cite{aaronson2023efficient,ogorman2022fermionic,bittel2025optimal}.

Some important open questions remain:

\begin{itemize}
    \item Can the query complexity for the active FLO algorithm be improved to $\O(n^3 / \varepsilon)$, matching the passive case? In \cref{sec:cubic_choi_alg} we show that this is essentially possible if we allow $n$ auxiliary modes to prepare Choi states. Can we still achieve this using at most $1$ ancilla?

    \item What is the optimal query complexity for learning FLOs? By a simple parameter counting argument plus the optimality of Heisenberg scaling, we conjecture that $\Theta(n^2/\varepsilon)$ queries are necessary and sufficient, analogous to the bounds of \cite{haah2023query} for generic unitary tomography.

    \item Can we apply the precision bootstrap to other efficient unitary learning problems to achieve Heisenberg scaling? For example, algorithms have been recently developed for near-Gaussian fermionic unitaries~\cite{iyer2025mildly,austin2025efficiently} and bosonic Gaussian unitaries~\cite{fanizza2025efficient}. Although the bootstrap is broadly applicable in principle, it seems challenging to apply it to the latter because such unitaries are represented by a \emph{noncompact} Lie group, $\Sp(2n, \R)$.
\end{itemize}

\section{Background}

\subsection{Notation}

The set of integers $\{1, \ldots, n\}$ is denoted by $[n]$. For a matrix $M$, $\|M\|$ is its operator (spectral) norm, $\|M\|_F$ its Frobenius norm, and $\|M\|_1$ its trace norm. Identity matrices are denoted by $\I$, whose dimension will be evident from context. For a vector $v \in \C^d$, $\|v\|$ is its usual $2$-norm and $\diag(v) \in \C^{d \times d}$ is the matrix with $v$ on the diagonal. Random variables are denoted by boldface symbols. Unless the base is specified, $\log(x)$ denotes the natural logarithm of $x > 0$.

\subsection{Linear algebra}

We regularly employ standard matrix factorizations such as the eigendecomposition, singular value decomposition (SVD), etc. One particularly important decomposition for Majorana covariance matrices is the normal form of a skew-symmetric matrix.

\begin{claim}[Normal form for skew-symmetric matrices]\label{claim:normal_form}
    Let $A \in \R^{2n \times 2n}$ be skew-symmetric, i.e., $A = -A^\T$. There exists an orthogonal matrix $W \in \Orth(2n)$ and nonnegative vector $\lambda \in \R^{n}_{\geq 0}$ such that
    \begin{equation}
        A = W \begin{pmatrix}
            0 & \diag(\lambda)\\
            -\diag(\lambda) & 0
        \end{pmatrix} W^\T.
    \end{equation}
    This decomposition can be computed in $\O(n^3)$ time.
\end{claim}

\begin{proof}
    The existence of this form is standard up to permutations, e.g., see \cite{horn2012matrix}. Specialized algorithms exist to compute it~\cite{bunch1982note}, although it suffices to recognize that for real skew-symmetric matrices, the real Schur decomposition coincides with this normal form (up to a permutation). 
\end{proof}

We also frequently round our estimated matrices to nearby ones with the appropriate structure (projector, unitary, etc). In operator norm, this incurs only a constant-factor amplification of the error.

\begin{claim}[Rounded matrix error]\label{claim:matrix_rounding_lemma}
    Let $A, B$ be two matrices of the conformable dimensions. Let $A = X \Sigma(A) Y^\dagger$ be the SVD of $A$ and $\Sigma(B)$ the matrix of singular values of $B$. Define $A^\star \coloneqq X \Sigma(B) Y^\dagger$. Then
    \begin{equation}
        \|A^\star - B\| \leq 2\|A - B\|.
    \end{equation}
\end{claim}

\begin{proof}
    By triangle inequality,
    \begin{equation}
        \|A^\star - B\| \leq \|A^\star - A\| + \|A - B\| = \|\Sigma(B) - \Sigma(A)\| + \|A - B\|
    \end{equation}
    where $\|\Sigma(B) - \Sigma(A)\| = \max_j |\sigma_j(A) - \sigma_j(B)|$. The claim follows from Weyl's inequality~\cite{horn2012matrix}:
    \begin{equation*}
        |\sigma_j(A) - \sigma_j(B)| \leq \|A - B\| \quad \forall j. \qedhere
    \end{equation*}
\end{proof}

\subsection{Concentration inequalities}

Concentration inequalities are invaluable tools for bounding sample complexities. We will only require two standard results.

\begin{proposition}[{Hoeffding~\cite[Theorem 2]{hoeffding1963probability}}]\label{prop:hoeffding}
    Let $\bm{x}_1, \ldots, \bm{x}_N$ be a sequence of independent, real random variables such that $|\bm{x}_\ell| \leq b$ almost surely for some fixed $b \geq 0$. Then for all $t \geq 0$,
    \begin{equation}
        \Pr\l( \l| \sum_{\ell=1}^N (\bm{x}_\ell - \E[\bm{x}_\ell]) \r| \geq t \r) \leq 2 \exp\l( \frac{-t^2/2}{Nb^2} \r).
    \end{equation}
\end{proposition}

\begin{proposition}[{Matrix Bernstein~\cite[Theorem 1.4]{tropp2012user}}]\label{prop:matrix_bernstein}
    Let $\bm{X}_1, \ldots, \bm{X}_N$ be a sequence of independent, random $n \times n$ Hermitian matrices. Suppose that each random matrix obeys
    \begin{equation}
        \E[\bm{X}_\ell] = 0 \text{ and } \|\bm{X}_\ell\| \leq B \text{ almost surely}
    \end{equation}
    for some fixed $B \geq 0$. Then for all $t \geq 0$,
    \begin{equation}\label{eq:bernstein}
        \Pr\l( \l\| \sum_{\ell=1}^N \bm{X}_\ell \r\| \geq t \r) \leq 2n \exp\l( \frac{-t^2/2}{\sigma^2 + Bt/3} \r), \text{ where } \sigma^2 \coloneqq \l\| \sum_{\ell=1}^N \E[\bm{X}_\ell^2] \r\|.
    \end{equation}
\end{proposition}

\subsection{Fermions}

A basic primer on fermions was provided in \cref{para:fermion_primer}. Here we record some more technical facts.

\begin{claim}[One-body transformations]
    For any state $\ket{\psi}$ with $1$-RDM $\rdm$ and covariance matrix $\covmat$, it holds that:
    \begin{enumerate}
        \item $U \rdm U^\dagger$ is the $1$-RDM of $\pas(U) \ket{\psi}$;
    
        \item $Q \covmat Q^\T$ is the covariance matrix of $\act(Q) \ket{\psi}$.
    \end{enumerate}
    In particular, $\diag(1^\eta \, 0^{n-\eta})$ is the $1$-RDM of $\ket{1^\eta \, 0^{n-\eta}}$ and
    \begin{equation}
        J \coloneqq \begin{pmatrix}
            0 & \I\\
            -\I & 0
        \end{pmatrix} \in \R^{2n \times 2n}
    \end{equation}
    is the covariance matrix of $\ket{\vac{n}}$.
\end{claim}

This naturally applies to any mixed state as well. Note that $J$ is the canonical symplectic form, i.e., $\Sp(2n, \R) \coloneqq \{S \in \R^{2n \times 2n} : S J S^\T = J\}$, agreeing with the fact that passive FLOs leave the vacuum invariant.

\begin{proposition}[{FLO stability bound~\cite[Lemma 6]{oszmaniec2022fermion}}]\label{prop:FLO_stability}
    Let $Q, R \in \Orth(2n)$. It holds that
    \begin{equation}\label{eq:FLO_stability}
        \diamdist(\act(Q), \act(R)) \leq n \|Q - R\|.
    \end{equation}
\end{proposition}

Note that the original statement from \cite{oszmaniec2022fermion} refers to $\SO(2n)$, but the claim is true over all of $\Orth(2n)$ as well. To see this, recall the coset $\Orth^{-}(2n) = \{R \in \Orth(2n) : \det(R) = -1\}$. First suppose $Q, R \in \Orth^{-}(2n)$. It is a standard fact that $\Orth^{-}(2n) = X \cdot \SO(2n)$ for any $X \in \Orth^{-}(2n)$. This implies that there exist $Q', R' \in \SO(2n)$ such that $Q = XQ'$, $R = XR'$. Take $X = \diag(1, -1, \ldots, -1)$ which represents the FLO $\act(X) = \gamma_1$. If \cref{eq:FLO_stability} holds for $Q', R' \in \SO(2n)$, then by unitary invariance of both the diamond and operator norms, it also holds for $Q, R \in \Orth^{-}(2n)$. Second, if instead $Q \in \SO(2n)$ and $R \in \Orth^{-}(2n)$, then the inequality holds trivially. Indeed, $\det(Q^\T R) = -1$ implies that $n\|Q - R\| = n\|\I - Q^\T R\| = 2n$. But the diamond distance is always at most $1$.

Note that for passive FLOs $\pas(U), \pas(V)$, the bound can be expressed as $n\|U - V\|$ because of the unitary similarity
\begin{equation}
    \begin{pmatrix}
        \Re A & -\Im A\\
        \Im A & \Re A
    \end{pmatrix} = \Omega^\T \begin{pmatrix}
        A^* & 0\\
        0 & A
    \end{pmatrix} \Omega^*, \quad \text{where } \Omega = \frac{1}{\sqrt{2}} \begin{pmatrix}
            \I & i\I\\
            \I & -i\I
        \end{pmatrix},
\end{equation}
for any $A \in \C^{n \times n}$. The factor of $n$ can be sharpened when we restrict to a fixed number sector, analogous to the bosonic case~\cite{arkhipov2015bosonsampling}.

\begin{claim}[Passive FLO stability bound]\label{claim:passive_FLO_stability}
    Let $U, V \in \U(n)$. It holds that
    \begin{equation}
        \max_{\ket{\psi} \in \wedge^\eta \C^n} \trdist(\pas(U) \ket{\psi}, \pas(V) \ket{\psi}) \leq \eta \phdist(U, V).
    \end{equation}
\end{claim}

\begin{proof}
    Because $\ket{\psi} \in \wedge^\eta \C^n$ is already antisymmetrized, we can write $\pas(U) \ket{\psi} = U^{\otimes \eta} \ket{\psi}$. Hence
    \begin{equation}
        \trdist(\pas(U) \ket{\psi}, \pas(V) \ket{\psi}) = \sqrt{1 - |\ev{(U^\dagger V)^{\otimes \eta}}{\psi}|^2} \leq \min_{\theta \in \R} \| (U^{\otimes \eta} - e^{i\theta} V^{\otimes \eta}) \ket{\psi} \|.
    \end{equation}
    By the max--min inequality and telescoping through tensor products,
    \begin{align}
        \max_{\ket{\psi} \in \wedge^\eta \C^n} \trdist(\pas(U) \ket{\psi}, \pas(V) \ket{\psi}) &\leq \min_{\theta \in \R} \max_{\ket{\psi} \in \wedge^\eta \C^n} \| (U^{\otimes \eta} - e^{i\theta} V^{\otimes \eta}) \ket{\psi} \| \notag\\
        &\leq \min_{\theta \in \R} \|U^{\otimes \eta} - e^{i\theta} V^{\otimes \eta}\|\\
        &\leq \min_{\theta \in \R} \eta \|U - e^{i\theta/\eta} V\| = \eta \phdist(U, V). \qedhere
    \end{align}
\end{proof}

\begin{claim}[FLO gate complexity]\label{claim:FLO_gate_complexity}
    Any FLO on $n$ modes can be decomposed into at most $\O(n^2)$ elementary two-mode gates, plus at most $1$ ``reflection gate'' $\gamma_1$. The algorithm that computes this decomposition runs in time $\O(n^3)$.
\end{claim}

\begin{proof}
    Constructive algorithms were first given by \cite{kivlichan2018quantum,jiang2018quantum}, although the fundamental idea dates back to conventional optics~\cite{reck1994experimental}. For completeness, we provide a sketch of the proof. Consider the active case, where without loss of generality we can assume $W \in \SO(2n)$ (because any $W' \in \Ominus(2n)$ can be written as $W' = XW$ with $\act(X) = \gamma_1$). We take the QR factorization of $W = QR$, where $Q$ is orthogonal and $R$ is triangular. This can be computed systematically by Givens rotation eliminations in $\O(n^3)$ time~\cite{horn2012matrix};~that is, we determine $Q = G_L \cdots G_2 G_1$ where $G_\ell$ are Givens rotations between two adjacent rows and $L \leq \binom{2n}{2}$. Additionally, since $W$ is orthogonal, so too must $R$;~this forces $R$ to be diagonal.

    Suppose the Givens rotation $G$ rotates rows $q, q+1$ by angle $\theta$. The desired transformation
    \begin{equation}
        \act(G)^\dagger \gamma_p \act(G) =
        \begin{cases}
            \cos(\theta) \gamma_q - \sin(\theta) \gamma_{q+1} & \text{if } p = q,\\
            \sin(\theta) \gamma_q + \cos(\theta) \gamma_{q+1} & \text{if } p = q + 1,\\
            \gamma_p & \text{else}.
        \end{cases}
    \end{equation}
    can be achieved with $\act(G) = e^{-\frac{\theta}{2} \gamma_p \gamma_{p+1}}$. Meanwhile since $R$ is diagonal and orthogonal, its diagonal entries can only be $\pm 1$ so $\act(R)^\dagger \gamma_p \act(R) = \pm \gamma_p$. These signs can be implemented by gates of the form $\gamma_p \gamma_q = i e^{-\frac{\pi}{2} \gamma_p \gamma_q}$, of which there are at most $n$ (see \cite[Supplementary Note 4]{zhao2024group}). Use the homomorphism $\act(W) = \act(G_L) \cdots \act(G_2) \act(G_1) \act(R)$ to conclude the gate complexity.
    
    For passive FLOs we can restrict to passive elementary gates because the transformation
    \begin{equation}
        \pas(G)^\dagger a_p \pas(G) = \begin{cases}
            \cos(\theta) a_q - \sin(\theta) a_{q+1} & \text{if } p = q,\\
            \sin(\theta) a_q + \cos(\theta) a_{q+1} & \text{if } p = q + 1,\\
            a_p & \text{else}.
        \end{cases}
    \end{equation}
    can be achieved with $\pas(G) = e^{-\theta (a_q^\dagger a_{q+1} - a_{q+1}^\dagger a_q)}$. Also note that in this $\U(n)$-representation, the diagonal transformation $R = \diag(e^{i\alpha_1}, e^{i\alpha_2}, \ldots, e^{i\alpha_n})$ is easily implemented via $e^{-i \alpha_q a_q^\dagger a_q} a_p e^{i \alpha_q a_q^\dagger a_q} = e^{i\alpha_q \delta_{pq}} a_p$.
\end{proof}

\section{Tomography of Slater determinants}\label{sec:state_tomo}

We begin with the base state tomography algorithm that will be used for the passive FLO learner. The version required for the active learner is analogous but essentially already known, so we defer its description to \cref{sec:gaussian_state_tomo}.

\subsection{Estimating the 1-RDM}

It suffices to estimate the $1$-RDM of a Slater determinant to uniquely identify it. This is standard and has been employed in prior works~\cite{aaronson2023efficient,ogorman2022fermionic,bittel2025optimal};~our contribution here is to demonstrate a variant with (1)~improved dependency on particle number, and (2)~isotropically distributed errors in the matrix. These subtle properties will later be crucial to obtain our claimed query complexity for FLO learning. To this end, we employ the classical shadows scheme of Low, designed precisely for this scenario~\cite{low2022classical}.

\begin{definition}
    Let $\rho$ be a quantum state of $\eta$ fermions on $n$ modes. We say that the \emph{$\U(n)$-shadows protocol} is the following procedure:~for each copy of $\rho$,
    \begin{enumerate}
        \item Draw a random unitary matrix $\bm{V} \sim \Haar(\U(n))$.
        \item Apply the FLO transformation $\rho \mapsto \pas(\bm{V}) \rho \pas(\bm{V})^\dagger$.
        \item Measure in the standard basis, obtaining the classical outcome $\bm{b} \in \{0, 1\}^n$ with probability $\ev{\pas(\bm{V}) \rho \pas(\bm{V})^\dagger}{\bm{b}}$.
    \end{enumerate}
    Each sample is stored as a tuple $(\bm{V}, \bm{b})$, which is an efficient classical description  of the postmeasurement state $\pas(\bm{V})^\dagger \op{\bm{b}}{\bm{b}} \pas(\bm{V})$.
\end{definition}

Supposing $\rho$ has particle number $\eta$, the outcomes $\bm{b}$ always have Hamming weight $\eta$ because $\pas(\bm{V})$ conserves particle number. Such a protocol amounts to implementing the quantum channel~\cite{huang2020predicting}
\begin{equation}
    \mathcal{M}: \rho \mapsto \E[\pas(\bm{V})^\dagger \op{\bm{b}}{\bm{b}} \pas(\bm{V})],
\end{equation}
where the expectation is taken over the draw of $\bm{V}$ and measurement outcomes $\bm{b}$. The inverse\footnote{Technically, the pseudoinverse over its image.} of this map may be determined analytically, yielding the formula $\rho = \E[\mathcal{M}^{-1}(\pas(\bm{V})^\dagger \op{\bm{b}}{\bm{b}} \pas(\bm{V}))]$. By linearity, this allows for estimating many non-commuting observables simultaneously without bias. \cite{low2022classical} derives such estimators for the $k$-RDM elements. The relevant $k = 1$ case can be summarized in a convenient, compact expression.

\begin{proposition}\label{prop:low_shadow_rdm}
    Let $\rho$ be an $n$-mode state of $\eta$ fermions and $\rdm$ its $1$-RDM. Let $(\bm{V}, \bm{b})$ be a single sample obtained by running the $\U(n)$-shadows protocol on a copy of $\rho$. Then the matrix $\widehat{\bm{\rdm}} \coloneqq \bm{V}^\dagger E(\bm{b}) \bm{V}$, where
    \begin{equation}\label{eq:1rdm_diag_term}
        E(b) \coloneqq (n + 1) \diag(b) - \eta\I,
    \end{equation}
    obeys $\E[\widehat{\bm{\rdm}}] = \rdm$.
\end{proposition}

\begin{proof}
    We unpack the notation of \cite[Theorem 5]{low2022classical}. Reducing that result to the $k = 1$ case yields the expression
    \begin{equation}
        \widehat{\bm{\rdm}}_{ij} = \melem{1_i}{\pas(\pi_{\bm{b}} \bm{V})^\dagger \mathcal{E}_\eta \pas(\pi_{\bm{b}} \bm{V})}{1_j},
    \end{equation}
    where $\pi_{b} \in \Z^{n \times n}$ is any permutation matrix that maps $\{i \in [n] : b_i = 1\}$ to $[\eta]$, and the operator $\mathcal{E}_\eta$ takes the form
    \begin{equation}
        \mathcal{E}_\eta = \sum_{k=1}^n (-1)^{s_k + 1} \binom{n - s_k}{1 - s_k} \binom{n - \eta + s_k}{s_k} \op{1_k}{1_k}, \quad s_k = \begin{cases}
            1, & k \leq \eta,\\
            0, & \text{else}.
        \end{cases}
    \end{equation}
    To get the desired expression for $\widehat{\bm{\rdm}}$, first recognize that the coefficients in $\mathcal{E}_\eta$ simplify either to $n - \eta + 1$ if $k \leq \eta$, or otherwise to $-\eta$ for the remaining $n - \eta$ terms. We can pull $\pi_{\bm{b}}$ out of $\pas(\cdot)$ and apply it to $\mathcal{E}_\eta$ to get
    \begin{equation}
    \begin{split}
        \mathcal{E}_\eta^{(\bm{b})} &\coloneqq \pas(\pi_{\bm{b}})^\dagger \mathcal{E}_\eta \pas(\pi_{\bm{b}})\\
        &= (n - \eta + 1) \sum_{k : \bm{b}_k = 1} \op{1_k}{1_k} - \eta \sum_{k : \bm{b}_k = 0} \op{1_k}{1_k}\\
        &= \sum_{k=1}^n E(\bm{b})_{kk} \op{1_k}{1_k}.
    \end{split}
    \end{equation}
    Finally, use $\melem{1_k}{\pas(\bm{V})}{1_j} = \bm{V}_{kj}$
    to arrive at
    \begin{equation}
    \begin{split}
        \widehat{\bm{\rdm}}_{ij} &= \melem{1_i}{\pas(\bm{V})^\dagger \mathcal{E}_\eta^{(\bm{b})} \pas(\bm{V})}{1_j}\\
        &= \sum_{k=1}^n \melem{1_i}{\pas(\bm{V})^\dagger}{1_k} E(\bm{b})_{kk} \melem{1_k}{\pas(\bm{V})}{1_j}\\
        &= \sum_{k=1}^n [\bm{V}^\dagger]_{ik} E(\bm{b})_{kk} \bm{V}_{kj} = \bm{V}^\dagger E(\bm{b}) \bm{V}.
    \end{split}
    \end{equation}
    The claim $\E[\widehat{\bm{\rdm}}] = \rdm$ follows from the fact that classical shadows produces unbiased estimators.
\end{proof}

Originally, \cite{low2022classical} resorted to analyzing the \emph{average} variance over all $k$-RDM entries, rather than worst-case variance bounds. This stemmed from a lack of known expression for the third moment over $\Haar(\pas(\U(n)))$. It turns out that we can obtain a nearly tight worst-case bound without using the third moment at all. Our bound is also with respect to the stronger metric of operator norm rather than max elementwise norm. Our proof uses surprisingly simple techniques and can likely be extended to arbitrary $k$-RDMs (given the appropriate generalization of $\mathcal{E}_\eta^{(\bm{b})}$).

To bound the copy complexity, we use the matrix Bernstein inequality (\cref{prop:matrix_bernstein}) in the standard way. To do so, we need to bound the range and variance of each sample.

\begin{lemma}\label{lem:rdm_bounds}
     Let $\rho$ be an $n$-mode state of $\eta$ fermions. Let $\widehat{\bm{\rdm}}_1, \ldots, \widehat{\bm{\rdm}}_N$ be i.i.d.~$\U(n)$-shadow estimates of its $1$-RDM $\rdm$, and define $\bm{X}_\ell \coloneqq \frac{1}{N}(\widehat{\bm{\rdm}}_\ell - \rdm)$. Then from \cref{prop:matrix_bernstein} we can take the parameter $B$ as
     \begin{equation}
         B = \frac{n + 1}{N},
     \end{equation}
     and $\sigma^2$ obeys
     \begin{equation}
         \sigma^2 \leq \frac{(n+1)(\eta+1) + 1}{N}
     \end{equation}
\end{lemma}

\begin{proof}
    Since each sample is identically distributed, let us temporarily drop the $\ell$ subscripts. The bound for $R$ is given by
    \begin{equation}
        \|\bm{X}\| \leq \frac{\|\widehat{\bm{\rdm}}\| + \|\rdm\|}{N} \leq \frac{n + 1}{N},
    \end{equation}
    which follows from the fact that $\|\rdm\| \leq 1$ for any state, and from \cref{eq:1rdm_diag_term} we have $\|\widehat{\bm{\rdm}}\| = \|E(\bm{b})\| = \max\{n+1-\eta, \eta\} \leq n$.
    
    For $\sigma^2$, we can compute it without needing the third moment of the Haar distribution. Observe that $\widehat{\bm{\rdm}}^2 = \bm{V}^\dagger E(\bm{b})^2 \bm{V}$. We can rewrite $E(\bm{b})^2$ as
    \begin{equation}\label{eq:E(b)_derivation}
    \begin{split}
        E(\bm{b})^2 &= (n+1)^2 \diag(\bm{b})^2 - 2\eta(n+1) \diag(\bm{b}) + \eta^2 \I\\
        &= (n+1)(n+1-2\eta) \diag(\bm{b}) + \eta^2 \I\\
        &= (n+1-2\eta)\l[ (n+1) \diag(\bm{b}) - \eta \I \r] + \eta(n+1-2\eta) \I + \eta^2 \I\\
        &= (n + 1 - 2\eta) E(\bm{b}) + \eta(n+1-\eta) \I,
    \end{split}
    \end{equation}
    where we used the fact that $\diag(\bm{b})^2 = \diag(\bm{b})$. Then summing over all $N$ independent samples, we get
    \begin{equation}
    \begin{split}
        \sum_{\ell=1}^N \E[\bm{X}_\ell^2] &= \frac{1}{N^2} \sum_{\ell=1}^N \l( \E[\widehat{\bm{\rdm}}_\ell^2] - \E[\widehat{\bm{\rdm}}_\ell]^2 \r)\\
        &= \frac{1}{N} \l[ (n+1-2\eta) \rdm + \eta(n+1-\eta) \I - \rdm^2 \r].
    \end{split}
    \end{equation}
    A crude but sufficient upper bound for the operator norm of this matrix can be found by triangle inequality and dropping all negative terms:\footnote{A tighter bound is $\max_{\lambda} |(n+1-2\eta) \lambda + \eta(n+1-\eta) - \lambda^2|$, where $\lambda \in [0, 1]$ in general and $\lambda \in \{0, 1\}$ for Slater determinants.}
    \begin{equation}\label{eq:variance_op_bound}
        \l\| (n+1-2\eta) \rdm + \eta(n+1-\eta) \I - \rdm^2 \r\| \leq (n+1) + \eta(n+1) + 1.
    \end{equation}
    The bound for $\sigma^2$ follows.
\end{proof}

We can now establish the number of copies to get small error in operator norm of the RDM. Note that we do not yet assume anything about $\rho$ besides number symmetry, so this intermediate result may be of independent interest for other contexts.

\begin{theorem}\label{thm:intermediate_1rdm}
    Let $\varepsilon, \delta \in (0, 1)$. Suppose $\rho$ is an $n$-mode state of $\eta$ fermions, and let $D_{ij} = \tr(a_i^\dagger a_j \rho)$ be its $1$-RDM. Consuming $N$ copies of $\rho$ with the $\U(n)$-shadows protocol, one can output an estimate $\overline{\bm{\rdm}} \in \C^{n \times n}$ such that
    \begin{equation}\label{eq:rdm_intermediate_tail_bound}
        \Pr\l( \|\overline{\bm{\rdm}} - \rdm\| \geq \varepsilon \r) \leq \delta,
    \end{equation}
    provided that
    \begin{equation}
        N \geq \frac{12 n\eta \log(2n/\delta)}{\varepsilon^2}.
    \end{equation}
\end{theorem}

\begin{proof}
    We adopt the notation from \cref{prop:matrix_bernstein,lem:rdm_bounds} and set $\overline{\bm{\rdm}} \coloneqq \frac{1}{N} \sum_{\ell=1}^N \widehat{\bm{\rdm}}_\ell$. Matrix Bernstein provides the upper bound on the probability in \cref{eq:rdm_intermediate_tail_bound}. Invoking \cref{lem:rdm_bounds}, which states that $B \leq \frac{n+1}{N}$ and $\sigma^2 \leq \frac{(n+1)(\eta+1) + 1}{N}$, we get the bound $N(\sigma^2 + B\varepsilon/3) \leq 6 n\eta$ for all $n, \eta \geq 1$ and $\varepsilon < 1$. Hence
    \begin{equation}
    \begin{split}
        \Pr\l( \|\overline{\bm{\rdm}} - \rdm\| \geq \varepsilon \r) &\leq 2n \exp\l( \frac{-N\varepsilon^2}{2N(\sigma^2 + B\varepsilon/3)} \r)\\
        &\leq 2n \exp\l( -\frac{N\varepsilon^2}{12 n\eta} \r).
    \end{split}
    \end{equation}
    The claim follows from demanding this be at most $\delta$.
\end{proof}

\subsection{Error bounds for Slater determinants}

Now we specialize to Slater determinants. To get from \cref{thm:intermediate_1rdm} to the desired result (\cref{thm:slater_tomo_intro}), we need to handle two remaining details:~(1)~rounding to a valid Slater RDM, and (2)~converting RDM error to trace distance. The former is handled immediately with \cref{claim:matrix_rounding_lemma}. To translate the errors, we use the sharp bound derived in \cite{bittel2025optimal}. We also need their expression for covariance matrices in terms of RDMs (when the state has number symmetry).

\begin{proposition}[{\cite[Theorem 1]{bittel2025optimal}}]\label{prop:covmat_bound}
    Let $\ket{\psi_1}, \ket{\psi_2}$ be pure fermionic Gaussian states. Then
    \begin{equation}
        \trdist(\ket{\psi_1}, \ket{\psi_2}) \leq \frac{1}{4} \|\covmat(\psi_1) - \covmat(\psi_2)\|_F,
    \end{equation}
    where $\covmat(\psi_j)$ is the covariance matrix of $\ket{\psi_j}$.
\end{proposition}

\begin{proposition}[{\cite[Lemma A10]{bittel2025optimal}}]\label{prop:rdm_to_covmat}
    Let $\rho$ be an $n$-mode state of $\eta$ fermions, $\covmat$ its covariance matrix, and $\rdm$ its $1$-RDM. The following identity holds:
    \begin{equation}
        \covmat = \begin{pmatrix}0 & \I\\ -\I & 0\end{pmatrix} + 2 \begin{pmatrix}
            \Im\rdm & -\Re\rdm\\
            \Re\rdm & \Im\rdm
        \end{pmatrix}.
    \end{equation}
\end{proposition}

Combining \cref{prop:covmat_bound,prop:rdm_to_covmat}, we get a bound on the trace distance of two Slater determinants (see also~\cite[Eq.~(A36)]{bittel2025optimal}). Importantly, converting to Frobenius norm only incurs a factor of $\sqrt{\eta}$ instead of the $\sqrt{n}$ appearing in the general Gaussian case.

\begin{lemma}\label{lem:trace-dist_op_norm}
    Let $\ket{\psi_1}, \ket{\psi_2}$ be $n$-mode, $\eta$-particle Slater determinants with $1$-RDMs $\rdm(\psi_1), \rdm(\psi_2)$. Then
    \begin{equation}
        \trdist(\ket{\psi_1}, \ket{\psi_2}) \leq \sqrt{\min\{\eta, n/2\}} \|\rdm(\psi_1) - \rdm(\psi_2)\|.
    \end{equation}
\end{lemma}

\begin{proof}
    Since the Frobenius norm is the entrywise $2$-norm,
    \begin{equation}
    \begin{split}
        \|\covmat(\psi_1) - \covmat(\psi_2)\|_F^2 &= 4 \l\| \begin{pmatrix}
            \Im[\rdm(\psi_1)] - \Im[\rdm(\psi_2)] & -(\Re[\rdm(\psi_1)] - \Re[\rdm(\psi_2)])\\
            \Re[\rdm(\psi_1)] - \Re[\rdm(\psi_2)] & \Im[\rdm(\psi_1)] - \Im[\rdm(\psi_2)]
        \end{pmatrix} \r{\|_F^2}\\
        &= 4 \l( 2\|{\Im[\rdm(\psi_1)] - \Im[\rdm(\psi_2)]}\|_F^2 + 2\|{\Re[\rdm(\psi_1)] - \Re[\rdm(\psi_2)]}\|_F^2 \r)\\
        &= 8 \|D(\psi_1) - D(\psi_2)\|_F^2.
    \end{split}
    \end{equation}
    Then use the fact that $D(\psi_j)$ is a rank-$\eta$ projector since $\ket{\psi_j}$ is a Slater determinant. By subadditivity of rank, the rank of the difference $D(\psi_1) - D(\psi_2)$ is at most $\min\{2\eta, n\}$. We conclude by chaining the estimate $\|X\|_F \leq \sqrt{\rank(X)} \|X\|$ with \cref{prop:covmat_bound}.
\end{proof}

We are now ready to prove \cref{thm:slater_tomo_intro} for the general $\eta$ case, which we rephrase below.

\begin{theorem}[\cref{thm:slater_tomo_intro}, $\eta \geq 1$ case]\label{thm:slater_tomo_geq1}
    Let $\varepsilon, \delta \in (0, 1)$. Let $\ket{\psi}$ be an $n$-mode, $\eta$-particle Slater determinant. There exists an algorithm which consumes $N = \O(n \eta^2 \log(n/\delta) / \varepsilon^2)$ copies of $\ket{\psi}$ and uses $\O(n^2 \eta^\alpha N + n^3)$ classical computational effort ($\alpha \leq 1$ is the rectangular matrix-multiplication exponent) to output an efficient classical description of a Slater determinant $\ket{\widehat{\bm{\psi}}}$ such that
    \begin{equation}
        \trdist(\ket{\widehat{\bm{\psi}}}, \ket{\psi}) \leq \varepsilon,
    \end{equation}
    with probability at least $1 - \delta$. Each measurement is implemented by $\O(n^2)$ elementary passive FLO gates.
\end{theorem}

\begin{proof}
    The algorithm is described in \cref{alg:slater}. First we verify its runtime. Recall from \cref{claim:FLO_gate_complexity} that any FLO can be implemented in $\O(n^2)$ elementary gates, and that this circuit can be determined in $\O(n^3)$ time. However the time complexity can be accelerated since we are compiling Haar-random FLO circuits, in which case recent work has shown how to draw and construct these circuits in only $\O(n^2)$ time~\cite{braccia2025optimal}. Each repetition is therefore dominated by the formation of the estimate $\overline{\bm{\rdm}}$, through the cost of matrix multiplication for $\bm{V}^\dagger \diag(\bm{b}) \bm{V}$ which reduces to multiplying an $n \times \eta$ matrix with an $\eta \times n$ matrix. The algorithm then performs one eigendecomposition at the end, an additive $\O(n^3)$ cost.

    Now we show the copy complexity. Let $\overline{\bm{\rdm}}$ be the unrounded estimate and $\bm{\rdm}^\star$ the rounded RDM, corresponding to a Slater determinant $\ket{\widehat{\bm{\psi}}}$. By \cref{claim:matrix_rounding_lemma,lem:trace-dist_op_norm}, we have
    \begin{equation}
        \trdist(\ket{\widehat{\bm{\psi}}}, \ket{\psi}) \leq \sqrt{\eta} \|\bm{\rdm}^\star - \rdm\| \leq 2 \sqrt{\eta} \|\overline{\bm{\rdm}} - \rdm\|.
    \end{equation}
    Therefore we want the spectral error on $\overline{\bm{\rdm}}$ to be at most $\frac{\varepsilon}{2\sqrt{\eta}}$. This occurs with probability $\geq 1 - \delta$ as long as we take $N = \l\lceil \frac{48 n\eta^2 \log(2n/\delta)}{\varepsilon^2} \r\rceil$, per \cref{thm:intermediate_1rdm}.
\end{proof}

\begin{algorithm}
    \caption{Learning Slater determinants:~$\slatertomo(\ket{\psi}, N)$}\label{alg:slater}
    \setcounter{AlgoLine}{0}
    \KwInput{$N$ copies of an $n$-mode state $\ket{\psi}$ and a particle number $\eta \in [n]$.}
    \KwOutput{A rank-$\eta$ projector $\bm{\rdm}^\star$ uniquely specifying a Slater determinant $\ket{\widehat{\bm{\psi}}}$.}

    \BlankLine

    Let $\bm{\rdm} \gets 0$;

    \RepTimes{$N$}{
        Draw a random $\bm{V} \sim \Haar(\U(n))$;

        Construct the FLO circuit $\pas(\bm{V})$ and apply it to $\ket{\psi}$;

        Measure in the standard basis, obtaining outcome $\bm{b} \in \{0, 1\}^n$;

        Let $E(\bm{b}) \gets (n+1) \diag(\bm{b}) - |\bm{b}| \I$;

        Set $\bm{\rdm} \gets \bm{\rdm} + \frac{1}{N} \bm{V}^\dagger E(\bm{b}) \bm{V}$;
    }

    Compute the eigendecomposition of $\bm{\rdm} = \bm{W} \bm{\Lambda} \bm{W}^\dagger$; \hfill \Comment{$\bm{\Lambda}$ in non-increasing order}

    Let $\bm{W}_\eta \in \C^{n \times \eta}$ be the first $\eta$ columns of $\bm{W}$;
    
    \Return{$\bm{\rdm}^\star \gets \bm{W}_\eta \bm{W}_\eta^\dagger$}
\end{algorithm}

\subsection{Refined analysis for single-particle states}

Now consider the $\eta = 1$ case, where we can remove the log factor in the copy complexity. We get this by reducing to the pure-state estimator of \cite{guta2020fast}. Note that the algorithm is the same as in \cref{alg:slater};~only the analysis differs.

The key realization is that when $\eta = 1$, the estimator from \cref{prop:low_shadow_rdm} becomes
\begin{equation}\label{eq:rdm_est_1-particle}
    \widehat{\bm{\rdm}} = (n+1) \bm{V}^\dagger \op{\bm{j}}{\bm{j}} \bm{V} - \I,
\end{equation}
where $\bm{j} \in [n]$ is the unique index where $\bm{b}_{\bm{j}} = 1$. This is equivalent to the uniform POVM estimator of \cite{guta2020fast}. Indeed, this is no coincidence;~our Haar-random $\U(n)$ measurements, when restricted to the single-particle subspace, precisely implement the continuous POVM $\{n\op{v}{v} \, \mathrm{d}v : \ket{v} \in \mathbb{S}^{n-1}\}$ over the complex unit sphere $\mathbb{S}^{n-1} \subset \C^n$. Then we can directly use their tail bound, which is stronger than the prior matrix Bernstein result.

\begin{proposition}[{\cite[Theorem 5]{guta2020fast}}]\label{prop:guta_thm5}
    Let $\rdm \equiv \op{u}{u} \in \C^{n \times n}$ be a rank-$1$ projector. Let $\widehat{\bm{\rdm}}_1, \ldots, \widehat{\bm{\rdm}}_N$ be i.i.d.~estimates of the form \cref{eq:rdm_est_1-particle}, obtained by measuring $\op{u}{u}$ in a uniform POVM. Set $\overline{\bm{\rdm}} \coloneqq \frac{1}{N} \sum_{\ell=1}^N \widehat{\bm{\rdm}}_\ell$. For all $t \geq 0$, it holds that
    \begin{equation}
        \Pr\l( \| \overline{\bm{\rdm}} - \op{u}{u} \| \geq t \r) \leq 2 \exp\l( 2.2n - \frac{Nt^2}{480} \r).
    \end{equation}
\end{proposition}

As a consequence, we get a claim parallel to \cite[Proposition 2.2]{haah2023query} for learning $\ket{u}$ up to a phase. We will re-derive the necessary pieces here, both to provide a self-contained presentation and also to address some subtle distinctions in our setting. The factor of $\frac{1}{\sqrt{2}}$ below is for convenience, as we will see shortly.

\begin{lemma}\label{lem:1particle_fro_dist}
    Let $\varepsilon, \delta \in (0, 1)$. Let $N \in \N^+$, $\rdm = \op{u}{u}$, and $\overline{\bm{\rdm}} \in \C^{n \times n}$ be as in \cref{prop:guta_thm5}. Take $\ket{\widehat{\bm{u}}}$ as the top eigenvector of $\overline{\bm{\rdm}}$. Then with probability at least $1 - \delta$, we have
    \begin{equation}
        \frac{1}{\sqrt{2}} \| \op{\widehat{\bm{u}}}{\widehat{\bm{u}}} - \op{u}{u} \|_F \leq \varepsilon,
    \end{equation}
    provided that
    \begin{equation}
        N \geq \frac{384(11 n + 5 \log(2/\delta))}{\varepsilon^2}.
    \end{equation}
\end{lemma}

\begin{proof}
    Use \cref{claim:matrix_rounding_lemma} to get
    \begin{equation}
        \| \op{\widehat{\bm{u}}}{\widehat{\bm{u}}} - \op{u}{u} \| \leq 2 \| \overline{\bm{\rdm}} - \op{u}{u} \|,
    \end{equation}
    hence
    \begin{equation}
        \frac{1}{\sqrt{2}} \| \op{\widehat{\bm{u}}}{\widehat{\bm{u}}} - \op{u}{u} \|_F \leq \| \op{\widehat{\bm{u}}}{\widehat{\bm{u}}} - \op{u}{u} \| \leq 2 \| \overline{\bm{\rdm}} - \op{u}{u} \|.
    \end{equation}
    Set $t = \varepsilon/2$ in \cref{prop:guta_thm5} to arrive at the claim.
\end{proof}

This implies the $\eta = 1$ case for \cref{thm:slater_tomo_intro}.

\begin{theorem}[\cref{thm:slater_tomo_intro}, $\eta = 1$ case]\label{thm:slater_tomo_eq1}
    Let $\varepsilon, \delta \in (0, 1)$. Let $\ket{\psi(u)}$ be a single-particle Slater determinant specified by a unit vector $\ket{u} \in \C^n$. There exists an algorithm which consumes $N = \O((n + \log(1/\delta)) / \varepsilon^2)$ copies of $\ket{\psi(u)}$ and uses $\O(n^2 N + n^3)$ classical computational effort to output a unit vector $\ket{\widehat{\bm{u}}} \in \C^n$ such that
    \begin{equation}
        \ket{\widehat{\bm{u}}} = e^{i\bm{\alpha}} \sqrt{1 - \bm{\varepsilon}} \ket{u} + \sqrt{\bm{\varepsilon}} \ket{\bm{w}}
    \end{equation}
    where $\bm{\alpha} \in \R$, $\bm{\varepsilon} \leq \varepsilon^2$, and
    \begin{equation}
        \trdist(\ket{\psi(\widehat{\bm{u}})}, \ket{\psi(u)}) \leq \varepsilon,
    \end{equation}
    all with probability at least $1 - \delta$.
\end{theorem}

\begin{proof}
    Because the algorithm is the same as \cref{thm:slater_tomo_geq1}, the time complexity is also the same. The only difference is that we take $N = \l\lceil \frac{384(11 n + 5 \log(2/\delta))}{\varepsilon^2} \r\rceil$ in \cref{alg:slater}. To establish correctness, chase the proof of \cref{lem:trace-dist_op_norm} to get
    \begin{equation}
        \trdist(\ket{\psi(\widehat{\bm{u}})}, \ket{\psi(u)}) \leq \frac{1}{\sqrt{2}} \| \op{\widehat{\bm{u}}}{\widehat{\bm{u}}} - \op{u}{u} \|_F = \sqrt{1 - |\ip{\widehat{\bm{u}}}{u}|^2}
    \end{equation}
    along with \cref{lem:1particle_fro_dist}. Because the uniform POVM is symmetric on $\mathbb{S}^{n-1}$, and all postprocessing is symmetric with respect to the subspace orthogonal to $\ket{u}$, the error $\ket{\bm{w}}$ is Haar random in that subspace.
\end{proof}

\begin{remark}
    The constants for this log-free case are about two orders of magnitude larger than that of \cref{thm:slater_tomo_geq1}. This slack is largely due to the $\frac{1}{480}$ in \cref{prop:guta_thm5}, which we have not attempted to optimize.\footnote{This remark applies to virtually all constants appearing in this paper.} We expect the performance to be reasonable in practice, which is typically seen in numerical simulations~\cite{flammia2012quantum,guta2020fast,huang2020predicting,bittel2025optimal}.
\end{remark}

\section{The bootstrap framework}

Here we review some technical aspects of the bootstrap algorithm, depicted in \cref{alg:bootstrap_process}. We closely follow the presentation of \cite{haah2023query}, making only a few minor modifications where appropriate.

\begin{algorithm}
    \caption{Bootstrapping to the Heisenberg limit:~$\bootstrap(\alg; \act(Q), \varepsilon, \delta)$}
    \label{alg:bootstrap_process}
    \setcounter{AlgoLine}{0}
    
    \KwInput{Query access to $\act(Q)$, error parameters $\varepsilon, \delta \in (0, 1)$, and FLO tomography algorithm $\alg : (\act(Q), \varepsilon, \delta) \mapsto \widehat{\bm{Q}}$ such that $\dist(\widehat{\bm{Q}}, Q) \leq \varepsilon$ with probability $\geq 1 - \delta$.}
    \nonl \hfill \Comment{Suppose $\alg$ uses $q \log(K/\delta) / \varepsilon^2$ queries}
    
    \KwOutput{An estimate $\widehat{\bm{Q}}$ such that $\dist(\widehat{\bm{Q}}, Q) \leq \varepsilon$ with probability $\geq 1 - \delta$.}
    \nonl \hfill \Comment{Using only $\O(q \log(K/\delta) / \varepsilon)$ queries}

    \BlankLine

    Let $\varepsilon_0 \gets \frac{1}{50}$; \hfill \Comment{Can take $\varepsilon_0 = \frac{1}{10}$ if $\dist$ is non-projective}

    Let $T \gets \lceil \log_2(1/\varepsilon)\rceil$;

    Let $\bm{V}_0 \gets \I$;

    \For{$t = 0, 1, \ldots, T$}{
        Let $p_t \gets 2^t$;

        Let $\delta_t \gets \frac{\delta}{2^{T + 1 - t}}$;

        Construct the FLO circuit $\act(\bm{V}_t^\dagger)$;

        $\bm{Q}_t \gets \alg((\act(Q) \act(\bm{V}_t^\dagger))^{p_t}, \varepsilon_0, \delta_t)$;
    
        Set $\bm{V}_{t+1} \gets \bm{Q}_t^{1/p_t} \bm{V}_t$;
    }

    \Return{$\widehat{\bm{Q}} \gets \bm{V}_{T+1}$}
\end{algorithm}

\begin{proposition}\label{prop:bootstrap_FLO}
    The output of \cref{alg:bootstrap_process} is correct for either of the following choices of $\dist$:
    \begin{enumerate}
        \item Projective:~$\dist(U, V) = \min_{|s| = 1} \|U - s V\|$,\footnote{The minimization of $s$ occurs over the same field as that of $U$ and $V$.}
        \item Non-projective:~$\dist(U, V) = \|U - V\|$,
    \end{enumerate}
    where either $U, V \in \U(n)$ (passive case) or $\Orth(2n)$ (active case). Furthermore if the base algorithm $\alg$ uses $\frac{q \log(K/\delta_t)}{\varepsilon_0^2}$ queries per iteration $t$ (for some $t$-independent parameters $q$ and $K$), then $\bootstrap(\alg)$ uses a total of $\O\l(\frac{q \log(K/\delta)}{ \varepsilon_0^2 \varepsilon}\r)$ queries. 
\end{proposition}

\begin{proof}
    This is essentially proven in \cite[Theorem 3.3]{haah2023query}. We only address some minor details\footnote{Some of these points were also discussed in \cite[Remark 3.4]{haah2023query}.} relevant to our setting. The original statement of applies directly to the passive FLO case over the projective metric. But the distance bounds they use also hold for the non-projective metric, with slightly smaller constants:
    \begin{equation}\label{eq:frac_power_nonproj}
        \|U^{1/p} - V^{1/p}\| \leq \frac{\pi}{p} \|U - V\|
    \end{equation}
    for any $p \geq 1$ and $U = e^X, V = e^Y$ such that $\|X\|, \|Y\| < \frac{1}{\pi}$~\cite[Eq.~(46)]{haah2023query}. The rest of their argument can then be used without modification.
    
    The statement also applies to active FLOs. It clearly holds over $\SO(2n)$ because this is a connected Lie subgroup of $\U(2n)$. To extend to the full orthogonal group, recall that the coset $\Ominus(2n)$ is equal to $X \cdot \SO(2n)$ for any reflection $X$. Thus for any two $U, V \in \Orth^{-}(2n)$ there exist $U', V' \in \SO(2n)$ such that $U = XU'$, $V = XV'$, and so the distance bounds apply by unitary invariance. Note that in the case that, say, $U \in \SO(2n)$ but $V \in \Ominus(2n)$, we always have $\min_{s\in\{\pm 1\}} \|U - sV\| = \|U - V\| = 2$;~taking $\varepsilon_0 \leq \frac{1}{10}$ is more than enough to avoid this happening.

    Finally we remark on some constants. Because we are not concerned with Item (b) from \cite[Theorem 3.3]{haah2023query}, one can choose a larger constant $\varepsilon_0 = \frac{1}{50}$ and exponent base for $\delta_t = \frac{\delta}{2^{T+1-t}}$ to achieve the desired final error and success probability. In fact, for the non-projective metric we can take any $\varepsilon_0 < \frac{1}{3\pi}$ because in that case, (1)~we only need to be within a ball of radius $r = \frac{1}{\pi}$ to improve the precision per iteration (guaranteed if $3\varepsilon_0 < r)$ and (2)~the tighter bound of \cref{eq:frac_power_nonproj} implies that the final error is $\frac{\pi}{p_T} \varepsilon_0 \leq \frac{1}{p_T} \leq \varepsilon$ if $\varepsilon_0 \leq \frac{1}{\pi}$ (already guaranteed by choosing $\varepsilon_0 < \frac{1}{3\pi}$).
\end{proof}

\section{Passive FLO algorithm}\label{sec:passive_alg}

In this section we describe the passive FLO algorithm. Once we establish the base case, the bootstrap procedure of \cref{alg:bootstrap_process} can be applied straightforwardly to achieve Heisenberg scaling.

\subsection{Learning in the projective metric}

We start by reviewing the algorithm of \cite{haah2023query}, which we can utilize directly when our metric is $\phdist$. Let $\pas(U)$ be the unknown FLO. The base algorithm is straightforward:~first, we use \cref{alg:slater} to learn the output states prepared from $\pas(U)$ acting on the initial states $\ket{1_j}$ for $j = 1, 2, \ldots, n$. This corresponds to the columns $\ket{u_j} = U\ket{j}$ up to some phase. Compiling each estimate $\ket{\widehat{\bm{u}}_j}$ into the columns of a matrix $\widehat{\bm{U}}$ will yield something close to $U$, except:
\begin{enumerate}
    \item The output is not guaranteed to be unitary, and
    \item The columns will be off by an unknown phase.
\end{enumerate}
To address the first issue, we can take the SVD of $\widehat{\bm{U}} = \bm{X} \bm{\Sigma} \bm{Y}^\dagger$. Rounding $\bm{\Sigma}$ to the identity yields an approximation $\bm{U}^\star = \bm{X} \bm{Y}^\dagger$, which is still close to $\widehat{\bm{U}}$ (\cref{claim:matrix_rounding_lemma}) but now guaranteed to be unitary. This process is outlined in Algorithm \ref{alg:phaseless_tomo}.

\begin{algorithm}
    \caption{Learning unitaries up to column phases:~$\coltomo(\pas(U), \varepsilon, \delta)$}\label{alg:phaseless_tomo}
    \setcounter{AlgoLine}{0}

    \KwInput{Query access to an $n$-mode FLO $\pas(U)$ and error parameters $\varepsilon, \delta \in (0, 1)$.}
    
    \KwOutput{A unitary matrix $\bm{U}^\star$ such that $\min_{\Theta \in \diag(\R^n)} \|\bm{U}^\star - U e^{i\Theta}\| \leq \varepsilon$ with probability at least $1 - \delta$.}

    \BlankLine

    $N \gets \l\lceil \frac{C (11n + 5\log(4n/\delta))}{\varepsilon^2} \r\rceil$; \hfill \Comment{$C$ is an absolute constant} \label{line:phaseless_N}

    \For{$j = 1, 2, \ldots, n$}{
        $\bm{\rdm}_j \gets \slatertomo(\pas(U) \ket{1_j}, N)$; \hfill \Comment{\cref{alg:slater}}

        Let $\ket{\widehat{\bm{u}}_j}$ be the top eigenvector of $\bm{\rdm}_j$;
    }

    Set $\widehat{\bm{U}} \gets \begin{pmatrix}
        \ket{\widehat{\bm{u}}_1} & \ket{\widehat{\bm{u}}_2} & \cdots & \ket{\widehat{\bm{u}}_n}
    \end{pmatrix}$;

    Compute the SVD of $\widehat{\bm{U}} = \bm{X} \bm{\Sigma} \bm{Y}^\dagger$;

    \Return{$\bm{U}^\star \gets \bm{X} \bm{Y}^\dagger$}
\end{algorithm}

\begin{proposition}\label{prop:phaseless_alg}
    The output of \cref{alg:phaseless_tomo} is correct.
\end{proposition}

\begin{proof}
    The proof is contained in the second half of \cite[Proof of Theorem 2.1]{haah2023query}. It is applicable to our setting because we have an equivalent form of state tomography using $\U(n)$-shadows (\cref{thm:slater_tomo_eq1}). We track down constants by a alternate version of their argument, which we defer to \cref{sec:rmt}.
\end{proof}

To address the second issue, \cite{haah2023query} introduces the following trick:~run \cref{alg:phaseless_tomo} again, but this time with the unitary $UF^\dagger$ where $F$ is the discrete Fourier transform (DFT). Let $\bm{G} = \coltomo(\pas(U) \pas(F^\dagger), \varepsilon, \delta)$ be the estimate for $UF^\dagger$. We classically compute $\bm{G}^\dagger \bm{U}^\star$, which is $\O(\varepsilon)$-close to $F$. Reading off the relevant phase for each column by comparing $\bm{G}^\dagger \bm{U}^\star$ with $F$, we deduce a correction to the original estimate $\bm{U}^\star$. This final step is somewhat technical;~we summarize it in \cref{alg:fix_col_phases}.

\begin{algorithm}
    \caption{Fixing up column phases:~$\colphasealg(V, G)$}\label{alg:fix_col_phases}
    \setcounter{AlgoLine}{0}

    \KwInput{Two unitary matrices $V, G \in \U(n)$.}
    
    \KwOutput{A unitary matrix $W \in \U(n)$.}

    \BlankLine

    $P \gets \dividealg(G^\dagger V, F)$; \hfill \Comment{Elementwise matrix division}

    Set $P_1 \gets \begin{pmatrix}
        P\ket{1} & P\ket{1} & \ldots & P\ket{1}
    \end{pmatrix} \in \C^{n \times n}$;

    $R \gets \Re[\dividealg(P, P_1)]$;

    $I \gets \Im[\dividealg(P, P_1)]$;

    \For{$j = 1, 2, \ldots, n$}{
        $x_j \gets \median\{ R_{j1}, R_{j2}, \ldots, R_{jn}\}$;

        $y_j \gets \median\{ I_{j1}, I_{j2}, \ldots, I_{jn}\}$;
        
        $\alpha_j \gets \arg(x_j + i y_j)$;
    }

    Set $\Psi \gets \diag( e^{i\alpha_1}, e^{i\alpha_2}, \ldots, e^{i\alpha_n} )$

    \Return{$W \gets V \Psi^\dagger$}
\end{algorithm}

\begin{proposition}[{\cite[Proposition 2.3]{haah2023query}}]\label{prop:col_phase_fix}
    Let $\varepsilon, \delta \in (0, 1)$, $U \in \U(n)$, and $F$ be the DFT matrix. If $\bm{V}, \bm{G}$ are unitaries such that
    \begin{equation}
        \Pr\l( \min_{\Theta \in \diag(\R^n)} \|\bm{V} - U e^{i\Theta}\| \leq \varepsilon \r) \geq 1 - \delta
    \end{equation}
    and
    \begin{equation}
        \Pr\l( \min_{\Theta \in \diag(\R^n)} \|\bm{G} - UF^\dagger e^{i\Theta}\| \leq \varepsilon \r) \geq 1 - \delta
    \end{equation}
    hold independently, then the output of $\colphasealg(\bm{V}, \bm{G})$ is a unitary $\bm{W}$ obeying
    \begin{equation}
        \Pr\l( \phdist(\bm{W}, U) \leq 25\varepsilon \r) \geq 1 - 2\delta.
    \end{equation}
\end{proposition}

\subsection{$\U(1)$ phase estimation}\label{sec:phase_learning}

The algorithm described above only returns a $\bm{U}^\star$ such that $\min_{\theta \in \R} \|e^{i\theta} \bm{U}^\star - U\|$ is small. While this is sufficient if we only restrict to number eigenstate inputs, we also need to estimate the overall $\U(1)$ phase to achieve small diamond distance. Recall that this phase is not a global phase on the FLO, but rather manifests as $\pas(e^{i\theta}\I) = e^{i\theta\,\Num}$.

In \cref{para:phase_est} we discussed the three possible options for estimating this phase, depending on the target metric and which operations we have access to. Here, we will focus only on the third option:~appending an ancillary mode to achieve a diamond distance learner. Similar interferometric principles and error analyses apply to the other two options, so the arguments we present here are readily adaptable with minimal modifications. We present the subroutine in \cref{alg:phase_est} and prove its correctness below.

\begin{algorithm}
    \caption{Learning the $\U(1)$ phase:~$\phaseest(\pas(V), N)$}\label{alg:phase_est}
    \setcounter{AlgoLine}{0}

    \KwInput{Query access to an $n$-mode FLO $\pas(V)$ and a number of queries $N$ per quadrature.}
    
    \KwOutput{A phase $\widehat{\bm{\theta}} \in (-\pi, \pi]$.}

    \BlankLine

    Append an auxiliary mode $a_{n+1}^\dagger, a_{n+1}$;
    
    \For{$O = X, Y$}{
        \Comment{See \cref{eq:quadratures,eq:quadrature_diagonalization} for the definition and eigenbasis of $X, Y$}
    
        Let $\bm{m}_O \gets 0$;
        
        \RepTimes{$N$}{
            Prepare the state $\pas(V) e^{\frac{\pi}{4} (a_1^\dagger a_{n+1}^\dagger - a_{n+1} a_1)} \ket{\vac{n+1}}$;

            Measure in the eigenbasis of $O$, obtaining outcome $\bm{b}_1, \bm{b}_{n+1} \in \{0, 1\}$;

            $\bm{m}_O \gets \bm{m}_O + \frac{1}{N} (\bm{b}_1 + \bm{b}_{n+1} - 1)$;
        }
    }

    \Return{$\widehat{\bm{\theta}} \gets \atantwo(\bm{m}_Y, \bm{m}_X)$}
\end{algorithm}

Denote the ancilla as mode $n + 1$ and define
\begin{equation}
    \ket{\Psi} \coloneqq \frac{\ket{\vac{n+1}} + \ket{1_1 1_{n+1}}}{\sqrt{2}}.
\end{equation}
This state can be prepared from the vacuum $\ket{\vac{n+1}}$ using the FLO gate $e^{\frac{\pi}{4} (a_1^\dagger a_{n+1}^\dagger - a_{n+1} a_1)}$. To expose the phase, consider the parametrization
\begin{equation}\label{eq:W_defn}
    \bm{U}^\star = e^{-i \bm{\theta}} \bm{W}^\dagger U, \quad \text{where } \bm{\theta} \coloneqq \argmin_{\varphi \in [-\pi, \pi)} \|e^{i\varphi} \bm{U}^\star - U\|
\end{equation}
so that $\bm{W} \in \U(n)$ satisfies $\|\bm{W} - \I\| = \phdist(\bm{U}^\star, U)$. Hence $U(\bm{U}^\star)^\dagger = e^{i\bm{\theta}} \bm{W}$ is close to $e^{i\bm{\theta}} \I$. Define the quadratures
\begin{equation}\label{eq:quadratures}
    X \coloneqq a_1^\dagger a_{n+1}^\dagger + a_{n+1} a_1, \quad Y \coloneqq i(a_1^\dagger a_{n+1}^\dagger - a_{n+1} a_1)
\end{equation}
and the ideal evolved state $\ket{\Psi(\bm{\theta})} \coloneqq \pas(e^{i\bm{\theta}} \I) \ket{\Psi}$. The state we actually prepare is $\ket{\widetilde{\Psi}(\bm{\theta})} \coloneqq \pas(U) \pas(\bm{U}^\star)^\dagger \ket{\Psi} = \pas(\bm{W}) \ket{\Psi(\bm{\theta})}$. Here, we write $\pas(\cdot)$ acting as usual on the first $n$ modes and trivially on the $(n+1)$st mode. The expectation values we have access to are $\ev{X}{\widetilde{\Psi}(\bm{\theta})}$ and similarly for $Y$, which simplify as follows.

\begin{claim}\label{claim:passive_perturbed_cos}
    Let $\ket{\widetilde{\Psi}(\bm{\theta})} = \pas(\bm{W}) \ket{\Psi(\bm{\theta})}$. It holds that
    \begin{subequations}
    \begin{align}
        \ev{X}{\widetilde{\Psi}(\bm{\theta})} &= |\bm{W}_{11}| \cos(\bm{\theta} + \arg(\bm{W}_{11})),\\
        \ev{Y}{\widetilde{\Psi}(\bm{\theta})} &= |\bm{W}_{11}| \sin(\bm{\theta} + \arg(\bm{W}_{11})).
    \end{align}
    \end{subequations}
\end{claim}

\begin{proof}
    Observe that $\ev{X}{\widetilde{\Psi}(\bm{\theta})} = \ev{\pas(\bm{W})^\dagger X \pas(\bm{W})}{\Psi(\bm{\theta})}$. Expand:
    \begin{equation}
        \pas(\bm{W})^\dagger a_1^\dagger a_{n+1}^\dagger \pas(\bm{W}) = \sum_{j,k=1}^{n+1} [\bm{W} \oplus 1]_{1j}^* [\bm{W} \oplus 1]_{n+1,k}^* a_j^\dagger a_k^\dagger
    \end{equation}
    and
    \begin{align}
        \ev{a_j^\dagger a_k^\dagger}{\Psi(\bm{\theta})} &= \frac{1}{2} \l( \melem{\vac{n+1}}{a_j^\dagger a_k^\dagger}{\vac{n+1}} + e^{i\bm{\theta}} \melem{\vac{n+1}}{a_j^\dagger a_k^\dagger}{1_1 1_{n+1}} \r.\\
        &\hphantom{=~} \l. + \, e^{-i\bm{\theta}} \melem{1_1 1_{n+1}}{a_j^\dagger a_k^\dagger}{\vac{n+1}} + \melem{1_1 1_{n+1}}{a_j^\dagger a_k^\dagger}{1_1 1_{n+1}} \r) \notag\\
        &= \frac{1}{2} e^{-i\bm{\theta}} (\delta_{1j} \delta_{n+1,k} - \delta_{1k} \delta_{n+1,j}),
    \end{align}
    which are the entries of the $(n+1) \times (n+1)$ matrix
    \begin{equation}
        A(\bm{\theta}) \coloneqq \frac{1}{2} e^{-i\bm{\theta}} \begin{pmatrix}
            0_{n \times n} & \ket{1}\\
            -\bra{1} & 0
        \end{pmatrix}.
    \end{equation}
    Thus
    \begin{equation}
    \begin{split}
        \ev{\pas(\bm{W})^\dagger a_1^\dagger a_{n+1}^\dagger \pas(\bm{W})}{\Psi(\bm{\theta})} &= [(\bm{W}^* \oplus 1) A(\bm{\theta}) (\bm{W}^\dagger \oplus 1)]_{1,n+1}\\
        &= \frac{1}{2} e^{-i\bm{\theta}} \l[ \begin{pmatrix}
            0_{n \times n} & \bm{W}^* \ket{1}\\
            -\bra{1} \bm{W}^\dagger & 0
        \end{pmatrix} \r{]_{1,n+1}}\\
        &= \frac{1}{2} e^{-i\bm{\theta}} \bm{W}_{11}^*.
    \end{split}
    \end{equation}
    If we denote this by $\bm{w}$, then simplifying $\langle X \rangle = \bm{w} + \bm{w}^*$ and $\langle Y \rangle = i(\bm{w} - \bm{w}^*)$ yields the claim.
\end{proof}

Thus in the infinite-sample limit, $\atantwo(\langle Y \rangle, \langle X \rangle) = \bm{\theta} + \arg(\bm{W}_{11})$ mod $2\pi$. It is easy to bound the systematic phase error $\arg(\bm{W}_{11})$ in terms of $\|\bm{W} - \I\|$. On the other hand, the statistical error from finite sampling can be bounded using standard techniques. For example, we can apply the following result originally from the context of the iterative quantum phase estimation algorithm.

\begin{proposition}[{\cite[Theorem IV.1]{van2020iterative}}]\label{prop:phase_err}
    Let $0 \leq \tau < \frac{\pi}{2}$ and $\delta \geq 0$. For any $\theta \in [-\pi, \pi)$, suppose we have two numbers $\widehat{c}$ and $\widehat{s}$ such that $|\widehat{c} - \cos(\theta)| \leq\ \delta$ and $|\widehat{s} - \sin(\theta)| \leq\ \delta$. Then using $\widehat{c}$ and $\widehat{s}$, we can compute an estimate $\widehat{\theta} \in \R$ such that $|\widehat{\theta} - \theta| \leq \tau$, provided that $\delta \leq \frac{\sin(\tau)}{\sqrt{2}}$.
\end{proposition}

Although not stated explicitly, it is easy to see geometrically that $\atantwo(\widehat{s}, \widehat{c})$ is a valid choice for $\widehat{\theta}$. With this, we can complete the error analysis for this portion of the algorithm.

\begin{theorem}\label{thm:phase_est_copies}
    Let $\ket{\widetilde{\Psi}(\bm{\theta})} = \pas(\bm{W}) \ket{\Psi(\bm{\theta})}$, where $\bm{\theta}$, $\bm{W}$ are as in \cref{eq:W_defn} with $\|\bm{W} - \I\| \leq \varepsilon < \frac{1}{2}$ holding with probability at least $1 - p > \frac{1}{2}$. Given $2N$ copies of this state, we can compute estimate some $\widehat{\bm{\theta}}$ such that
    \begin{equation}
        |e^{i\widehat{\bm{\theta}}} - e^{i\bm{\theta}}| \leq (\pi + 2) \varepsilon \quad \text{except with probability } p + 2p',
    \end{equation}
    provided that $N \geq \frac{(6 + 4\sqrt{2}) \log(2/p')}{\varepsilon^2}$.
\end{theorem}

\begin{proof}
    First we bound the statistical error. Although we can in principle evaluate the variance of $X$ and $Y$, it suffices to simply bound their ranges. The operators $X$ and $Y$ are diagonalized by Bogoliubov transformations:
    \begin{equation}\label{eq:quadrature_diagonalization}
        b_1(\varphi) \coloneqq \frac{a_1 + e^{i\varphi} a_{n+1}^\dagger}{\sqrt{2}}, \quad b_2(\varphi) \coloneqq \frac{a_{n+1} - e^{i\varphi} a_1^\dagger}{\sqrt{2}}, \quad \text{and } d_j(\varphi) \coloneqq b_j^\dagger(\varphi) b_j(\varphi),
    \end{equation}
    so that
    \begin{subequations}
    \begin{align}
        X &= d_1(0) + d_2(0) - \I,\\
        Y &= \textstyle{d_1(\frac{\pi}{2}) + d_2(\frac{\pi}{2}) - \I}.
    \end{align}
    \end{subequations}
    Thus the spectra of $X$ and $Y$ are $\{-1, 0, 1\}$. Measuring in the eigenbasis of these quadratures yields random variables with magnitudes bounded by $1$. Write $\bm{W}_{11} = \bm{r} e^{i\bm{\xi}}$ per \cref{claim:passive_perturbed_cos} and let $t \geq 0$. By Hoeffding's inequality (\cref{prop:hoeffding}), we can compute estimates $\widehat{\bm{c}}$ and $\widehat{\bm{s}}$ such that
    \begin{equation}
        |\widehat{\bm{c}} - \bm{r} \cos(\bm{\theta} + \bm{\xi})| \leq t \quad \text{and} \quad |\widehat{\bm{s}} - \bm{r} \sin(\bm{\theta} + \bm{\xi}))| \leq t
    \end{equation}
    except with probability $2p'$, provided that $N \geq \frac{2 \log(2/p')}{t^2}$.
    
    Next we combine this with the systematic error due to $\bm{W}$. By triangle inequality, the necessary cosine/sine bounds for applying \cref{prop:phase_err} are:
    \begin{align}
        |\widehat{\bm{c}} - \cos(\bm{\theta} + \bm{\xi})| &\leq |\widehat{\bm{c}} - \bm{r} \cos(\bm{\theta} + \bm{\xi})| + |\bm{r} \cos(\bm{\theta} + \bm{\xi}) - \cos(\bm{\theta} + \bm{\xi})| \notag\\
        &\leq t + |\bm{r} - 1|\\
        &\leq t + \|\bm{W} - \I\| \notag\\
        &\leq t + \varepsilon \tag{except with probability $p$},
    \end{align}
    and similarly for $|\widehat{\bm{s}} - \sin(\bm{\theta} + \bm{\xi})|$. Use the fact that, for $0 \leq \tau < \frac{\pi}{2}$, the condition $\delta \leq \frac{\sin(\tau)}{\sqrt{2}}$ always holds whenever $\delta \leq \frac{\sqrt{2}}{\pi} \tau$. Taking $\delta = t + \varepsilon$, one choice of constants that satisfies this linear inequality is $t = (\sqrt{2} - 1) \varepsilon$ and $\tau = \pi \varepsilon$. This implies that $\widehat{\bm{\theta}}$ satisfies $|\widehat{\bm{\theta}} - (\bm{\theta} + \bm{\xi})| \leq \pi \varepsilon$, so that
    \begin{equation}
    \begin{split}
        |e^{i\widehat{\bm{\theta}}} - e^{i\bm{\theta}}| &\leq |e^{i(\widehat{\bm{\theta}} - (\bm{\theta} + \bm{\xi}))} - 1| + |e^{i\bm{\xi}} - 1|\\
        &\leq |\widehat{\bm{\theta}} - (\bm{\theta} + \bm{\xi})| + |\bm{W}_{11} - 1| + |\bm{r} - 1|\\
        &\leq (\pi + 2) \varepsilon,
    \end{split}
    \end{equation}
    where we used the inequality $|e^{ix} - 1| = 2|{\sin(x/2)}| \leq |x|$.
\end{proof}

\begin{corollary}\label{cor:phase_unitary_final}
    Let $U \in \U(n)$ and $0 < \varepsilon, p < \frac{1}{2}$. Suppose we have a unitary matrix $\bm{U}^\star$ obeying
    $\phdist(\bm{U}^\star, U) \leq \varepsilon$ except with probability $p$. Using $\O(\log(1/p)/\varepsilon^2)$ queries to $\pas(U)$, we can output a unitary matrix $\bm{U}^{\sharp}$ such that
    \begin{equation}
        \|\bm{U}^{\sharp} - U\| \leq 7\varepsilon \quad \text{except with probability } 2p.
    \end{equation}
\end{corollary}

\begin{proof}
    Run the protocol of \cref{thm:phase_est_copies} with $p' = \frac{p}{2}$, where copies of $\ket{\widetilde{\Psi}(\bm{\theta})}$ are prepared via
    \begin{equation}
        \ket{\widetilde{\Psi}(\bm{\theta})} = \pas(U) \pas(\bm{U}^\star)^\dagger e^{\frac{\pi}{4} (a_1^\dagger a_{n+1}^\dagger - a_{n+1} a_1)} \ket{\vac{n+1}}.
    \end{equation} 
    From the phase estimate $\widehat{\bm{\theta}}$ construct $\bm{U}^{\sharp} \coloneqq e^{i \widehat{\bm{\theta}}} \bm{U}^\star$, which obeys
    \begin{align}
        \|\bm{U}^{\sharp} - U\| &= \|e^{i (\widehat{\bm{\theta}} - \bm{\theta})} \bm{W}^\dagger - \I\| \notag\\
        &\leq \|e^{i (\widehat{\bm{\theta}} - \bm{\theta})} \bm{W}^\dagger - \bm{W}^\dagger\| + \|\bm{W}^\dagger - \I\|\\
        &\leq (\pi + 2) \varepsilon + \varepsilon. \qedhere
    \end{align}
\end{proof}

\subsection{The complete algorithm}

Combining all of these steps, we can learn the passive FLO with the appropriate phases between columns as well as the necessary $\U(1)$ phase.

\begin{algorithm}
    \caption{Learning passive FLOs, base case:~$\pastomo(\pas(U), \varepsilon, \delta)$}\label{alg:base_passive}
    \setcounter{AlgoLine}{0}

    \KwInput{Query access to an $n$-mode FLO $\pas(U)$ and error parameters $\varepsilon, \delta \in (0, 1)$.}
    
    \KwOutput{A unitary matrix $\bm{U}^\sharp$ such that $\|\bm{U}^\sharp - U\| \leq \varepsilon$ with probability at least $1 - \delta$.}

    \BlankLine

    $\bm{V} \gets \coltomo(\pas(U), \frac{\varepsilon}{175}, \frac{\delta}{4})$; \hfill \Comment{\cref{alg:phaseless_tomo}}

    Construct the FLO circuit $\pas(F^\dagger)$; \hfill \Comment{$F$ is the DFT matrix}

    $\bm{G} \gets \coltomo(\pas(U) \pas(F^\dagger), \frac{\varepsilon}{175}, \frac{\delta}{4})$;

    $\bm{U}^\star \gets \colphasealg(\bm{V}, \bm{G})$; \hfill \Comment{\cref{alg:fix_col_phases}}

    Construct the FLO circuit $\pas(\bm{U}^\star)^\dagger$; \label{line:U_star}

    Let $N_\ph \gets \l\lceil \frac{572 \log(8/\delta)}{\varepsilon^2} \r\rceil$; \label{line:N_ph}

    $\widehat{\bm{\theta}} \gets \phaseest(\pas(U) \pas(\bm{U}^\star)^\dagger, N_\ph)$; \label{line:phase_est} \hfill \Comment{\cref{alg:phase_est}}

    \Return{$\bm{U}^\sharp \gets e^{i \widehat{\bm{\theta}}} \bm{U}^\star$}
\end{algorithm}

\begin{claim}
    The output of \cref{alg:base_passive} is correct and costs $\O(n^2 \log(1/\delta) / \varepsilon^2)$ queries.
\end{claim}

\begin{proof}
    $\bm{V}$ and $\bm{G}$ are estimated with error $\frac{\varepsilon}{175}$ except with probability at most $\frac{\delta}{4}$ each. By \cref{prop:col_phase_fix}, the phase-corrected unitary $\bm{U}^\star$ is $\frac{\varepsilon}{7}$-close to $U$ in $\phdist$ error, except with probability $\frac{\delta}{2}$. Running the phase estimation step, we can conclude the final error bound via \cref{cor:phase_unitary_final}.
\end{proof}

By the stability bound of \cref{prop:FLO_stability}, the output $\bm{U}^\sharp$ represents an $\varepsilon$-close FLO in diamond distance if we learn it to within $\varepsilon/n$ spectral distance of $U$. This proves \cref{thm:passive_alg_intro}.

\begin{theorem}[\cref{thm:passive_alg_intro}]\label{thm:passive_alg_main}
    The output of $\bootstrap(\pastomo; \pas(U), \frac{\varepsilon}{n}, \delta)$ describes an FLO which is $\varepsilon$-close to $\act(Q)$ in diamond distance, with probability at least $1 - \delta$. This algorithm costs $\O(n^3 \log(1/\delta) / \varepsilon)$ queries, $\O(n^3/\varepsilon)$ quantum gates per experiment, and $\O(n^4 \log^2(n/{\min\{\varepsilon, \delta\}}))$ classical computational time.
\end{theorem}

\begin{proof}
    As \cref{alg:base_passive} indicates, $\pastomo(\act(Q), \frac{1}{10}, \delta)$ makes $\O(n^2 \log(1/\delta))$ queries. Hence by \cref{prop:bootstrap_FLO} the bootstrapped process with error $\varepsilon/n$ makes a total of $\O(n^3 \log(1/\delta) / \varepsilon)$ queries. According to \cref{prop:FLO_stability}, $\|\bm{U}^\sharp - U\| \leq \varepsilon/n$ implies that the diamond distance between the FLOs is at most $\varepsilon$.

    The gate count per experiment is as follows. We use the fact that any FLO can be implemented in $\O(n^2)$ quantum gates (\cref{claim:FLO_gate_complexity}). Let $\bm{V}_t$ be the current estimate at each iteration $t = 0, 1, \ldots, T = \lceil \log_2(n/\varepsilon) \rceil$ and $p_t = 2^t$ (see \cref{alg:bootstrap_process}). Synthesizing the unitary $(\pas(U) \pas(\bm{V}_t^\dagger))^{p_t}$ requires $\O(p_t n^2)$ gates, which is the dominant gate complexity. Thus we use at most $\O(p_T n^2) = \O(n^3/\varepsilon)$ gates per experiment.

    We conclude with the classical cost. The circuit for $\pas(\bm{V}_t^\dagger)$ can be determined in $\O(n^3)$ time (this only needs to be calculated once per iteration). Within each call to the base tomography, we have:
    \begin{enumerate}
        \item $\coltomo$ calls $\slatertomo$ $n$ times, plus a final SVD, for a total of $\O(n(n^2 N_t + n^3) + n^3) = \O(n^4 + n^3 \log(n/\delta_t))$ operations (\cref{thm:slater_tomo_eq1});

        \item Determining the circuits for $\pas(F^\dagger)$ and $\pas(\bm{U}^\star)^\dagger$ costs $\O(n^3)$ operations (\cref{claim:FLO_gate_complexity});

        \item $\colphasealg$ costs $\O(n^2)$ operations since we only perform elementwise matrix operations (\cref{alg:fix_col_phases});

        \item $\phaseest$ costs $\O(N_{\ph,t}) = \O(\log(1/\delta_t))$ operations (\cref{alg:phase_est}).
    \end{enumerate}
    Summing over $t$ from $0$ to $T$ with $\delta_t = \frac{\delta}{2^{T+1-t}}$ yields a total classical computational complexity of
    \begin{equation}
    \begin{split}
        \O\l( \sum_{t=0}^T \l( n^4 + n^3 \log(n/\delta_t) \r) \r) = \O\l( n^4 T + n^3 (T^2 + \log(n/\delta) \r)
    \end{split}
    \end{equation}
    where $T = \O(\log(n/\varepsilon))$.
\end{proof}

If we only target small trace distance error over all states in some $\eta$-particle sector (\cref{cor:passive_alg_intro_number_sym}), we skip \cref{line:U_star,line:N_ph,line:phase_est} in \cref{alg:base_passive}, which is the only place where an ancilla mode is introduced. The inequality from \cref{claim:passive_FLO_stability} also means that we can relax the error on $U$ to $\varepsilon/\eta$. The proof of \cref{cor:passive_alg_intro_number_sym} is then completely analogous as above.

\section{Active FLO algorithm}\label{sec:active_alg}

Now we turn to active FLOs. Most of the discussion is devoted to the base tomography algorithm;~bootstrapping to Heisenberg scaling will be relatively straightforward, given the previous discussion.

\subsection{Learning from the vacuum}

Let $\act(Q)$ be the unknown FLO and $\covmat$ the covariance matrix of $\act(Q) \ket{\vac{n}}$. In \cref{sec:gaussian_state_tomo}, we show that there is an algorithm $\gausstomo$ that can learn $\covmat$ to within $\delta_{\acttext}$ error using $\Ot(n^2 / \delta_{\acttext}^2)$ copies. Recall that $\covmat = QJQ^\T$ where
\begin{equation}
    J = \begin{pmatrix}
        0 & \I\\
        -\I & 0
    \end{pmatrix}
\end{equation}
is the covariance matrix of the vacuum. Then from the normal form of an estimate to $\covmat$, we can extract an approximation of $Q$, modulo some unknown passive FLO. The following lemma shows that the error on $\covmat$ directly translates to an error on $Q$ under the appropriate quotient distance.

\begin{lemma}\label{lem:FLO_error_from_covmat}
    Let $Q_1, Q_2 \in \Orth(2n)$. Set $\covmat_1 = Q_1 J Q_1^\T$ and $\covmat_2 = Q_2 J Q_2^\T$. Then
    \begin{equation}
        \min_{R \in \USp{2n}} \|Q_1 - Q_2 R\| \leq \|\covmat_1 - \covmat_2\|.
    \end{equation}
\end{lemma}

\begin{proof}
    For convenience, we diagonalize the covariance matrices:
    \begin{equation}\label{eq:covmat_diagonalization}
        \covmat_1 = U_1 D U_1^\dagger, \quad \text{where } U_1 = Q_1 \Omega^\T, \ \Omega \coloneqq \frac{1}{\sqrt{2}} \begin{pmatrix}
            \I & i\I\\
            \I & -i\I
        \end{pmatrix}, \text{ and } D = \begin{pmatrix}
            i\I & 0\\
            0 & -i\I
        \end{pmatrix},
    \end{equation}
    and similarly $\covmat_2 = U_2 D U_2^\dagger$ with $U_2 = Q_2 \Omega^\T$. By unitary invariance of the operator norm, we have
    \begin{equation}
        \|\covmat_1 - \covmat_2\| = \|U_1^\dagger (\covmat_1 - \covmat_2) U_2\| = \|DW - WD\|
    \end{equation}
    where $W \coloneqq U_1^\dagger U_2$. Expressing
    \begin{equation}
        W = \begin{pmatrix}
            W_{11} & W_{12}\\
            W_{21} & W_{22}
        \end{pmatrix},
    \end{equation}
    we get
    \begin{equation}\label{eq:DW_commutator_bound}
    \begin{split}
        \|\covmat_1 - \covmat_2\| &= \|DW - WD\|\\
        &= \l\| \begin{pmatrix}
            W_{11} & W_{12}\\
            -W_{21} & -W_{22}
        \end{pmatrix} - \begin{pmatrix}
            W_{11} & -W_{12}\\
            W_{21} & -W_{22}
        \end{pmatrix} \r\|\\
        &= \l\| \begin{pmatrix}
            0 & 2 W_{12}\\
            -2 W_{21} & 0
        \end{pmatrix} \r\|\\
        &= 2 \max\{\|W_{12}\|, \|W_{21}\|\}.
    \end{split}
    \end{equation}
    In fact, from the structure of $W = \Omega^* (Q_1^\T Q_2) \Omega^\T$, we have $W_{21} = W_{12}^*$, hence $\|\covmat - \covmat'\| = 2 \|W_{12}\|$. This can be seen by direct calculation:
    \begin{equation}\label{eq:W_symmetry}
    \begin{split}
        W &= \Omega^* \underbrace{\begin{pmatrix}
            A & B\\
            C & D
        \end{pmatrix}}_{Q_1^\T Q_2} \Omega^\T\\
        &= \frac{1}{2} \begin{pmatrix}
            (A + D) + i(B - C) & (A - D) - i(B + C)\\
            (A - D) + i(B + C) & (A + D) - i(B - C)
        \end{pmatrix}\\
        &= \begin{pmatrix}
            W_{11} & W_{12}\\
            W_{12}^* & W_{11}^*
        \end{pmatrix}.
    \end{split}
    \end{equation}
    Furthermore, the singular values of the blocks $W_{11}, W_{12}$ are tightly related. This can be seen from the CSD of $W \in \U(2n)$, which states that there exist unitaries $X_1, X_2, Y_1, Y_2 \in \U(n)$ such that~\cite{horn2012matrix}:
    \begin{equation}\label{eq:csd_W}
        \begin{pmatrix}
            X_1 & 0\\
            0 & X_2
        \end{pmatrix}
        \begin{pmatrix}
            W_{11} & W_{12}\\
            W_{21} & W_{22}
        \end{pmatrix}
        \begin{pmatrix}
            Y_1^\dagger & 0\\
            0 & Y_2^\dagger
        \end{pmatrix} = \begin{pmatrix}
            C & S\\
            -S & C
        \end{pmatrix}.
    \end{equation}
    Here, $C$ is a diagonal matrix containing the singular values of $W_{11}$ in non-increasing order and $S = \sqrt{\I - C^2}$ contains the singular values of $W_{12}$ (in reverse order). Note that the singular values of the blocks all lie within $[0, 1]$ due to unitarity.
    
    Next we compare this to $\|Q_1 - Q_2 R\|$ for some $R \in \USp{2n}$. It is useful to express
    \begin{equation}
    \begin{split}
        \|Q_1 - Q_2 R\| &= \|\I - Q_1^\T Q_2 R\|\\
        &= \|\I - W \Omega^* R \Omega^\T\|\\
        &= \l\|\I - W \begin{pmatrix}
            V^* & 0\\
            0 & V
        \end{pmatrix} \r\|,
    \end{split}
    \end{equation}
    where the final line comes from parametrizing
    \begin{equation}
        R = \begin{pmatrix}
            \Re V & -\Im V\\
            \Im V & \Re V
        \end{pmatrix}
    \end{equation}
    for some $V \in \U(n)$ and conjugating it with $\Omega^* (\cdot) \Omega^\T$ (the calculation is analogous to \cref{eq:W_symmetry}). Take the left polar decomposition of $W_{11}$:
    \begin{equation}
        W_{11} = HZ, \text{ where } H \coloneqq \sqrt{W_{11} W_{11}^\dagger}, \ Z \in \U(n).
    \end{equation}
    The minimum over all $V \in \U(n)$ is at most the value at any particular point, so we can get an upper bound by evaluating the norm at $V = Z^\T$:
    \begin{equation}
    \begin{split}
        \min_{R \in \USp{2n}} \|Q - Q'R\| &= \min_{V \in \U(n)} \l\|\I - W \begin{pmatrix}
            V^* & 0\\
            0 & V
        \end{pmatrix} \r\|\\
        &\leq \l\|\I - W \begin{pmatrix}
            Z^\dagger & 0\\
            0 & Z^\T
        \end{pmatrix} \r\|\\
        &= \l\| \begin{pmatrix}
            \I - W_{11} Z^\dagger & W_{12} Z^\T\\
            W_{12}^* Z^\dagger & \I - W_{11}^* Z^\T
        \end{pmatrix} \r\|\\
        &= \l\| \begin{pmatrix}
            \I - H & W_{12} Z^\T\\
            W_{12}^* Z^\dagger & \I - H^*
        \end{pmatrix} \r\|.
    \end{split}
    \end{equation}
    We proceed by triangle inequality, splitting the matrix into its diagonal and off-diagonal blocks. For the diagonal blocks, since $H$ is Hermitian, $H^* = H^\T$ has the same spectrum. Thus it suffices to consider the norm of $\I - H$. By construction, the eigenvalues of $H$ are the singular values of $W_{11}$, so from the CSD we have that $0 \preceq H \preceq \I$ and hence
    \begin{equation}
    \begin{split}
        \|\I - H\| &= 1 - \sigma_{\min}(H)\\
        &= 1 - \ev{C}{n}\\
        &= 1 - \sqrt{1 - \ev{S}{1}^2}\\
        &\leq \ev{S}{1}^2 = \|W_{12}\|^2.
    \end{split}
    \end{equation}
    Meanwhile, the off-diagonal blocks clearly have operator norm $\|W_{12}\|$. Altogether, combining these facts with \cref{eq:DW_commutator_bound} we get
    \begin{equation}
    \begin{split}
        \min_{R \in \USp{2n}} \|Q_1 - Q_2 R\| &\leq \l\| \begin{pmatrix}
            \I - H & W_{12} Z\\
            W_{12}^* Z^* & \I - H^*
        \end{pmatrix} \r\|\\
        &\leq \l\| \begin{pmatrix}
            \I - H & 0\\
            0 & \I - H^*
        \end{pmatrix} \r\| + \l\| \begin{pmatrix}
            0 & W_{12} Z\\
            W_{12}^* Z^* & 0
        \end{pmatrix} \r\|\\
        &\leq \|W_{12}\|^2 + \|W_{12}\|\\
        &\leq 2 \|W_{12}\|\\
        &= \|\covmat_1 - \covmat_2\|,
    \end{split}
    \end{equation}
    which is what we had set out to prove.
\end{proof}

\begin{corollary}\label{cor:vacuum_learning_error}
    Let $\delta_{\acttext}, \eta_{\acttext} \in (0, 1)$. There is an algorithm that uses $N_{\acttext} = \l\lceil \frac{32 n^2 \log(4n/\eta_{\acttext})}{\delta_{\acttext}^2} \r\rceil$ queries to $\act(Q)$ and outputs an orthogonal matrix $\widehat{\bm{Q}}_{\acttext}$ with the following guarantee:
    \begin{equation}
        \exists Q_{\pastext} \in \USp{2n} : \|\widehat{\bm{Q}}_{\acttext} Q_{\pastext} - Q\| \leq \delta_{\acttext},
    \end{equation}
    except with probability at most $\eta_{\acttext}$.
\end{corollary}

\begin{proof}
    The algorithm is $\gausstomo(\act(Q) \ket{\vac{n}}, N_{\acttext})$ from \cref{alg:gauss} except we output the orthogonal matrix $\bm{W}$ directly. The number of copies follows from \cref{thm:intermediate_covmat,claim:matrix_rounding_lemma}, and \cref{lem:FLO_error_from_covmat} converts the covariance matrix error to orthogonal matrix error.
\end{proof}

\subsection{Learning the passive remainder}

Let $\widehat{\bm{Q}}_{\mathrm{act}}$ be the output of \cref{cor:vacuum_learning_error}. Consider the factorization $Q = \bm{Q}_{\acttext} \bm{Q}_{\pastext}$ where
\begin{equation}\label{eq:Qpas_defn}
    \bm{Q}_{\pastext} \coloneqq \argmin_{R \in \USp{2n}} \|\widehat{\bm{Q}}_{\mathrm{act}} R - Q\|.
\end{equation}
Such a factorization always exists, for example as a consequence of the Bloch--Messiah decomposition~\cite{bloch1962canonical}. Then, defining $\bm{Z} \coloneqq \widehat{\bm{Q}}_{\acttext}^\T \bm{Q}_{\acttext}$, it follows that
\begin{equation}
    \act(\widehat{\bm{Q}}_{\acttext}^\T) \act(Q) = \act(\bm{Z}) \act(\bm{Q}_{\pastext})
\end{equation}
with $\|\bm{Z} - \I\| \leq \delta_{\acttext}$ except with probability $\eta_{\acttext}$. For the remainder of this section we will condition on this high probability event, so we drop the boldface type on those quantities.

Let $U \in \U(n)$ such that $\act(Q_{\pastext}) = \pas(U)$. If it were the case that $Z = \I$, we could have directly applied the passive FLO analysis using the states $\pas(U) \ket{1_j}$ without modification. Unfortunately, $\act(Z)$ is a non-trivial active perturbation, meaning it induces leakage into different particle-number sectors. Hence our prior error analysis from \cref{sec:passive_alg} does not entirely hold. Our high-level goal will be to determine what $\delta_{\acttext}$ suffices to control this symmetry-breaking perturbation.

To begin, let us define the states which will serve as inputs to the passive FLO algorithm:
\begin{equation}\label{eq:active_perturbed_probe_states}
    \ket{\psi_j} \coloneqq \act(\widehat{Q}_{\acttext}^\T) \act(Q) \ket{1_j} = \act(Z) \pas(U) \ket{1_j}.
\end{equation}
Denote their $1$-RDMs as $\rdm_j$. Because $\ket{\psi_j}$ is no longer Slater, $\rdm_j$ is not an exact rank-$1$ projector onto the $j$th column of $U$. Nonetheless, we will show that it is $\delta_{\acttext}$-close to $\op{u_j}{u_j}$. In order to prove this, we first introduce a convenient mapping from covariance matrices to RDMs. This mapping is applicable to any quantum state, in contrast to the result of \cref{prop:rdm_to_covmat} which is only valid for number-conserving states.

\begin{claim}\label{claim:covmat_to_rdm}
    Let $\rho$ be an $n$-mode state with $1$-RDM $\rdm$ and covariance matrix $\covmat$. Then $\rdm = \frac{1}{2} (\I + T(\covmat))$, where $T : \R^{2n \times 2n} \to \C^{n \times n}$ is the linear mapping
    \begin{equation}
        T : \begin{pmatrix}
            A & B\\
            C & D
        \end{pmatrix} \mapsto \frac{B - C + i(A + D)}{2}.
    \end{equation}
    Furthermore, the induced spectral norm of $T$ is
    \begin{equation}
        \|T\|_{\infty\to\infty} \coloneqq \sup_{\|X\| \neq 0} \frac{\|T(X)\|}{\|X\|} = 1.
    \end{equation}
\end{claim}

\begin{proof}
    Recall the relation between annihilation/creation operators and Majorana operators:
    \begin{equation}
        a_j = \frac{\gamma_j - i\gamma_{j+n}}{2}, \quad a_j^\dagger = \frac{\gamma_j + i\gamma_{j+n}}{2}.
    \end{equation}
    Also recall that the covariance matrix entries are
    \begin{equation}
        \covmat_{j,k} = -i \tr(\gamma_j \gamma_k \rho) + i\delta_{jk}.
    \end{equation}
    Thus we can write the $1$-RDM as
    \begin{equation}
    \begin{split}
        \rdm_{jk} &= \tr(a_j^\dagger a_k \rho)\\
        &= \frac{1}{4} \l( \tr(\gamma_j \gamma_k \rho) + \tr(\gamma_{j+n} \gamma_{k+n} \rho) + i\tr(\gamma_{j+n} \gamma_k \rho) - i\tr(\gamma_j \gamma_{k+n} \rho) \r)\\
        &= \frac{1}{4} \l( 2\delta_{jk} + i\covmat_{j,k} + i\covmat_{j+n,k+n} - \covmat_{j+n,k} + \covmat_{j,k+n} \r).
    \end{split}
    \end{equation}
    Equivalently in block matrix form, if we write
    \begin{equation}
        \covmat = \begin{pmatrix}
            \covmat_{11} & \covmat_{12}\\
            \covmat_{21} & \covmat_{22}
        \end{pmatrix}
    \end{equation}
    with each block an $n \times n$ real matrix, then
    \begin{equation}
        \rdm = \frac{\I + \frac{1}{2}(\covmat_{12} - \covmat_{21}) + \frac{i}{2}(\covmat_{11} + \covmat_{22})}{2} = \frac{\I + T(\covmat)}{2}.
    \end{equation}
    
    To compute the induced norm of $T$, observe that we can write
    \begin{equation}
        T(X) = i P X P^\dagger, \quad \text{where } P \coloneqq \frac{1}{\sqrt{2}} \begin{pmatrix}
            \I & i\I
        \end{pmatrix} \in \C^{n \times 2n}.
    \end{equation}
    Since $PP^\dagger = \I$, $\|P\| = 1$. Hence $\|T(X)\| \leq \|P\|^2 \|X\| = \|X\|$. For the reverse direction, consider $X = \diag(\I, \I)$. Then $T(X) = i\I$ attains $\|T(X)\| = \|X\|$.
\end{proof}

With this, we can bound the distance between the RDMs of any two states which are FLO-rotated from some common initial state. Again, this result holds for generic initial states.

\begin{lemma}\label{lem:rdm_from_covmat_rotation}
    Fix an $n$-mode state $\sigma$. Let $Q_1, Q_2 \in \Orth(2n)$ and define $\rho_1 \coloneqq \act(Q_1) \sigma \act(Q_1)^\dagger$, $\rho_2 \coloneqq \act(Q_2) \sigma \act(Q_2)^\dagger$. Let $\rdm(\rho_1), \rdm(\rho_2)$ be their respective $1$-RDMs. Then
    \begin{equation}
        \|\rdm(\rho_1) - \rdm(\rho_2)\| \leq \|Q_1 - Q_2\|.
    \end{equation}
\end{lemma}

\begin{proof}
    Let $\covmat(\sigma)$ be the covariance matrix of $\sigma$. By \cref{claim:covmat_to_rdm}, we have $\rdm(\rho_1) = \frac{1}{2}(\I + T(Q_1 \covmat(\sigma) Q_1^\T))$ and $\rdm(\rho_2) = \frac{1}{2}(\I + T(Q_2 \covmat(\sigma) Q_2^\T))$. Set $W \coloneqq Q_1^\T Q_2$. Using linearity of $T$ and the fact that $\|T\|_{\infty\to\infty} = 1$, we get
    \begin{align}
        \|\rdm(\rho) - \rdm(\rho')\| &= \frac{1}{2} \|T(Q_1 \covmat(\sigma) Q_1^\T - Q_2 \covmat(\sigma) Q_2^\T)\| \notag\\
        &\leq \frac{1}{2} \|T\|_{\infty\to\infty} \|Q_1 \covmat(\sigma) Q_1^\T - Q_2 \covmat(\sigma) Q_2^\T\| \notag\\
        &= \frac{1}{2} \|\covmat(\sigma) - W \covmat(\sigma) W^\T\|\\
        &= \frac{1}{2} \|\covmat(\sigma)(\I - W^\T) + (\I - W) \covmat(\sigma) W^\T\| \notag\\
        &\leq \|\I - W\| = \|Q_1 - Q_2\|. \qedhere
    \end{align}
\end{proof}

As a special case, we get the desired bound between $\rdm_j$ and $\op{u_j}{u_j}$.

\begin{corollary}\label{cor:E_bound}
    Let $Z \in \SO(2n)$ and $U \in \U(n)$. For the states $\ket{\psi_j} = \act(Z) \pas(U) \ket{1_j}$ it holds that
    \begin{equation}
        \| \rdm_j - \op{u_j}{u_j} \| \leq \|Z - \I\|,
    \end{equation}
    where $\rdm_j$ is the $1$-RDM of $\ket{\psi_j}$ and $\ket{u_j} = U \ket{j}$.
\end{corollary}

\begin{proof}
    Apply \cref{lem:rdm_from_covmat_rotation} with $\sigma = \pas(U) \op{1_j}{1_j} \pas(U)^\dagger$, $Q_1 = Z$, and $Q_2 = \I$.
\end{proof}

At this point, we can execute the error analysis for learning $U$ (aka, $Q_{\pastext}$). For each $j \in [n]$, let $\overline{\bm{\rdm}}_j$ be the averaged estimate for $\rdm_j$. Suppose we have the uniform guarantee:
\begin{equation}\label{eq:perturbed_rdm_shadows_guarantee}
    \| \overline{\bm{\rdm}}_j - \rdm_j \| \leq \delta_{\pastext} \quad \text{except with probability } \eta_{\pastext},
\end{equation}
for some $\delta_{\pastext}, 
\eta_{\pastext} \in (0, 1)$. It remains to determine the sufficient number $N_{\pastext}$ of copies of $\ket{\psi_j}$ to achieve this;~we will address this aspect later in \cref{sec:Un_shadows_redux} (see \cref{cor:stage2_copy_complexity}). For now, the intuition is that as long as $\delta_{\acttext}$ is small enough, the copy complexity compared to the passive case is nearly unchanged.

Supposing that the guarantee holds, the analogue of \cref{prop:phaseless_alg} is as follows.

\begin{lemma}\label{lem:phaseless_act_learner}
    Let $\overline{\bm{\rdm}}_1, \ldots, \overline{\bm{\rdm}}_n$ be as in \cref{eq:perturbed_rdm_shadows_guarantee}. Take their top eigenvectors $\ket{\widehat{\bm{u}}_j}$ and concatenate them into the columns of a matrix $\widehat{\bm{U}} \in \C^{n \times n}$. If we construct the unitary matrix $\bm{U}^\star \coloneqq \bm{X} \bm{Y}^\dagger$, where $\widehat{\bm{U}} = \bm{X} \bm{\Sigma} \bm{Y}^\dagger$ is the SVD of $\widehat{\bm{U}}$, then except with probability at most $n \eta_{\pastext}$,
    \begin{equation}\label{eq:phaseless_error_bound_active}
        \min_{\Theta \in \diag(\R^n)} \|\bm{U}^\star - U e^{i\Theta}\| \leq 4\sqrt{2n}(\delta_{\pastext} + \delta_{\acttext}),
    \end{equation}
    where $\delta_{\acttext}$ is the error from \cref{cor:vacuum_learning_error}.
\end{lemma}

\begin{proof}
    We condition on the event in \cref{eq:perturbed_rdm_shadows_guarantee}, which occurs for all $j \in [n]$ except with probability at most $n \eta_{\pastext}$ by a union bound. Use \cref{claim:matrix_rounding_lemma} to assert that
    \begin{equation}
    \begin{split}
        \|\op{\widehat{\bm{u}}_j}{\widehat{\bm{u}}_j} - \op{u_j}{u_j}\| &\leq 2\|\overline{\bm{\rdm}}_j - \op{u_j}{u_j}\|\\
        &\leq 2(\|\overline{\bm{\rdm}}_j - \rdm_j\| + \|\rdm_j - \op{u_j}{u_j}\|)\\
        &\leq 2(\delta_{\pastext} + \delta_{\acttext}).
    \end{split}
    \end{equation}
    It is straightforward to convert this to a Euclidean distance between the vectors:
    \begin{equation}
    \begin{split}
        \min_{\theta \in \R} \| \ket{\widehat{\bm{u}}_j} - e^{i\theta} \ket{u_j} \| &= \sqrt{2 - 2 |\ip{\widehat{\bm{u}}_j}{u_j}|}\\
        &\leq \sqrt{2(1 - |\ip{\widehat{\bm{u}}_j}{u_j}|^2)}\\
        &= \frac{\sqrt{2}}{2} \|\op{\widehat{\bm{u}}_j}{\widehat{\bm{u}}_j} - \op{u_j}{u_j}\|_1\\
        &\leq \sqrt{2} \|\op{\widehat{\bm{u}}_j}{\widehat{\bm{u}}_j} - \op{u_j}{u_j}\|\\
        &\leq 2\sqrt{2} (\delta_{\pastext} + \delta_{\acttext}).
    \end{split}
    \end{equation}
    Hence there exists some diagonal real matrix $\Theta$ such that
    \begin{equation}\label{eq:unitary_fro_bound}
        \|\widehat{\bm{U}} - U e^{i\Theta}\| \leq \|\widehat{\bm{U}} - Ue^{i\Theta}\|_F \leq 2\sqrt{2n} (\delta_{\pastext} + \delta_{\acttext}),
    \end{equation}
    from which \cref{eq:phaseless_error_bound_active} follows by another application of \cref{claim:matrix_rounding_lemma}.
\end{proof}

\begin{remark}
    The astute reader may recognize that we use a lossy conversion from operator norm to Frobenius norm in \cref{eq:unitary_fro_bound}, which pays a factor of $\sqrt{n}$. In contrast, the analogous proof of \cite[Theorem 2.1]{haah2023query} (see \cref{prop:col_phase_fix}) uses the random matrix theory of isotropic errors to achieve an $n$-independent bound. Unfortunately, it is challenging to obtain such a bound here because $Z$ breaks that isotropy in the $\C^n$-space. Moreover, we will show in \cref{thm:rdm_var_arb} that taking $\delta_{\acttext} \sim 1/\sqrt{n}$ is already required in order to control the $\U(n)$-shadows variance.
\end{remark}

Recall that the passive algorithm recovers the column phases using the Fourier transform trick from \cref{prop:col_phase_fix}. We apply the same here, using the states
\begin{equation}
    \ket{\widetilde{\psi}_j} \coloneqq \act(\widehat{Q}_{\acttext}^\T) \act(Q) \pas(F)^\dagger \ket{1_j} = \act(Z) \pas(UF^\dagger) \ket{1_j},
\end{equation}
where $F \in \U(n)$ is the discrete Fourier transform.  Note that the presence of $\act(Z)$ causes no fundamental obstruction, as the success of the trick only relies on an error bound of the form \cref{eq:phaseless_error_bound_active}.

\begin{theorem}\label{thm:stage_2_alg}
    Let $\varepsilon_{\pastext} \leq \frac{1}{8}$. Suppose $\delta_{\pastext}, \delta_{\acttext} > 0$ are such that
    \begin{equation}
        4\sqrt{2n}(\delta_{\pastext} + \delta_{\acttext}) \leq \varepsilon_{\pastext}.
    \end{equation}
    Then given estimates for the $1$-RDMs of all $2n$ states $\ket{\psi_j}$ and $\ket{\widetilde{\psi}_j}$, each obeying \cref{eq:perturbed_rdm_shadows_guarantee}, we can output a unitary matrix $\bm{U}^\sharp$ such that $\phdist(\bm{U}^\sharp, U) \leq 25\varepsilon_{\pastext}$ with probability at least $1 - 2n\eta_{\pastext}$.
\end{theorem}

\begin{proof}
    The proof is identical to \cref{prop:col_phase_fix}, using \cref{lem:phaseless_act_learner} as the base ``up to column phases'' algorithm.
\end{proof}

\subsection{$\U(n)$-shadows beyond number symmetry}\label{sec:Un_shadows_redux}

\cite{low2022classical} originally described the $\U(n)$-shadows protocol within a fixed particle number subspace. This restriction is of course meaningful, but it turns out to be not necessary. The purpose of this subsection is to show how the method extends to arbitrary states, regardless of number symmetry. This is crucial for our analysis, because the probe states $\ket{\psi_j}$ and $\ket{\widetilde{\psi}_j}$ do not have such symmetry.

The key fact about classical shadows is that the learnability of observables is ultimately state-agnostic. There is a simple sufficient criterion to check:~whether or not the observable $O$ is orthogonal to the kernel of the measurement channel $\mathcal{M}$. Indeed, recall that classical shadow estimators are of the form $\tr(O\mathcal{M}^{-1}(\bm{\sigma}))$, where $\bm{\sigma}$ is the postmeasurement state obeying $\E[\bm{\sigma}] = \mathcal{M}(\rho)$ for input state $\rho$. But $\mathcal{M}$ is self-adjoint, so this is equal to $\tr(\mathcal{M}^{-1}(O) \bm{\sigma})$. If $O \in \ker(\mathcal{M})$, then the pseudoinverse returns $0$ and so the shadows protocol cannot learn $\tr(O\rho)$ for any $\rho$. More generally, if $O$ has any support in $\ker(\mathcal{M})$, then that component is killed and we cannot recover the associated information. Conversely, however, if $O \in \ker(\mathcal{M})^\perp$ then  $\E[\tr(O\mathcal{M}^{-1}(\bm{\sigma}))] = \tr(O \rho)$ for \emph{all} states $\rho$.

With this in mind, we can show that the estimator from \cref{prop:low_shadow_rdm} is capable of learning the fermionic $1$-RDM of any state, regardless of any symmetry it obeys.

\begin{claim}\label{claim:generalized_Un_shadow}
    The estimator from \cref{prop:low_shadow_rdm} continues to obey $\E[\widehat{\bm{\rdm}}] = \rdm$ for \emph{any} quantum state $\rho$, provided that we generalize the definition of $E(b)$ to
    \begin{equation}
        E(b) \coloneqq (n + 1) \diag(b) - |b| \I,
    \end{equation}
    where $|b|$ is the Hamming weight of $b \in \{0, 1\}^n$.
\end{claim}

\begin{proof}
    As alluded to above, this follows from a more general property of classical shadows. Recall that the $\U(n)$-shadow channel $\mathcal{M}$ is an average over two basic operations:~rotation by the group $\pas(\U(n))$ and projection into the number basis $\{0, 1\}^n$. Both are block-diagonal in the number basis, which implies that $\mathcal{M}$ admits the orthogonal decomposition
    \begin{equation}
        \mathcal{M} = \bigoplus_{k=0}^n \mathcal{M}_k.
    \end{equation}
    In the case that $\rho$ lies entirely in one particular number sector, say $k$, then all postmeasurement states $\bm{\sigma} = \pas(\bm{V})^\dagger \op{\bm{b}}{\bm{b}} \pas(\bm{V})$ also lie in that subspace. Hence $\mathcal{M}^{-1}(\bm{\sigma}) = \mathcal{M}_k^{-1}(\bm{\sigma})$ and we recover the familiar estimator from \cref{prop:low_shadow_rdm}:
    \begin{equation}\label{eq:rdm_shadow_decomp}
    \begin{split}
        \widehat{\bm{\rdm}}_{ij} &\coloneqq \tr(a_i^\dagger a_j \mathcal{M}^{-1}(\bm{\sigma}))\\
        &= \tr(a_i^\dagger a_j \mathcal{M}_\eta^{-1}(\bm{\sigma}))\\
        &= [\bm{V}^\dagger ((n+1) \diag(\bm{b}) - k \I) \bm{V}]_{ij}.
    \end{split}
    \end{equation}
    Now consider arbitrary $\rho$;~each instance of $\bm{\sigma}$ still lies in some $|\bm{b}|$-number sector, although $|\bm{b}|$ itself is now a random variable. Therefore $\mathcal{M}^{-1}(\bm{\sigma}) = \mathcal{M}_{|\bm{b}|}^{-1}(\bm{\sigma})$ so replacing $k \to |\bm{b}|$ in \cref{eq:rdm_shadow_decomp} shows the claim.
\end{proof}

\subsubsection{Perturbed RDM estimation}

The first modification of the analysis concerns the columnwise estimates of $U$. We re-analyze the application of matrix Bernstein for the states $\ket{\psi_j}$ and $\ket{\widetilde{\psi}_j}$, which are of the form $\act(Z) \ket{\phi}$ for an arbitrary single-particle Slater determinant $\ket{\phi}$. It will be useful to work with $Z$ in the ladder operator (rather than Majorana) representation. This is standard fare for working with Bogoliubov transformations, but for completeness we show how to derive it from the Majorana representation in \cref{sec:bogoliubov}.

\begin{claim}\label{claim:FLO_bogoliubov}
    Let $Z \in \Orth(2n)$. The FLO transformation by $Z$ on annihilation operators $a_1, \ldots, a_n$ can be expressed as
    \begin{equation}
        \act(Z)^\dagger a_j \act(Z) = \sum_{k=1}^n \l( \alpha_{jk} a_k + \beta_{jk}^* a_k^\dagger \r),
    \end{equation}
    where $\alpha, \beta \in \C^{n \times n}$ are such that
    \begin{equation}\label{eq:fgu_to_bogo}
        \begin{pmatrix}
            \alpha & \beta^*\\
            \beta & \alpha^*
        \end{pmatrix} = \Omega Z \Omega^\dagger, \quad \text{where } \Omega = \frac{1}{\sqrt{2}} \begin{pmatrix}
            \I & i\I\\
            \I & -i\I
        \end{pmatrix}.
    \end{equation}
\end{claim}

The matrix variances are bounded as follows.

\begin{theorem}\label{thm:rdm_var_arb}
     Let $\ket{\psi} = \act(Z) \ket{\phi}$ where $Z \in \SO(2n)$ and $\ket{\phi}$ is a single-particle Slater determinant. Let $\widehat{\bm{\rdm}}_1, \ldots, \widehat{\bm{\rdm}}_N$ be i.i.d.~copies of its $\U(n)$-shadow estimate $\widehat{\bm{\rdm}} = \bm{V}^\dagger E(\bm{b}) \bm{V}$, and define $\bm{X}_\ell \coloneqq \frac{1}{N}(\widehat{\bm{\rdm}}_\ell - \rdm)$ where $\rdm$ is the $1$-RDM of $\ket{\psi}$. Suppose that $\|Z - \I\| \leq \frac{c}{\sqrt{n}}$ for some $0 < c < 1$. Then the parameters for the matrix Bernstein inequality (\cref{prop:matrix_bernstein}) applied to the sequence $(\bm{X}_\ell)_{\ell \in [N]}$ can be taken as
     \begin{equation}
         B = \frac{n + 1}{N} \quad \text{and} \quad \sigma^2 \leq \frac{C_1 n + C_2}{N},
     \end{equation}
     where $C_1 = 2(\sqrt{2 + 7c^2 + c^4} + 1) + c^2$ and $C_2 = 5 + 8c^2 + c^4$.
\end{theorem}

\begin{proof}
    The norm bound $B$ is the same as in \cref{lem:rdm_bounds} because it still holds that $\|\rdm\| \leq 1$ and $\|\widehat{\bm{\rdm}}\| = \|E(\bm{b})\| = \max\{n+1-|\bm{b}|, |\bm{b}|\} \leq n$.
    
    For the variance, recall \cref{eq:E(b)_derivation} except replacing $\eta$ with $|\bm{b}|$:
    \begin{equation}
        E(\bm{b})^2 = (n + 1 - 2|\bm{b}|) E(\bm{b}) + |\bm{b}|(n+1-|\bm{b}|) \I.
    \end{equation}
    Then since $\widehat{\bm{\rdm}}^2 = \bm{V}^\dagger E(\bm{b})^2 \bm{V}$, we have
    \begin{equation}\label{eq:sigma_bound_perturb}
    \begin{split}
        \sigma^2 &= \l\| \sum_{\ell=1}^N \E[\bm{X}_\ell^2] \r\| = \frac{1}{N} \l\| {\E[\widehat{\bm{\rdm}}^2] - \rdm^2} \r\|\\
        &= \frac{1}{N} \l\| {\E[(n+1-2|\bm{b}|) \widehat{\bm{\rdm}} + |\bm{b}|(n+1-|\bm{b}|) \I] - \rdm^2} \r\|\\
        &= \frac{1}{N} \l\| {(n+1) \rdm - 2 \E[|\bm{b}| \widehat{\bm{\rdm}}] + \l( (n+1) \E|\bm{b}| -\E|\bm{b}|^2 \r) \I - \rdm^2} \r\|.
    \end{split}
    \end{equation}
    In contrast to the number-conserving case, $|\bm{b}|$ is now a random variable, so we need to control its first two moments. These are simply the moments of the particle number operator $\Num \coloneqq \sum_{j=1}^n a_j^\dagger a_j$ with respect to the fixed state $\ket{\psi}$, as the random measurement bases $\bm{V}$ commute with $\Num$. We defer this calculation to \cref{lem:num_op_moment_bounds} below, where we show that
    \begin{align}
        \ev{\Num}{\psi} &\leq 1 + \|\beta\|_F^2, \label{eq:num_bound}\\
        \ev{\Num^2}{\psi} &\leq 2 + 7\|\beta\|_F^2 + \|\beta\|_F^4, \label{eq:num2_bound}
    \end{align}
    where $\beta \in \C^{n \times n}$ is given by the representation of $Z$ described in \cref{claim:FLO_bogoliubov}.
    
    To bound the norm of $\beta$ in terms of $\delta_{\acttext} \geq \|Z - \I\|$, observe that
    \begin{equation}
        \|\beta\| = \l\| \begin{pmatrix}
            0 & 0\\
            \beta & 0
        \end{pmatrix} \r\| = \l\| \begin{pmatrix}
            0 & 0\\
            0 & \I
        \end{pmatrix} \begin{pmatrix}
            \alpha - \I & \beta^*\\
            \beta & \alpha^* - \I
        \end{pmatrix} \begin{pmatrix}
            \I & 0\\
            0 & 0
        \end{pmatrix} \r\| \leq \| \Omega Z \Omega^\dagger - \I\|.
    \end{equation}
    Hence the Frobenius norm obeys $\|\beta\|_F^2 \leq n \delta_{\acttext}^2$ and so
    \begin{align}
        \E|\bm{b}| &= \ev{\Num}{\psi} \leq 1 + n\delta_{\acttext}^2,\\
        \E|\bm{b}|^2 &= \ev{\Num^2}{\psi} \leq 2 + 7n\delta_{\acttext}^2 + n^2\delta_{\acttext}^4.
    \end{align}
    A crude triangle inequality suffices to get an estimate of \cref{eq:sigma_bound_perturb}:
    \begin{equation}
        N\sigma^2 \leq (n+2) + 2 \|{\E[|\bm{b}| \widehat{\bm{\rdm}}]}\| + (n+1) \E|\bm{b}| + \E|\bm{b}|^2.
    \end{equation}
    To handle the cross term we apply Cauchy--Schwarz for expectations:
    \begin{equation}
    \begin{split}
        \|{\E[|\bm{b}| \widehat{\bm{\rdm}}]}\| &= \max_{v \in \mathbb{S}^{n-1}} \l| {\E[|\bm{b}| \ev{\widehat{\bm{\rdm}}}{v}]} \r|\\
        &\leq \max_{v \in \mathbb{S}^{n-1}} \l( \sqrt{\E|\bm{b}|^2} \sqrt{\E[\ev{\widehat{\bm{\rdm}}}{v}^2]} \r)\\
        &\leq \sqrt{\E|\bm{b}|^2} \sqrt{\E\|\widehat{\bm{\rdm}}\|^2}\\
        &\leq n \sqrt{\E|\bm{b}|^2}.
    \end{split}
    \end{equation}
    The advertised bound for $\sigma^2$ follows from setting $\delta_{\acttext} = \frac{c}{\sqrt{n}}$ and some elementary algebra.
\end{proof}

Thus, provided that we execute the first-stage learning (\cref{cor:vacuum_learning_error}) with error $\delta_{\acttext} = \O(1/\sqrt{n})$, the number of copies to estimate the RDM of $\act(Z) \ket{\phi}$ is nearly the same as that for $\ket{\phi}$ itself. This establishes the size of $N$ sufficient to guarantee the event in \cref{eq:perturbed_rdm_shadows_guarantee}.

\begin{corollary}\label{cor:stage2_copy_complexity}
    Let $\ket{\psi} = \act(Z) \ket{\phi}$, $\rdm$, and $(\widehat{\bm{\rdm}}_\ell)_{\ell \in [N]}$ be as in \cref{thm:rdm_var_arb}. Fix $\delta_{\pastext}, \eta_{\pastext} \in (0, 1)$. Except with probability $\eta_{\pastext}$, it holds that
    \begin{equation}
        \|\overline{\bm{\rdm}} - \rdm\| \leq \delta_{\pastext}, \quad \text{where } \overline{\bm{\rdm}} \coloneqq \frac{1}{N} \sum_{\ell=1}^N \widehat{\bm{\rdm}}_\ell,
    \end{equation}
    provided that
    \begin{equation}
        N \geq \frac{(C_1' n + C_2') \log(2n/\eta_{\pastext})}{\delta_{\pastext}^2}.
    \end{equation}
    The constants above can be taken as $C_i' = 2(C_i + \frac{1}{3})$ where $C_1, C_2$ are also from \cref{thm:rdm_var_arb}.
\end{corollary}

\begin{proof}
    Follows from applying \cref{thm:rdm_var_arb} to \cref{prop:matrix_bernstein}.
\end{proof}

The remainder of this subsection is devoted to proving the claimed bounds on $\ev{\Num}{\psi}$ and $\ev{\Num^2}{\psi}$ from \cref{eq:num_bound,eq:num2_bound}. We will need the following formulation of Wick's theorem due to Lieb~\cite{lieb1968theorem} (see also \cite{bertsch2012symmetry,hu2014matrix} for more modern treatments and generalizations).

\begin{proposition}[Wick's theorem]\label{prop:wick}
    Let $\ket{\psi}$ be a fermionic Gaussian state and $b_1, \ldots, b_m$ any sequence of single-mode operators, i.e.,
    \begin{equation}
        b_j = \sum_{k=1}^n \l( A_{jk} a_k + B_{jk} a_k^\dagger \r)
    \end{equation}
    for arbitrary complex coefficients $A_{jk}, B_{jk} \in \C$. Define the skew-symmetric matrix $S \in \C^{m \times m}$ whose entries are the two-point correlators:
    \begin{equation}
        S_{pq} \coloneqq \begin{cases}
        \ev{b_p b_q}{\psi} & \text{if } p < q,\\
        -S_{qp} & \text{if } p > q,\\
        0 & \text{if } p = q.
        \end{cases}
    \end{equation}
    Then the many-body correlator is the Pfaffian of $S$:
    \begin{equation}
        \ev{b_1 \cdots b_m}{\psi} = \pf(S).
    \end{equation}
\end{proposition}

To form the associated $S$ matrix, we need to establish the relevant two-point correlators. It will be convenient to view the FLO $\act(Z)$ as rotating the modes, while taking the state in Wick's theorem as $\ket{\phi}$.

\begin{claim}\label{claim:two-pt-corr}
    Let $\ket{\phi}$ be a Slater determinant. Let $Z \in \Orth(2n)$ with $\alpha, \beta \in \C^{n \times n}$ as in \cref{claim:FLO_bogoliubov}. Define the quasi-particle operators
    \begin{equation}
        b_j \coloneqq \act(Z)^\dagger a_j \act(Z), \quad b_j^\dagger = \act(Z)^\dagger a_j^\dagger \act(Z)
    \end{equation}
    for each $j \in [n]$. Then
    \begin{subequations}\label{eq:bbs}
    \begin{align}
        \ev{b_j^\dagger b_k}{\phi} &= [ \alpha^* C \alpha^\T + \beta (\I - C^\T) \beta^\dagger ]_{jk}, \label{eq:b^b}\\
        \ev{b_j^\dagger b_k^\dagger}{\phi} &= [ \alpha^* C \beta^\T + \beta (\I - C^\T) \alpha^\dagger ]_{jk}, \label{eq:b^b^}\\
        \ev{b_j b_k^\dagger}{\phi} &= [ \I - \alpha C^\T \alpha^\dagger - \beta^* (\I - C) \beta^\T ]_{jk}, \label{eq:bb^}\\
        \ev{b_j b_k}{\phi} &= [\beta^* C \alpha^\T + \alpha (\I - C^\T) \beta^\dagger]_{jk}, \label{eq:bb}
    \end{align}
    \end{subequations}
    where $C_{pq} \coloneqq \ev{a_p^\dagger a_q}{\phi}$ is the $1$-RDM of $\ket{\phi}$.
\end{claim}

\begin{proof}
    First, we establish \cref{eq:b^b}:
    \begin{equation}
    \begin{split}
        \ev{b_j^\dagger b_k}{\phi} &= \sum_{p,q=1}^n \ev{\l( \alpha_{jp}^* a_p^\dagger + \beta_{jp} a_p \r) \l( \alpha_{kq} a_q + \beta_{kq}^* a_q^\dagger \r)}{\phi}\\
        &= \sum_{p,q=1}^n \ev{\l( \alpha_{jp}^* \alpha_{kq} a_p^\dagger a_q + \beta_{jp} \beta_{kq}^* a_p a_q^\dagger \r)}{\phi}\\
        &= \sum_{p,q=1}^n \l( \alpha_{jp}^* \alpha_{kq} \ev{a_p^\dagger a_q}{\phi} + \beta_{jp} \beta_{kq}^* (\delta_{pq} - \ev{a_q^\dagger a_p}{\phi}) \r)\\
        &= [ \alpha^* C \alpha^\T + \beta (\I - C^\T) \beta^\dagger ]_{jk},
    \end{split}
    \end{equation}
    where we used the fact that $\ev{a_p a_q}{\phi} = \ev{a_p^\dagger a_q^\dagger}{\phi} = 0$ on the second line and the CAR $a_p a_q^\dagger + a_q^\dagger a_p = \delta_{pq}$ on the third. Then the derivation for \cref{eq:b^b^} is analogous, merely swapping the roles of $\alpha$ and $\beta$ acting on the right. Finally, \cref{eq:bb^} is a consequence of the identity $\ev{b_j b_k^\dagger}{\phi} = \delta_{jk} - \ev{b_k^\dagger b_j}{\phi}$ (the $b_j$'s obey the same CAR because anticommutators are unitarily invariant), and \cref{eq:bb} follows from $\ev{b_j b_k}{\phi} = \ev{b_k^\dagger b_j^\dagger}{\phi}^*$ (also recall that $C$ is Hermitian).
\end{proof}

We are now ready to prove the claim from \cref{eq:num_bound,eq:num2_bound}.

\begin{lemma}\label{lem:num_op_moment_bounds}
    Let $\ket{\psi} = \act(Z) \ket{\phi}$ where $Z \in \Orth(2n)$ and $\ket{\phi}$ is a single-particle Slater determinant. It holds that
    \begin{align}
        \ev{\Num}{\psi} &\leq 1 + \|\beta\|_F^2, \label{eq:num_bound_rep}\\
        \ev{\Num^2}{\psi} &\leq 2 + 7\|\beta\|_F^2 + \|\beta\|_F^4, \label{eq:num2_bound_rep}
    \end{align}
    where $\beta \in \C^{n \times n}$ is as in \cref{claim:FLO_bogoliubov}.
\end{lemma}

\begin{proof}
    We begin with the first moment. Define $b_j$ as in \cref{claim:two-pt-corr}. Using \cref{eq:b^b}, we get
    \begin{equation}
    \begin{split}
        \ev{\Num}{\psi} &= \sum_{j=1}^n \ev{a_j^\dagger a_j}{\psi}\\
        &= \sum_{j=1}^n \ev{b_j^\dagger b_j}{\phi}\\
        &= \tr(\alpha^* C \alpha^\T + \beta (\I - C^\T) \beta^\dagger).
    \end{split}
    \end{equation}
    We simplify this via the following facts. First, unitarity of $\Omega Z \Omega^\dagger$ implies that $\alpha^\T \alpha^* + \beta^\T \beta^* = \I$. Second, $\tr(C) = 1$ because $\ket{\phi}$ is a single-particle state. Finally, because both $\beta^\T \beta^*$ and $C$ are PSD, $\tr(\beta^\T \beta^* C) \geq 0$. Altogether, these imply that
    \begin{equation}
    \begin{split}
        \ev{\Num}{\psi} &= \tr((\alpha^\T \alpha^* + \beta^\T \beta^*) C) - 2 \tr(\beta^\T \beta^* C) + \|\beta\|_F^2\\
        &\leq 1 + \|\beta\|_F^2,
    \end{split}
    \end{equation}
    which is \cref{eq:num_bound_rep}.

    Next we consider the second moment. For each term in $\Num^2 = \sum_{j,k=1}^n a_j^\dagger a_j a_k^\dagger a_k$, we apply Wick's theorem (\cref{prop:wick}). Recall that the Pfaffian of a $4 \times 4$ matrix $S$ is $\pf(S) = S_{12} S_{34} - S_{13} S_{24} + S_{14} S_{23}$. Hence with \cref{eq:bbs},
    \begin{equation}
    \begin{split}
        \ev{a_j^\dagger a_j a_k^\dagger a_k}{\psi} &= \ev{b_j^\dagger b_j b_k^\dagger b_k}{\phi}\\
        &= \ev{b_j^\dagger b_j}{\phi} \ev{b_k^\dagger b_k}{\phi} - \ev{b_j^\dagger b_k^\dagger}{\phi} \ev{b_j b_k}{\phi} + \ev{b_j^\dagger b_k}{\phi} \ev{b_j b_k^\dagger}{\phi}\\
        &= \ev{b_j^\dagger b_j}{\phi} \ev{b_k^\dagger b_k}{\phi}\\
        &\hphantom{=~} - [ \alpha^* C \beta^\T + \beta (\I - C^\T) \alpha^\dagger ]_{jk} [\beta^* C \alpha^\T + \alpha (\I - C^\T) \beta^\dagger]_{jk}\\
        &\hphantom{=~} + [ \alpha^* C \alpha^\T + \beta (\I - C^\T) \beta^\dagger ]_{jk} [ \I - \alpha C^\T \alpha^\dagger - \beta^* (\I - C) \beta^\T ]_{jk}.
    \end{split}
    \end{equation}
    Now sum each term from Wick's theorem over $j$ and $k$. The first simply yields $\ev{\Num}{\psi}^2$. The second is
    \begin{equation}
    \begin{split}
        &\hphantom{=~} \sum_{j,k=1}^n [ \alpha^* C \beta^\T + \beta (\I - C^\T) \alpha^\dagger ]_{jk} [\beta^* C \alpha^\T + \alpha (\I - C^\T) \beta^\dagger]_{jk}\\
        &= \sum_{j,k=1}^n [ \alpha^* C \beta^\T + \beta (\I - C^\T) \alpha^\dagger ]_{jk} [\beta C^* \alpha^\dagger + \alpha^* (\I - C^\dagger) \beta^\T]_{jk}^*\\
        &= \sum_{j,k=1}^n [ \alpha^* C \beta^\T + \beta (\I - C^\T) \alpha^\dagger ]_{jk} [ \beta C^\T \alpha^\dagger + \alpha^* \beta^\T - \alpha^* C \beta^\T ]_{jk}^*\\
        &= -\sum_{j,k=1}^n [ \alpha^* C \beta^\T + \beta (\I - C^\T) \alpha^\dagger ]_{jk} [ \alpha^* C \beta^\T + \beta (\I - C^\T) \alpha^\dagger - \alpha^* \beta^\T - \beta \alpha^\dagger ]_{jk}^*\\
        &= -\|\alpha^* C \beta^\T + \beta (\I - C^\T) \alpha^\dagger\|_F^2,
    \end{split}
    \end{equation}
    where on the last line we used that $\alpha^* \beta^\T + \beta \alpha^\dagger = 0$ due to unitarity of $\Omega Z \Omega^\dagger$. The third term from Wick's theorem is
    \begin{equation}
    \begin{split}
        &\hphantom{=~} \sum_{j,k=1}^n [ \alpha^* C \alpha^\T + \beta (\I - C^\T) \beta^\dagger ]_{jk} [ \I - \alpha C^\T \alpha^\dagger - \beta^* (\I - C) \beta^\T ]_{jk}\\
        &= \sum_{j,k=1}^n [ \alpha^* C \alpha^\T + \beta (\I - C^\T) \beta^\dagger ]_{jk} [ \I - \alpha^* C^\dagger \alpha^\T - \beta (\I - C^*) \beta^\dagger ]_{jk}^*\\
        &= \tr(\alpha^* C \alpha^\T + \beta (\I - C^\T) \beta^\dagger) - \|\alpha^* C \alpha^\T + \beta (\I - C^\T) \beta^\dagger\|_F^2\\
        &= \ev{\Num}{\psi} - \|\alpha^* C \alpha^\T + \beta (\I - C^\T) \beta^\dagger\|_F^2.
    \end{split}
    \end{equation}
    We collect these results to obtain
    \begin{align}
        \ev{\Num^2}{\psi} &= \ev{\Num}{\psi}^2 + \|\alpha^* C \beta^\T + \beta (\I - C^\T) \alpha^\dagger\|_F^2 + \ev{\Num}{\psi} - \|\alpha^* C \alpha^\T + \beta (\I - C^\T) \beta^\dagger\|_F^2 \notag\\
        &\leq (1 + \|\beta\|_F^2)^2 + (1 + \|\beta\|_F^2) + \|\alpha^* C \beta^\T + \beta (\I - C^\T) \alpha^\dagger\|_F^2.\label{eq:num2_intermediate_bound}
    \end{align}
    It remains to bound the last term. Note that the Frobenius norm has the property $\|AB\|_F \leq \|A\|_F \|B\|$, which is stronger than mere submultiplicativity. Use this property to get
    \begin{equation}\label{eq:fro_norm_aCb}
    \begin{split}
        \|\alpha^* C \beta^\T + \beta (\I - C^\T) \alpha^\dagger\|_F &\leq \|\alpha^* C \beta^\T\|_F + \|\beta (\I - C^\T) \alpha^\dagger\|_F\\
        &\leq \|\alpha\| \|C\| \|\beta\|_F + \|\beta\|_F \|\I - C\| \|\alpha\| \\
        &\leq 2 \|\beta\|_F,
    \end{split}
    \end{equation}
    where we used the additional fact that $\|\alpha\| \leq 1$ since $\alpha$ is a block within a unitary matrix. Plug this bound into \cref{eq:num2_intermediate_bound} to conclude the result of \cref{eq:num_bound_rep}.
\end{proof}

We remark that one could obtain a slightly tighter bound by placing the Frobenius norm onto the $C$ and $\I - C$ terms rather than $\beta$ in \cref{eq:fro_norm_aCb}. We opt for the looser bound above due to clarity of presentation, as this ultimately only affects constant factors in our analysis.

\subsubsection{Perturbed $\U(1)$ phase estimation}

The final piece of the base passive algorithm is to learn the overall phase of $U$. Again, the procedure from is unchanged from \cref{sec:phase_learning};~we only need to re-analyze the error bounds in the presence of $Z$. Following the presentation above, we will condition on the high-probability events that both $\bm{Z}$ and $\bm{W}$ (from \cref{eq:W_defn}) are close to identity.

\begin{theorem}\label{thm:active_alg_phase}
    Let $Q \in \Orth(2n)$ and $0 < \varepsilon, p < \frac{1}{2}$. Suppose we have matrices $\widehat{Q}_{\acttext} \in \Orth(2n)$ and $\widehat{U} \in \U(n)$, such that:
    \begin{enumerate}
        \item $\phdist(\widehat{U}, U) \leq \varepsilon$ where $U$ corresponds to $Q_{\pastext}$ as in \cref{eq:Qpas_defn}, and
        \item $\|Z - \I\| \leq \varepsilon $ where $Z = \widehat{Q}_{\acttext}^\T Q \widehat{Q}_{\pastext}^\T$ ($\widehat{Q}_{\pastext} \in \USp{2n}$ corresponds to $\widehat{U}$).
    \end{enumerate}
    Using $\O(\log(1/p)/\varepsilon^2)$ queries to $\act(Q)$, we can output a unitary matrix $\bm{U}^{\sharp}$ such that
    \begin{equation}
        \|\bm{U}^{\sharp} - U\| \leq 9\varepsilon \quad \text{except with probability } 2p.
    \end{equation}
\end{theorem}

\begin{proof}
    Begin by observing that the unitary we apply is
    \begin{equation}
        \act(\widehat{Q}_{\acttext}^\T) \act(Q) \pas(\widehat{U}^\dagger) = \act(Z) \pas(U) \pas(\widehat{U}^\dagger) = \act(Z) \pas(W) \pas(e^{i\theta} \I),
    \end{equation}
    where $\theta \in [-\pi, \pi)$ and $W \in \U(n)$ are as in \cref{eq:W_defn}. The analysis is therefore identical to \cref{thm:phase_est_copies}, up to conjugating $X$ and $Y$ by $\act(Z)$. Let us re-define $\ket{\widetilde{\Psi}(\theta)} \coloneqq \act(Z) \pas(W) \ket{\Psi(\theta)}$ and write $[\alpha W]_{11} = re^{i\xi}$. We will show shortly that
    \begin{subequations}\label{eq:perturbed_cos}
    \begin{align}
        \ev{X}{\widetilde{\Psi}(\theta)} &= r \cos(\theta + \xi),\\
        \ev{Y}{\widetilde{\Psi}(\theta)} &= r \sin(\theta + \xi).
    \end{align}
    \end{subequations}
    The proof then follows \cref{thm:phase_est_copies,cor:phase_unitary_final} except that now $|r - 1| \leq \|\alpha W - \I\| \leq \|W - \I\| + \|\alpha - \I\| \leq 2 \varepsilon$.

    We conclude with the derivation of \cref{eq:perturbed_cos}. Embed $Z$ into $\SO(2n+2)$ such that it acts on the ancilla mode trivially. In the $(\alpha, \beta)$ representation, the Bogoliubov transformation is
    \begin{equation}
        \act(Z)^\dagger a_1^\dagger a_{n+1}^\dagger \act(Z) = \sum_{m=1}^n (\alpha_{1m}^* a_m^\dagger + \beta_{1m} a_m) a_{n+1}^\dagger.
    \end{equation}
    The expectation of the $a_m^\dagger a_{n+1}^\dagger$ terms was computed in \cref{claim:passive_perturbed_cos}:
    \begin{equation}
        \ev{\pas(W)^\dagger a_m^\dagger a_{n+1}^\dagger \pas(W)}{\Psi(\theta)} = \frac{1}{2} e^{-i\theta} W_{m1}^*.
    \end{equation}
    The $a_m a_{n+1}^\dagger$ term follows an analogous calculation:
    \begin{equation}
    \begin{split}
        \pas(W)^\dagger a_m a_{n+1}^\dagger \pas(W) &= \sum_{j,k=1}^{n+1} [W \oplus 1]_{mj} [W \oplus 1]_{n+1,k}^* a_j a_k^\dagger
    \end{split}
    \end{equation}
    and
    \begin{equation}
    \begin{split}
        \ev{a_j a_k^\dagger}{\Psi(\theta)} &= \frac{1}{2} \l( \melem{\vac{n+1}}{a_j a_k^\dagger}{\vac{n+1}} + e^{i\theta} \melem{\vac{n+1}}{a_j a_k^\dagger}{1_1 1_{n+1}} \r.\\
        &\hphantom{=~} \l. + \, e^{-i\theta} \melem{1_1 1_{n+1}}{a_j a_k^\dagger}{\vac{n+1}} + \melem{1_1 1_{n+1}}{a_j a_k^\dagger}{1_1 1_{n+1}} \r)\\
        &= \frac{1}{2} (\delta_{jk} - \melem{1_1 1_{n+1}}{a_k^\dagger a_j}{1_1 1_{n+1}}) = B_{jk},
    \end{split}
    \end{equation}
    where we define the matrix $B \coloneqq \frac{1}{2} \diag(0, 1, \ldots, 1, 0) \in \R^{(n+1) \times (n+1)}$. Hence
    \begin{equation}
    \begin{split}
        \ev{\pas(W)^\dagger a_m a_{n+1}^\dagger \pas(W)}{\Psi(\theta)} = [(W \oplus 1) B (W^\dagger \oplus 1)]_{m,n+1} = 0
    \end{split}
    \end{equation}
    since the last column of $(W \oplus 1) B (W^\dagger \oplus 1)$ is $0$. Altogether, we get
    \begin{equation*}
        \ev{\pas(W)^\dagger \act(Z)^\dagger a_1^\dagger a_{n+1}^\dagger \act(Z) \pas(W)}{\Psi(\theta)} = \frac{1}{2} e^{-i\theta} [\alpha W]_{11}^*. \qedhere
    \end{equation*}
\end{proof}

\subsection{Piecing together the two-stage base algorithm}

Let us now summarize the components constituting our base tomography algorithm for active FLOs. The idea is conceptually straightforward, outlined in \cref{alg:base_active}. Note that $C_1, C_2, C_3,$ and $K$ are some absolute constants;~an explicit but loose choice can be found below \cref{eq:base_queries}.

\begin{algorithm}
    \caption{Learning active FLOs, base case:~$\acttomo(\act(Q), \varepsilon, \delta)$}\label{alg:base_active}
    \setcounter{AlgoLine}{0}
    \KwInput{Query access to an $n$-mode active FLO $\act(Q)$ and error parameters $\varepsilon, \delta \in (0, 1)$.}
    
    \KwOutput{An orthogonal matrix $\widehat{\bm{Q}}$ such that $\|\widehat{\bm{Q}} - Q\| \leq \varepsilon$ with probability at least $1 - \delta$.}

    \BlankLine

    Let $N_{\acttext} \gets \l\lceil \frac{C_1 n^3 \log(K n/\delta)}{\varepsilon^2} \r\rceil$;
    
    $\bm{\covmat}^{\star} \gets \gausstomo(\act(Q) \ket{\vac{n}}, N_{\acttext})$; \hfill \Comment{\cref{alg:gauss}}

    Compute the normal form of $\bm{\covmat}^{\star} = \widehat{\bm{Q}}_{\acttext} J \widehat{\bm{Q}}_{\acttext}^\T$;

    Construct the FLO circuit $\act(\widehat{\bm{Q}}_{\acttext}^\T)$;

    Let $N_{\pastext} \gets \l\lceil \frac{C_2 n^2 \log(K n^2/\delta)}{\varepsilon^2} \r\rceil$;

    Let $N_{\ph} \gets \l\lceil \frac{C_3 \log(K/\delta)}{\varepsilon^2} \r\rceil$;
    
    $\bm{U}^{\sharp} \gets \pastomo(\act(\widehat{\bm{Q}}_{\acttext}^\T) \act(Q), N_{\pastext}, N_{\ph})$; \hfill \Comment{\cref{alg:base_passive}}

    Set $\widehat{\bm{Q}}_{\pastext} \gets \begin{pmatrix}
        \Re \bm{U}^{\sharp} & -\Im \bm{U}^{\sharp}\\
        \Im \bm{U}^{\sharp} & \Re \bm{U}^{\sharp}
    \end{pmatrix}$;\label{line:Q_pas}
    
    \Return{$\widehat{\bm{Q}} \gets \widehat{\bm{Q}}_{\acttext} \widehat{\bm{Q}}_{\pastext}$}
\end{algorithm}

\begin{claim}
    The output of \cref{alg:base_active} is correct and costs $\O(n^3 \log(n/\delta) / \varepsilon^2)$ queries.
\end{claim}

\begin{proof}
    The algorithm begins by running $\gausstomo$ on $N_{\acttext}$ copies of the state $\act(Q) \ket{\vac{n}}$. By \cref{cor:vacuum_learning_error}, if we take $N_{\acttext} = \l\lceil \frac{32 n^2 \log(4n/\eta_{\acttext})}{\delta_{\acttext}^2} \r\rceil$ then there exists a symplectic $\bm{Q}_{\pastext} \in \USp{n}$ such that $\|\widehat{\bm{Q}}_{\acttext} \bm{Q}_{\pastext} - Q\| \leq \delta_{\acttext}$, except with probability $\eta_{\acttext}$.
    
    Define $\bm{Z} \coloneqq \widehat{\bm{Q}}_{\acttext}^\T Q \bm{Q}_{\pastext}^\T$ and let $\bm{U}$ be the $\U(n)$-representation of $\bm{Q}_{\pastext}$. Condition on the success of the previous step. The circuit $\bm{\Circ} \coloneqq \act(\widehat{\bm{Q}}_{\acttext}^\T) \act(Q)$ is equivalent to $\act(\bm{Z}) \pas(\bm{U})$, which we input into $\pastomo$. By convention, $N_{\pastext}$ is the number of copies per state of the form $\ket{\psi_j} \coloneqq \bm{\Circ} \ket{1_j}$ and $\ket{\widetilde{\psi}_j} \coloneqq \bm{\Circ} \pas(F^\dagger) \ket{1_j}$ prepared by $\pastomo$, for a total of $2n N_{\pastext}$ queries to $\act(Q)$. The error analysis for this subroutine is as follows. Set $\delta_{\acttext} = \frac{c}{\sqrt{n}}$ for some small $c < 1$ to be determined later. Then we can use \cref{cor:stage2_copy_complexity}, which implies that $N_{\pastext} = \l\lceil \frac{48n \log(2n/\eta_{\pastext})}{\delta_{\pastext}^2} \r\rceil$ copies suffices to learn an RDM to error $\delta_{\pastext}$ in operator norm, except with probability $\eta_{\pastext}$. We set $\delta_{\pastext} = \frac{c}{\sqrt{n}}$ as well, allowing us to use \cref{thm:stage_2_alg} to find a unitary $\widehat{\bm{U}}$ such that $\phdist(\widehat{\bm{U}}, \bm{U}) \leq 200\sqrt{2}c$ except with probability $2n\eta_{\pastext}$.

    The final piece of $\pastomo$ is the $\U(1)$ phase estimation. This returns a phase $\widehat{\bm{\theta}} \in (-\pi, \pi]$ such that, by \cref{thm:active_alg_phase}, the unitary $\bm{U}^\sharp \coloneqq e^{i\widehat{\bm{\theta}}} \widehat{\bm{U}}$ obeys $\|\bm{U}^\sharp - \bm{U}\| \leq 1800\sqrt{2}c$. Conditioned on all prior steps succeeding, this holds with probability $1 - 2\eta_{\ph}$ if we make $2N_{\ph} = 2\l\lceil \frac{(6 + 4\sqrt{2}) \log(2/\eta_{\ph})}{80000 c^2} \r\rceil$ queries to $\act(Q)$ (\cref{thm:phase_est_copies}). The unconditional success probability is therefore at least $1 - \eta_\acttext - 2n\eta_\pastext - 2\eta_\ph$ by a union bound, to achieve an error of
    \begin{equation}
        \|\widehat{\bm{Q}}_\acttext \widehat{\bm{Q}}_\pastext - Q\| \leq \frac{c}{\sqrt{n}} + 1800\sqrt{2}c.
    \end{equation}
    Choosing $c = \frac{\varepsilon}{2600}$ is more than enough to bound this by $\varepsilon$, and choosing $\eta_\acttext = \frac{\delta}{3}$, $\eta_\pastext = \frac{\delta}{6n}$, and $\eta_\ph = \frac{\delta}{6}$ bounds the failure probability by $\delta$. The resulting query complexity is
    \begin{align}
        N_\acttext + 2n N_\pastext + 2 N_\ph &= \l\lceil \frac{32 n^2 \log(4n/\eta_{\acttext})}{\delta_{\acttext}^2} \r\rceil + 2n \l\lceil \frac{48n \log(2n/\eta_{\pastext})}{\delta_{\pastext}^2} \r\rceil + 2 \l\lceil \frac{(6 + 4\sqrt{2}) \log(2/\eta_{\ph})}{80000 c^2} \r\rceil \notag\\
        &\leq \l\lceil \frac{C_1 n^3 \log(K n/\delta)}{\varepsilon^2} \r\rceil + 2n \l\lceil \frac{C_2 n^2 \log(K n^2/\delta)}{\varepsilon^2} \r\rceil + 2 \l\lceil \frac{C_3 \log(K/\delta)}{\varepsilon^2} \r\rceil, \label{eq:base_queries}
    \end{align}
    where one can take $C_1 = 2.2 \times 10^8$, $C_2 = 3.3 \times 10^8$, $C_3 = 1000$, and $K = 12$.
\end{proof}

If we instead only target \cref{item:anc_2} from \cref{rem:ancilla_free_intro} (ancilla-free learning with parity-conserving interferometry), then we only aim to learn $\act(Q)$ up to a factor of $e^{i\pi\,\Num}$. The $\SO(2n)$-representation of this is $-\I$, hence the distance metric in \cref{alg:base_active} should be changed to the projective metric $\min_{s \in \{\pm 1\}} \|\widehat{\bm{Q}} - s Q\|$. This is precisely what it means to learn the $\U(1)$ phase $\bm{\theta}$ up to mod $\pi$;~the rest of the argument follows without modification.

\subsection{Applying the bootstrap}

As with the passive algorithm, bootstrapping to Heisenberg scaling is straightforward. We will only explicitly write down the analysis the diamond-distance learner here;~the ancilla-free analysis is completely analogous.

\begin{theorem}[\cref{thm:active_alg_intro}]\label{thm:active_alg_main}
    The output of $\bootstrap(\acttomo; \act(Q), \frac{\varepsilon}{n}, \delta)$ describes an FLO which is $\varepsilon$-close to $\act(Q)$ in diamond distance, with probability at least $1 - \delta$. The algorithm costs $\O(n^4 \log(n/\delta) / \varepsilon)$ queries, $\O(n^3/\varepsilon)$ quantum gates per experiment, and $\O(n^{\omega+3} \log^2(n/{\min\{\varepsilon, \delta\}}))$ classical computational time.
\end{theorem}

\begin{proof}
    As \cref{alg:base_active} indicates, $\acttomo(\act(Q), \frac{1}{10}, \delta)$ makes $\O(n^3 \log(n/\delta))$ queries. Hence by \cref{prop:bootstrap_FLO} the bootstrapped process with error $\varepsilon/n$ makes a total of $\O(n^4 \log(n/\delta) / \varepsilon)$ queries. This error in the operator norm of the $\Orth(2n)$-representation is chosen such that the diamond distance error is at most $\varepsilon$, per \cref{prop:FLO_stability}. See \cref{thm:passive_alg_main} for the gate complexity argument (note that both passive and active FLOs use $\O(n^2)$ gates).

    For the classical cost, we have that the circuit for $\act(\bm{V}_t^\dagger)$ can be determined in $\O(n^3)$ time (this only needs to be calculated once per iteration). For each iteration $t$,
    \begin{enumerate}
        \item $\gausstomo$ and computing the normal form costs $\O(N_\acttext n^\omega + n^3) = \O(n^{\omega+3} \log(n/\delta_t))$ operations (\cref{prop:final_gauss_alg_analysis});

        \item Determining the circuit for $\act(\widehat{\bm{Q}}_{\acttext})$ costs $\O(n^3)$ operations;

        \item $\pastomo$ costs $\O(n^4 \log(n/\delta_t))$ operations (from \cref{thm:passive_alg_main}, adjusted to account for the fact that the perturbed RDM estimates are rank-$\O(n)$ rather than rank-$1$);

        \item Forming $\widehat{\bm{Q}}_{\acttext} \widehat{\bm{Q}}_{\pastext}$ costs $\O(n^\omega)$ operations.
    \end{enumerate}
    The $\gausstomo$ step asymptotically dominates, so the total time complexity is
    \begin{equation*}
        \O\l( \sum_{t=0}^T n^{\omega+3} \log(n/\delta_t) \r) = \O\l( n^{\omega+3} (\log(n/\delta) T + T^2) \r) = \O\l( n^{\omega+3} \log^2(n/{\min\{\varepsilon, \delta\}}) \r). \qedhere
    \end{equation*}
\end{proof}

\section*{Acknowledgments}
\addcontentsline{toc}{section}{Acknowledgments}

We thank Sabee Grewal, Vishnu Iyer, Daniel Liang, and Antonio Anna Mele for helpful conversations. This work was supported by the Laboratory Directed Research and Development program at Sandia National Laboratories, under the Gil Herrera Fellowship in Quantum Information Science. Sandia National Laboratories is a multimission laboratory managed and operated by National Technology and Engineering Solutions of Sandia, LLC, a wholly owned subsidiary of Honeywell International, Inc., for the U.S.~Department of Energy’s National Nuclear Security Administration under contract DE-NA-0003525. AZ also acknowledges support from the U.S.~Department of Energy, Office of Science, Office of Advanced Scientific Computing Research, Accelerated Research in Quantum Computing.

\printbibliography
\addcontentsline{toc}{section}{References}

\appendix

\section{Tomography of Gaussian states}\label{sec:gaussian_state_tomo}

To learn Gaussian states of indeterminate particle number, we deploy a version of fermionic shadows employing active FLO measurements~\cite{zhao2021fermionic,wan2023matchgate,ogorman2022fermionic,heyraud2025unified}. Performance-wise, all these works present nearly identical protocols;~however, technically only \cite{zhao2021fermionic,heyraud2025unified} use parity-conserving random FLOs. For simplicity of exposition, we will adopt the $\SO(2n)$ distribution studied in \cite{heyraud2025unified}, although its Clifford subgroup is simpler to implement and promises the same guarantees~\cite{zhao2021fermionic}. Details of the measurement protocol notwithstanding, the final tomography analysis is effectively equivalent to \cite[Proposition D1]{bittel2025optimal}.

\begin{definition}
    Let $\rho$ be a quantum state on $n$ modes. The \emph{$\SO(2n)$-shadows protocol} is the following procedure:~for each copy of $\rho$,
    \begin{enumerate}
        \item Draw a random matrix $\bm{R} \sim \Haar(\SO(2n))$.
        \item Apply the unitary transformation $\rho \mapsto \act(\bm{R}) \rho \act(\bm{R})^\dagger$.
        \item Measure in the standard basis, obtaining the classical outcome $\bm{b} \in \{0, 1\}^n$ with probability $\ev{\act(\bm{R}) \rho \act(\bm{R})^\dagger}{\bm{b}}$.
    \end{enumerate}
    Each sample is stored as a tuple $(\bm{R}, \bm{b})$, which is an efficient classical description  of the postmeasurement state $\act(\bm{R})^\dagger \op{\bm{b}}{\bm{b}} \act(\bm{R})$.
\end{definition}

As before, this defines a quantum channel $\mathcal{M}$ such that $\rho = \E[\mathcal{M}^{-1}(\act(\bm{R})^\dagger \op{\bm{b}}{\bm{b}} \act(\bm{R}))]$. The $\SO(2n)$-shadows were designed to efficiently recover few-body fermionic observables;~the covariance matrix estimator can be expressed compactly as follows.

\begin{proposition}\label{prop:matchgate_shadow_rdm}
    Let $\rho$ be an $n$-mode state and $\covmat$ its covariance matrix. Let $(\bm{R}, \bm{b})$ be a single sample obtained by running the $\SO(2n)$-shadows protocol on a copy of $\rho$. Then the estimate $\widehat{\bm{\covmat}} = (2n-1) \bm{R}^\T J(\bm{b}) \bm{R}$, where
    \begin{equation}\label{eq:covmat_diag_term}
        J(b) \coloneqq \begin{pmatrix}
            0 & (-1)^{\diag(b)}\\
            -(-1)^{\diag(b)} & 0
        \end{pmatrix},
    \end{equation}
    obeys $\E[\widehat{\bm{\covmat}}] = \covmat$.
\end{proposition}

\begin{proof}
    The formula can be derived from any of the aforementioned papers on matchgate shadows;~we will follow \cite[Eq.~(37)]{wan2023matchgate} due to its relatively compact presentation. There they show that, for any quadratic Majorana observable $-i \gamma_j \gamma_k$, with $j \neq k$,
    \begin{equation}
        \tr(-i \gamma_j \gamma_k \mathcal{M}^{-1}(\act(\bm{R})^\dagger \op{\bm{b}}{\bm{b}} \act(\bm{R}))) = (2n-1) \pf((\bm{R}^\T J(\bm{b}) \bm{R})[(j, k)]),
    \end{equation}
    where $A[(j, k)]$ denotes the $(j, k)$-principal submatrix of a matrix $A$. For a $2 \times 2$ skew-symmetric matrix, the Pfaffian is simply the upper-right entry, so the right-hand side is $(2n-1) [\bm{R}^\T J(\bm{b}) \bm{R}]_{jk} = \widehat{\bm{\covmat}}_{jk}$. The left-hand side is $\tr(-i \gamma_j \gamma_k \rho) = \covmat_{jk}$ in expectation, by construction of the classical shadows.
\end{proof}

We can analyze the sample complexity for estimating $\covmat$ from $\SO(2n)$-shadows, again using the matrix Bernstein inequality. Note that although the precise statement we have written in \cref{prop:matrix_bernstein} concerns Hermitian matrices, the result holds more broadly. Here, it is enough to observe that if $\covmat$ is skew-symmetric, then $i\covmat$ is Hermitian. We need to determine the $B$ and $\sigma^2$ parameters in this case.

\begin{lemma}\label{lem:bernstein_params_gauss}
    Let $\rho$ be an $n$-mode state and $\covmat$ its covariance matrix. Let $\widehat{\bm{\covmat}}_1, \ldots, \widehat{\bm{\covmat}}_N$ be i.i.d.~$\SO(2n)$-shadow estimates of $\covmat$. Define $\bm{X}_\ell \coloneqq \frac{1}{N}(\widehat{\bm{\covmat}}_\ell - \covmat)$. Then from \cref{prop:matrix_bernstein} we can take the parameter $B$ as
    \begin{equation}
        B = \frac{2n}{N},
    \end{equation}
    and $\sigma^2$ obeys
    \begin{equation}
        \sigma^2 \leq \frac{(2n-1)^2 + 1}{N}.
    \end{equation}
\end{lemma}

\begin{proof}
    Since $\bm{R}$ and $J(\bm{b})$ are both orthogonal, so too is $\bm{R}^\T J(\bm{b}) \bm{R}$. Recall also that any covariance matrix obeys $\|\covmat\| \leq 1$. Thus for any $\ell \in [N]$,
    \begin{equation}
        \|\bm{X}_\ell\| \leq \frac{\|\widehat{\bm{\covmat}}_\ell\| + \|\covmat\|}{N} \leq \frac{(2n-1) + 1}{N} = \frac{2n}{N}.
    \end{equation}
    For the variance, observe that $J(\bm{b})^2 = -\I$, so
    \begin{equation}
        \sigma^2 = \l\| \frac{1}{N^2} \sum_{\ell=1}^N (\E[\widehat{\bm{\covmat}}_\ell^2] - \covmat^2) \r\| = \frac{1}{N} \l\| -(2n-1)^2 \I - \covmat^2 \r\| \leq \frac{(2n-1)^2 + 1}{N}.
    \end{equation}
\end{proof}

Parallel to \cref{thm:intermediate_1rdm}, this gives us the sample complexity for estimating the fermionic covariance matrix of \emph{any} state.

\begin{theorem}\label{thm:intermediate_covmat}
    Let $\varepsilon, \delta \in (0, 1)$. Suppose $\rho$ is an $n$-mode state, and let $\covmat_{jk} = -\frac{i}{2} \tr([\gamma_j, \gamma_k] \rho)$ be its covariance matrix. Consuming $N$ copies of $\rho$ with the $\SO(2n)$-shadows protocol, one can output an estimate $\overline{\bm{\covmat}} \in \R^{2n \times 2n}$ such that
    \begin{equation}\label{eq:covmat_intermediate_tail_bound}
        \Pr\l( \|\overline{\bm{\covmat}} - \covmat\| \geq \varepsilon \r) \leq \delta,
    \end{equation}
    provided that
    \begin{equation}\label{eq:covmat_intermediate_sample_complexity}
        N \geq \frac{8 n^2 \log(4n/\delta)}{\varepsilon^2}.
    \end{equation}
\end{theorem}

\begin{proof}
    Let $\widehat{\bm{\covmat}}_1, \ldots, \widehat{\bm{\covmat}}_N$ as in \cref{lem:bernstein_params_gauss} and set $\overline{\bm{\covmat}} \coloneqq \frac{1}{N} \sum_{\ell=1}^N \widehat{\bm{\covmat}}_\ell$. We recall the matrix Bernstein inequality (\cref{prop:matrix_bernstein}), where note that the matrices have linear dimension $2n$ and the $B$ and $\sigma^2$ parameters are given by \cref{lem:bernstein_params_gauss}:
    \begin{equation}
        \Pr\l( \|\overline{\bm{\covmat}} - \covmat\| \geq \varepsilon \r) \leq 4n \exp\l( \frac{-N \varepsilon^2 / 2}{(2n-1)^2 + 1 + 2n\varepsilon/3} \r).
    \end{equation}
    For $n \geq 1$ and $\varepsilon < 1$, taking $N$ as in \cref{eq:covmat_intermediate_sample_complexity} suffices to bound this probability by $\delta$.
\end{proof}

This is the copy complexity to learn the covariance matrix in operator norm, which is sufficient for our active FLO algorithm. The state tomography protocol with a final rounding step is outlined in \cref{alg:gauss}.

\begin{algorithm}[H]
    \caption{Learning pure Gaussian states:~$\gausstomo(\ket{\psi}, N)$}\label{alg:gauss}
    \setcounter{AlgoLine}{0}
    \KwInput{$N$ copies of an $n$-mode pure state $\ket{\psi}$.}
    \KwOutput{A covariance matrix $\bm{\covmat}^\star$ uniquely specifying a pure fermionic Gaussian state $\ket{\widehat{\bm{\psi}}}$.}

    \BlankLine

    Let $\overline{\bm{\covmat}} \gets 0$;

    \RepTimes{$N$}{
        Draw a random $\bm{R} \sim \Haar(\SO(2n))$;

        Construct the FLO circuit $\act(\bm{R})$ and apply it to $\ket{\psi}$;

        Measure in the standard basis, obtaining outcome $\bm{b} \in \{0, 1\}^n$;

        Let $J(\bm{b}) \gets \begin{pmatrix}
            0 & (-1)^{\diag(\bm{b})}\\
            -(-1)^{\diag(\bm{b})} & 0
        \end{pmatrix}$;

        Set $\overline{\bm{\covmat}} \gets \overline{\bm{\covmat}} + \frac{2n - 1}{N} \bm{R}^\T J(\bm{b}) \bm{R}$;
    }

    Compute the normal form of $\overline{\bm{\covmat}} = \bm{W} \bm{\Lambda} \bm{W}^\T$; \hfill \Comment{Top-right block of $\bm{\Lambda}$ is PSD}
    
    \Return{$\bm{\covmat}^\star \gets \bm{W} J \bm{W}^\T$}
\end{algorithm}

For completeness we show how this implies a trace-distance learner with $\Ot(n^3/\varepsilon^2)$ copies. The analysis for mixed states is very similar and achieves an $\Ot(n^4/\varepsilon^2)$ copy complexity (see \cite[Theorem D1]{bittel2025optimal}).

\begin{proposition}\label{prop:final_gauss_alg_analysis}
    Let $\varepsilon, \delta \in (0, 1)$. Let $\ket{\psi}$ be an $n$-mode pure Gaussian state. There exists an algorithm which consumes $N = \O(n^3 \log(n/\delta) / \varepsilon^2)$ copies of $\ket{\psi}$ and uses $\O(n^\omega N + n^3)$ classical computational effort to output an efficient classical description of a pure Gaussian state $\ket{\widehat{\bm{\psi}}}$ such that
    \begin{equation}\label{eq:gaussian_trace_dist_tomo}
        \trdist(\ket{\widehat{\bm{\psi}}}, \ket{\psi}) \leq \varepsilon \quad \text{with probability at least } 1 - \delta.
    \end{equation}
    Each measurement is implemented by $\O(n^2)$ elementary FLO gates.
\end{proposition}

\begin{proof}
    First we check the runtime of \cref{alg:gauss}. Any FLO requires $\O(n^2)$ elementary gates, and random instances can be constructed in $\O(n^2)$ time~\cite{braccia2025optimal}. Each repetition is dominated by the matrix multiplication of $\bm{R}^\T J(\bm{b}) \bm{R}$. The algorithm concludes by computing the normal form of $\overline{\bm{\covmat}}$, which takes $\O(n^3)$ time (\cref{claim:normal_form}).

    For the copy complexity, let $\covmat$, $\overline{\bm{\covmat}}$, and $\bm{\covmat}^\star$ be the covariance matrix of $\ket{\psi}$, the unrounded estimate, and the rounded estimate corresponding to Gaussian $\ket{\widehat{\bm{\psi}}}$, respectively. By \cref{claim:matrix_rounding_lemma,prop:covmat_bound}, we have
    \begin{equation}
        \trdist(\ket{\widehat{\bm{\psi}}}, \ket{\psi}) \leq \frac{1}{4} \|\bm{\covmat}^\star - \covmat\|_F \leq \frac{3}{4} \|\overline{\bm{\covmat}} - \covmat\|_F \leq \frac{3\sqrt{2n}}{4} \|\overline{\bm{\covmat}} - \covmat\|.
    \end{equation}
    Per \cref{thm:intermediate_covmat}, if we take $N = \l\lceil \frac{9n^3 \log(4n/\delta)}{\varepsilon^2} \r\rceil$ then \cref{eq:gaussian_trace_dist_tomo} holds.
\end{proof}

\section{Improved query complexity using Choi states}\label{sec:cubic_choi_alg}

Here we present an active FLO learner with query complexity $\Ot(n^3 / \varepsilon)$, matching that of our passive algorithm up to a logarithmic factor. The caveat is that we require an auxiliary quantum memory of $n$ modes to prepare so-called fermionic Choi states. We borrow the definition from \cite{iyer2025mildly}, modifying it to ensure even parity.

\begin{definition}
    Consider a Fock space of $2n$ fermion modes, where we regard the first $n$ as the system modes and the last $n$ as auxiliary modes. The \emph{fermionic EPR state} is a pure state $\ket{\fermEPR}$ such that
    \begin{equation}\label{eq:fEPR_def}
        \op{\fermEPR}{\fermEPR} = \prod_{j=1}^{2n} \l( \frac{\I + (-1)^{j+1} i \gamma_j \gamma_{j+2n}}{2} \r),
    \end{equation}
    and we say that $\mathcal{U} \ket{\fermEPR}$ is the \emph{fermionic Choi state} of a unitary $\mathcal{U}$ on $2n$ modes, provided that it obeys $[\mathcal{U}, \gamma_{j+2n}] = 0$ for all $j \in [2n]$.
\end{definition}

\begin{claim}\label{claim:fEPR_gauss}
    $\ket{\fermEPR}$ is a Gaussian state with even parity.
\end{claim}

\begin{proof}
    Write the vacuum state as
    \begin{equation}
        \op{\vac{2n}}{\vac{2n}} = \prod_{j=1}^n \l( \frac{\I - i \gamma_j \gamma_{j + n}}{2} \r) \l( \frac{\I - i \gamma_{j+2n} \gamma_{j + 3n}}{2} \r).
    \end{equation}
    Graphically, this is corresponds to a perfect matching on $[4n]$, as is \cref{eq:fEPR_def}. Let $\pi \in \Sym_{4n}$ be the permutation that maps between the two perfect matchings via $j+n \leftrightarrow j+2n$. This consists of $n$ swaps, so if $P_\pi \in \Orth(4n)$ is the matrix representation of $\pi$, then the corresponding unitary is $\act(P_\pi)$ with $\det(P_\pi) = (-1)^n$. To ensure even parity for all $n$, we can additionally flip every other matching by another introducing another permutation $\sigma \in \Sym_{4n}$ that performs $j \leftrightarrow j+2n$ if and only if $j$ is odd. This is equivalent to the staggered sign appearing in \cref{eq:fEPR_def} since Majoranas anticommute. Overall, this implies that $\act(P_\sigma P_\pi) \ket{\vac{2n}} = \ket{\fermEPR}$ where $\det(P_\sigma P_\pi) = (-1)^{2n} = 1$.
\end{proof}

\begin{claim}\label{claim:choi_covmat}
    For an FLO $\act(Q)$ on $n$ system modes, define
    \begin{equation}
        \ket{\fermEPR(Q)} \coloneqq \act(\widetilde{Q}) \ket{\fermEPR} \quad \text{where } \widetilde{Q} \coloneqq \begin{pmatrix}
            Q & 0\\
            0 & \I
        \end{pmatrix}.
    \end{equation}
    Its covariance matrix $\covmat$ takes the form
    \begin{equation}
        \covmat = \begin{pmatrix}
            0 & QS\\
            -(QS)^\T & 0
        \end{pmatrix} \quad \text{where } S = \diag(-1, 1, -1, \ldots, 1).
    \end{equation}
\end{claim}

\begin{proof}
    Let $1 \leq j < k \leq 2n$. Use the fact that $\widetilde{Q}$ acts trivially on the auxiliary modes and that Majorana monomials are trace-orthogonal to get
    \begin{equation}
        \covmat_{j,k+2n} = -i \ev{\gamma_j \gamma_{k+2n}}{\fermEPR(Q)} = \frac{(-1)^{k}}{2^{2n}} \tr\l( \gamma_j \act(\widetilde{Q}) \gamma_k \act(\widetilde{Q})^\dagger \r) = (-1)^k Q_{jk}.
    \end{equation}
    Meanwhile the diagonal blocks of $\covmat$ vanish because $\ket{\fermEPR(Q)}$ is pure Gaussian, so $\covmat$ must be orthogonal.
\end{proof}

The Choi-state algorithm is simple:~learning the covariance matrix of $\ket{\fermEPR(Q)}$ via \cref{alg:gauss} yields a constant-error estimate of $Q$ using only $\Ot(n^2)$ copies. This can then be bootstrapped into an $\Ot(n^3 / \varepsilon)$-query protocol.

\begin{lemma}\label{lem:learning_from_choi}
    There is an efficient algorithm that consumes $\O(n^2 \log(n/\delta) / \varepsilon^2)$ copies of $\ket{\fermEPR(Q)}$ and outputs some $\widehat{\bm{Q}} \in \Orth(2n)$ such that
    \begin{equation}\label{eq:choi_Q_guarantee}
        \Pr\l( \|\widehat{\bm{Q}} - Q\| \leq \varepsilon \r) \geq 1 - \delta.
    \end{equation}
\end{lemma}

\begin{proof}
    Use \cref{thm:intermediate_covmat} to get an $(\varepsilon/2)$-estimate $\overline{\bm{\covmat}}$ of the covariance matrix $\covmat$ of $\ket{\fermEPR(Q)}$. This uses $N = \l\lceil \frac{128 n^2 \log(8n/\delta)}{\varepsilon^2} \r\rceil$ copies. By \cref{claim:matrix_rounding_lemma}, if we extract the top-right block of $\overline{\bm{\covmat}}$, round it to a nearby orthogonal matrix (e.g., by taking the SVD), and right-multiply by $S$, then the solution satisfies \cref{eq:choi_Q_guarantee}.
\end{proof}

The bootstrap argument is by now standard.

\begin{theorem}
    There is an efficient algorithm that uses $\O(n^3 \log(n/\delta) / \varepsilon)$ queries to $\act(Q)$ and $n$ ancillary modes to produce $\widehat{\bm{Q}} \in \Orth(2n)$ such that
    \begin{equation}
        \Pr\l( \diamdist(\act(\widehat{\bm{Q}}), \act(Q)) \leq \varepsilon \r) \geq 1 - \delta.
    \end{equation}
    All operations are Gaussian and parity-conserving.
\end{theorem}

\begin{proof}
    Run $\bootstrap(\alg; \frac{\varepsilon}{n}, \delta)$ as in \cref{alg:bootstrap_process}, where $\alg$ is the algorithm described in \cref{lem:learning_from_choi}. By \cref{prop:bootstrap_FLO}, the base cost for constant error $\frac{1}{10}$ and failure probability $\delta_t = \frac{\delta}{2^{T+1-t}}$ is $\O(n^2 \log(n/\delta_t))$ queries, so the entire procedure uses $\O(n^3 \log(n/\delta) / \varepsilon)$ queries. That all operations are parity-conserving Gaussian follows from the fact that the unitary which prepares $\ket{\fermEPR}$ is an $\SO(4n)$ FLO (\cref{claim:fEPR_gauss}).
\end{proof}

\section{Alternate error analysis for $\U(n)$ tomography}\label{sec:rmt}

\begin{proof}[Proof (of \cref{prop:phaseless_alg})]
    \cite{haah2023query} begin by writing the output of each state tomography as
    \begin{equation}
        \ket{\widehat{\bm{u}}_j} = e^{i\bm{\alpha}_j} \sqrt{1 - \bm{\varepsilon}_j} \ket{u} + \sqrt{\bm{\varepsilon}_j} \ket{\bm{w}}
    \end{equation}
    where $\bm{\varepsilon}_j \leq \varepsilon^2$ with probability at least $1 - \frac{\delta}{2n}$ and $\ket{\bm{w}}$ is Haar-random on the subspace orthogonal to $\ket{u_j}$. By \cref{thm:slater_tomo_eq1}, we can guarantee this with $N = \l\lceil \frac{384(11n + 5\log(4n/\delta)}{\varepsilon^2} \r\rceil$ copies of $\pas(U) \ket{1_j}$. They then re-express this as
    \begin{equation}
        \widehat{\bm{U}} - U \bm{A} = U \bm{A} \bm{\Delta} + \bm{W} \bm{E},
    \end{equation}
    where $\bm{W}$ is the random matrix with $\ket{\bm{w}_j}$ as its columns, and we have defined $\bm{A} \coloneqq \diag(e^{i\bm{\alpha}_1}, \ldots, e^{i\bm{\alpha}_n})$, $\bm{\Delta} \coloneqq \diag(\sqrt{1 - \bm{\varepsilon}_1}, \ldots, \sqrt{1 - \bm{\varepsilon}_n}) - \I$, and $\bm{E} \coloneqq \diag(\sqrt{\bm{\varepsilon}_1}, \ldots, \sqrt{\bm{\varepsilon}_n})$. By a union bound, both $\|\bm{\Delta}\|$ and $\|\bm{E}\|$ are at most $\varepsilon$ except with probability $\delta/2$. This implies that
    \begin{equation}\label{eq:bound_with_W}
        \min_{\Theta \in \diag(\R^n)} \|\widehat{\bm{U}} - U e^{i\Theta}\| \leq (1 + \|\bm{W}\|) \varepsilon \quad \text{with probability at least } 1 - \frac{\delta}{2}.
    \end{equation}
    Using techniques from random matrix theory, \cite{haah2023query} argue that the norm of $\bm{W}$ is bounded by some unspecified constant with high constant probability, say $\geq 0.98$. To boost this probability also to $\geq 1 - \frac{\delta}{2}$ they use a standard median-of-means trick, repeating the column tomography process $\O(\log(1/\delta))$ times.
    
    We provide an alternative proof of the statement here which gets the $\delta$ failure probability directly. Observe that the columns of $\bm{W}$ are:
    \begin{enumerate}
        \item Independent,
        
        \item Subgaussian with Orlicz $\psi_2$-norm $\| \ket{\bm{w}_j} \|_{\psi_2} \leq \frac{1}{\sqrt{n-1}}$, and \label{item:psi2_norm}

        \item Isotropic on average:~$\E[\bm{W} \bm{W}^\dagger] = \I$.
    \end{enumerate}
    The first point is by construction. The second is because each $\ket{\bm{w}_j}$ is uniform on the sphere orthogonal to $\ket{u_j}$, hence subgaussian;~we show at the end how to derive the constant in the $\psi_2$-norm. The third follows from the fact that $\E\op{\bm{w}_j}{\bm{w}_j}$ is equal to the normalized projector onto the subspace orthogonal to $\ket{u_j}$:
    \begin{equation}
    \begin{split}
        \E[\bm{W} \bm{W}^\dagger] = \E \sum_{j=1}^n \op{\bm{w}_j}{\bm{w}_j} = \sum_{j=1}^n \frac{\I - \op{u_j}{u_j}}{n - 1} = \I.
    \end{split}
    \end{equation}
    
    Now we use random matrix theory to bound the norm of $\bm{W}$. The columns of $\bm{W}$ are not exactly isotropic, but only isotropic on averge, so we need to use a non-isotropic concentration bound appearing in \cite[Theorem 5.39, Remark 5.40]{vershynin2012introduction}:~for every $t \geq 0$,
    \begin{equation}\label{eq:non-isotropic_bound}
        \Pr\l( \l\| \frac{1}{n} \bm{W} \bm{W}^\dagger - \frac{1}{n} \I \r\| \leq \max\{\gamma, \gamma^2\} \r) \geq 1 - 2\exp\l( -\frac{c_1 t^2}{K^4} \r) \quad \text{where } \gamma = K^2 \sqrt{\frac{\log 9}{c_1}} + \frac{t}{\sqrt{n}},
    \end{equation}
    $K = \max_{j \in [n]} \| \ket{\bm{w}_j} \|_{\psi_2}$, and one can choose $c_1 = \frac{1}{128e^2}$. This implies that with the same probability,
    \begin{equation}
        \|\bm{W}\| \leq \sqrt{1 + n \max\{\gamma, \gamma^2\}} \leq 1 + \frac{1}{2} n \gamma,
    \end{equation}
    assuming that $n$ is sufficiently large enough so that $\gamma \leq 1$. Then, it suffices to set $t = K^2 \sqrt{\frac{\log(4/\delta)}{c_1}}$ to get
    \begin{equation}
        \|\bm{W}\| \leq 1 + \frac{1}{2} n K^2 \l( \sqrt{\frac{\log 9}{c_1}} + \sqrt{\frac{\log(4/\delta)}{c_1 n}} \r)
    \end{equation}
    except with probability $\delta/2$. Assuming that the final failure probability $\delta$ (by a union bound with the event in \cref{eq:bound_with_W}) is no less than $e^{-5n}$,\footnote{The constant $5$ is arbitrary.} we can conclude that
    \begin{equation}
        \min_{\Theta \in \diag(\R^n)} \|\widehat{\bm{U}} - U e^{i\Theta}\| \leq 120 \varepsilon \quad \text{with probability at least } 1 - \delta.
    \end{equation}
    Rescaling $\varepsilon$ implies that the constant $C$ in \cref{line:phaseless_N} of \cref{alg:phaseless_tomo} is no larger than $5.6 \times 10^6$.
    
    It remains to establish the constant in \cref{item:psi2_norm}. It is a standard fact that Haar-random unit vectors $\ket{\bm{w}}$ in $\C^d$ are subgaussian with $\|\ket{\bm{w}}\|_{\psi_2} \lesssim \frac{1}{\sqrt{d}}$~\cite[Theorem 3.4.5]{vershynin2026high};~we make this constant explicit. First, the subgaussian (or Orlicz $\psi_2$-)norm of a random scalar variable $\bm{X}$ is defined as
    \begin{equation}
        \|\bm{X}\|_{\psi_2} \coloneqq \inf\{ b > 0 : \E e^{\bm{X}^2 / b^2} \leq 2 \}.
    \end{equation}
    The generalization to random vectors is then $\|\ket{\bm{w}}\|_{\psi_2} \coloneqq \sup_{\ket{v}} \| \ip{v}{\bm{w}} \|_{\psi_2}$. Denoting $\bm{X} = |\ip{v}{\bm{w}}|$, we expand the exponential:
    \begin{equation}
        \E e^{\bm{X}^2 / b^2} = 1 + \sum_{p=1}^\infty \frac{1}{p!} \frac{\E[\bm{X}^{2p}]}{b^{2p}}.
    \end{equation}
    The moments of $|\ip{v}{\bm{w}}|^{2p}$ are well-known;~for example, using \cite[Theorem 22]{mele2024introduction} gets
    \begin{equation}
        \E[|\ip{v}{\bm{w}}|^{2p}] = \frac{1}{\binom{p + d - 1}{p}}.
    \end{equation}
    Using a computer algebra system, we find that
    \begin{equation}
        \sum_{p=1}^\infty \frac{1}{p! \binom{p + d - 1}{p} b^{2p}} = e^{b^2/2} b^{2(d-1)} ((d-1)! - \Gamma(d, b^{-2})),
    \end{equation}
    where $\Gamma(s, x) = \int_x^\infty t^{s-1} e^{-t} \, dt$ is the incomplete Gamma function;~all we use is that $\Gamma(s, x) \geq 0$ for real arguments. Write $b = \sqrt{\frac{c}{d}}$ for some constant $c > 0$ to be determined. Using Stirling's approximation,
    \begin{equation}
    \begin{split}
        \sum_{p=1}^\infty \frac{1}{p! \binom{p + d - 1}{p} b^{2p}} &\leq \exp\l( \frac{c}{2d} \r) \l( \frac{c}{d} \r{)^{d-1}} \sqrt{2\pi(d-1)} \l( \frac{d-1}{e} \r{)^{d-1}} \exp\l( \frac{1}{12(d-1)} \r)\\
        &\leq \exp\l( \frac{6c + 1}{12(d-1)} \r) \l( \frac{e}{c} \r{)^{-(d-1)}} \sqrt{2\pi(d-1)}.
    \end{split}
    \end{equation}
    We can choose $c = 1$ for simplicity;~for all $d \geq 2$ this bound is at most $\sqrt{2}e^{-41/24} < 0.26$. Hence $\E e^{\bm{X}^2 / b^2} < 1.26 < 2$ for $b = \sqrt{\frac{1}{d}}$, making this is a valid bound on the subgaussian norm. We apply this to $\ket{\bm{w}_j}$ with $d = n - 1$.
\end{proof}

\section{Quasi-particle representation of FLOs}\label{sec:bogoliubov}

\begin{proof}[Proof (of \cref{claim:FLO_bogoliubov})]
    Define the following ``vector-of-operators'' notation:
    \begin{equation}
        \vec{A} \coloneqq \begin{pmatrix}
            \vec{A}_1\\
            \vec{A}_2
        \end{pmatrix} \quad \text{where } \vec{A}_1 \coloneqq \begin{pmatrix}
            a_1\\
            \vdots\\
            a_n
        \end{pmatrix} \text{ and } \vec{A}_2 \coloneqq \begin{pmatrix}
            a_1^\dagger\\
            \vdots\\
            a_n^\dagger
        \end{pmatrix},
    \end{equation}
    \begin{equation}
        \vec{\gamma} \coloneqq \begin{pmatrix}
            \vec{\gamma}_1\\
            \vec{\gamma}_2
        \end{pmatrix} \quad \text{where } \vec{\gamma}_1 \coloneqq \begin{pmatrix}
            \gamma_1\\
            \vdots\\
            \gamma_n
        \end{pmatrix} \text{ and } \vec{\gamma}_2 \coloneqq \begin{pmatrix}
            \gamma_{n+1}\\
            \vdots\\
            \gamma_{2n}
        \end{pmatrix}.
    \end{equation}
    These two are related via
    \begin{equation}
        \vec{A} = \frac{1}{\sqrt{2}} \Omega \vec{\gamma} \quad \text{where } \Omega = \frac{1}{\sqrt{2}} \begin{pmatrix}
            \I & i\I\\
            \I & -i\I
        \end{pmatrix}.
    \end{equation}
    We also use the notation that operators transform elementwise, e.g.,
    \begin{equation}
        \act(Z)^\dagger \vec{\gamma} \act(Z) = \begin{pmatrix}
            \act(Z)^\dagger \gamma_1 \act(Z)\\
            \vdots\\
            \act(Z)^\dagger \gamma_{2n} \act(Z)
        \end{pmatrix} = Z \vec{\gamma}.
    \end{equation}
    Write $Z$ in $n \times n$ blocks:
    \begin{equation}
        Z = \begin{pmatrix}
            Z_{11} & Z_{12}\\
            Z_{21} & Z_{22}
        \end{pmatrix}.
    \end{equation}
    Then we straightforwardly compute:
    \begin{equation}
    \begin{split}
        \act(Z)^\dagger \vec{A} \act(Z) &= \frac{1}{\sqrt{2}} \act(Z)^\dagger (\Omega \vec{\gamma}) \act(Z)\\
        &= \frac{1}{2} \act(Z)^\dagger \begin{pmatrix}
            \vec{\gamma}_1 + i\vec{\gamma}_2\\
            \vec{\gamma}_1 - i\vec{\gamma}_2
        \end{pmatrix} \act(Z)\\
        &= \frac{1}{2} \begin{pmatrix}
            Z_{11} \vec{\gamma}_1 + Z_{12} \vec{\gamma}_2 + iZ_{21} \vec{\gamma}_1 + iZ_{22} \vec{\gamma}_2\\
            Z_{11} \vec{\gamma}_1 + Z_{12} \vec{\gamma}_2 - iZ_{21} \vec{\gamma}_1 - iZ_{22} \vec{\gamma}_2
        \end{pmatrix}\\
        &= \frac{1}{2} \begin{pmatrix}
            (Z_{11} + iZ_{21}) (\vec{A}_1 + \vec{A}_2) -i (Z_{12} + iZ_{22})(\vec{A}_1 - \vec{A}_2)\\
            (Z_{11} - iZ_{21}) (\vec{A}_1 + \vec{A}_2) -i (Z_{12} - iZ_{22})(\vec{A}_1 - \vec{A}_2)
        \end{pmatrix}\\
        &= \frac{1}{2} \begin{pmatrix}
            [Z_{11} + Z_{22} - i(Z_{12} - Z_{21})] \vec{A}_1 + [Z_{11} - Z_{22} + i(Z_{12} + Z_{21})] \vec{A}_2\\
            [Z_{11} - Z_{22} - i(Z_{12} + Z_{21})] \vec{A}_1 + [Z_{11} + Z_{22} + i(Z_{12} - Z_{21})] \vec{A}_2
        \end{pmatrix}\\
        &= \begin{pmatrix}
            \alpha & \beta^*\\
            \beta & \alpha^*
        \end{pmatrix} \vec{A}.
    \end{split}
    \end{equation}
    where we have defined $\alpha \coloneqq \frac{1}{2} [Z_{11} + Z_{22} - i(Z_{12} - Z_{21})]$ and $\beta \coloneqq \frac{1}{2} [Z_{11} - Z_{22} - i(Z_{12} + Z_{21})]$ in the final line. In particular, one checks that
    \begin{equation*}
        \Omega Z \Omega^\dagger = \begin{pmatrix}
            \alpha & \beta^*\\
            \beta & \alpha^*
        \end{pmatrix}. \qedhere
    \end{equation*}
\end{proof}

\end{document}